\documentclass[11pt]{amsart}
\usepackage[centering]{geometry}
\usepackage{cancel}
\usepackage{amsthm}
\allowdisplaybreaks[1]
\usepackage{array}   
\newcolumntype{L}{>{$}l<{$}} 

\usepackage{tensor}
\usepackage{slashed}

\usepackage{wrapfig}
\usepackage[pdftex]{graphicx}
\usepackage{pdfpages}

  \usepackage[utopia,cal=cmcal,greekuppercase=upright]{mathdesign}
\usepackage[utf8]{inputenc}
\usepackage[pdfencoding=auto,psdextra]{hyperref}
\usepackage[capitalise]{cleveref}
\crefname{lem}{Lemma}{Lemmas}
\Crefname{lem}{Lemma}{Lemmas}
\crefname{thm}{Theorem}{Theorems}
\Crefname{thm}{Theorem}{Theorems}
\crefname{prop}{Proposition}{Propositions}
\Crefname{Prop}{Proposition}{Propositions}

\usepackage{esint}
\usepackage{enumerate}
\usepackage{mymathsty}
\usepackage{mymathsybs}
\usepackage{fancyhdr} 
\usepackage{microtype}
\usepackage[all]{xy}
\xyoption{poly}

 \numberwithin{equation}{section}

\DeclareMathSymbol{\T}{\mathbin}{AMSb}{"54}

\begin{document}
 \author{Yang Liu}
 \address{
 SISSA - Scuola Internazionale Superiore di Studi Avanzati \\
Via Bonomea n. 265 \\
34136 Opicina (Trieste)
 }
 \email{yliu@sissa.it}
 \title {
    General Rearrangement Lemma for Heat Trace Asymptotic on
    Noncommutative Tori
}



\date{\today} 
\thanks{
} 

\begin{abstract}

We study a technical problem arising from the spectral geometry of 
noncommutative tori: 
the small time heat trace asymptotic associated to a general second order
elliptic operator.
We extend the rearrangement operators in the conformal case to the general
setting using hypergeometric integrals over Grassmannians. 
The main result is the explicit formula of the second heat coefficient in terms
of the coefficients. When specializing to examples in conformal case, 
we not only recover  results  in  previous works but also obtain some
extra functional relations whose validation provides experimental support to
the main results.
At last, we verify  the relations based on combinatorial
properties derived from the  hypergeometric features.

\end{abstract}

 \maketitle

 \tableofcontents

\section{Introduction}

We provide in this paper an upgrade to the rearrangement lemma
in the computation of the  
small time heat trace expansion on noncommutative tori via 
pseudo-differential calculus.
It allows us to incorporate general second order elliptic operators $P$  of the
form in \cref{eq:P-gen-form-defn} and
 in principle, to obtain closed formulas of the functional densities of the
 heat coefficients (see \cref{eq:V_j-defn}) written in terms of the derivatives of
 the coefficients of $P$. 
 The primary application of the technical question is the spectral geometry on
 noncommutative tori and toric noncommutative manifolds (cf. \cite{MR1846904}, 
 \cite{MR3262521}),
 in which basic notions in Riemannian geometry, such as metric and curvature,
 are investigated in a purely operator-theoretic framework. 
 Following Connes's spectral paradigm,
  the metrics are  implemented as  geometric
 operators $P$, playing the role of  the Laplacian or the squared Dirac operators.
 The coefficients of heat trace expansion of $P$ that we would like to
understand, are the associated  local invariants.
In particular, by analogy with results on Riemannian manifolds,
the second heat coefficient appeared in the main results
encodes the full information of the scalar curvature.
The conformal aspects of program has been carried out in great detail on
noncommutative two tori \cite{MR3540454,MR3194491,MR3148618} and toric
noncommutative manifolds \cite{Liu:2015aa,LIU2017138}. More references can be
found in recent surveys \cite{lesch2018modular,fathizadeh2019curvature}.

In the commutative world, curvature at the infinitesimal level appears as
commutators of covariant derivatives  of the metric connection. Such
noncommutativity  globally influences the shape of the underlying manifold.
The new notion of curvature for noncommutative spaces contains 
an additional noncommutativity arising from the metric itself:
the metric coordinates do not commute with their derivatives. The so-called
rearrangement lemma defines the building blocks of such contribution.
In more detail, on noncommutative tori $C^\infty(\mathbb{T}^m_\theta)$,
metric tensor $g = (g_{ij})$ becomes
the coefficient matrices $\mathbf A = (k_{ij})\in
\GL_m(C^\infty(\mathbb{T}^m_\theta))$  in \cref{defn:conds-for-A} 
 appeared in the leading term of the underlying geometric operator $P$.  
Even though we can assume that the entries $k_{ij}$ mutually commute, there is
no control on the commutativity  when their derivatives are involved.
As a result, except
the Leibniz property, many other basic formulas in calculus require upgrades.
Let us look at a toy example:  for a derivation $\nabla$ and $k \in
C^\infty(\mathbb{T}^m_\theta)$  invertible, 
we can use the Leibniz property to expand:       
\begin{align}
    \nabla^2 k^3 
       &=
       k^2 (\nabla^2 k) + (\nabla k) k^2 + k (\nabla^2  k) k +
    2\sbrac{
    k (\nabla k) (\nabla k) +  (\nabla k) k (\nabla k) + (\nabla k) (\nabla k) k
    }    \nonumber
    \\
       &=
    k^2 (1+ \mathbf y +\mathbf y^2) (\nabla^2 k) + 
    2 k (1+\mathbf y^{(1)} +\mathbf y^{(1)} \mathbf y^{(2)})
    (\nabla k \otimes  \nabla k)
    .
    \label{eq:toyeg-modspecfun}
\end{align}
Compared to
$\nabla^2 k^3 = 2k^2 (\nabla k) + 6k(\nabla k) (\nabla k)$ for $k$ and $\nabla
k$ commute, we see that 
general local differential expressions $L(k,\nabla k, \nabla^2 k, \ldots)$
derived from the operator valued coordinate $k$ 
involve new coefficients:  rearrangement operators 
(see \S\ref{subsec:continuous-funcal}, \S\ref{subsec:smooth-funcal}
for definitions), 
such as 
$1+ \mathbf y +\mathbf y^2$
and $1+\mathbf y^{(1)} +\mathbf y^{(1)} \mathbf y^{(2)}$ appeared in 
\cref{eq:toyeg-modspecfun} above, where $\mathbf y = k^{-1}(\cdot )k$. 
A more interesting example is the Duhamel's formula for the
derivative of the exponential of a self-adjoint
$h \in  C^\infty(\mathbb{T}^m_\theta)$: 
\begin{align*}
     \nabla  (e^{h}) = 
    \int_0^1  e^{(1-s)h} ( \nabla  h) e^{sh} ds  
    = e^h \frac{ \exp( -\ad_h ) -1}{-\ad_h} (\nabla (h)), 
    \, \, \, \, 
    \ad_h = [h,\cdot].
\end{align*}

What is universal behind the rearrangement operators is the spectral functions
like $1+y+y^2$, $1+y_1 + y_1 y_2$ and $(e^x-1) /x$. 
It has been observed by Connes and Moscovici
that the one variable functions appeared in their work \cite{MR3194491}  
show great resemblance to those in topology   generating   characteristic
classes.
In the conclusion of \cite{2016arXiv161109815C}, it is pointed out that the
functions obtained in the paper seems familiar in transcendence theory.
The author added hypergeometric features \cite{Liu:2018aa,Liu:2018ab} 
into the rearrangement lemma that 
concerns what kind of functions shall arise in the pseudo-differential approach
to the heat coefficients.
It turns out that the arguments in \cite{Liu:2018aa,Liu:2018ab} 
can be adapted to handle 
the general situation involving $m \times  m$ operator-valued
coordinates where $m$ denotes the dimension. 
The key feature is that the spectral functions are given by hypergeometric
integrals. For non-experts of special functions, ``hypergeometric'' refers to
the property which sounds interesting from algebraic geometry point of view:
namely, the integrals (\cref{eq:Falpha(A-xi)-scalar})
are constructed out of the combinatorial data
(cf. \cref{eq:omega-alpha-u-defn,eq:B(A-u)-scalar-version})
of standard simplexes $\triangle^n$. 
In fact, for the diagonal case $\mathsf F_\alpha(z)_{s_1, \ldots, s_N}$ 
(\cref{eq:sfFa-defn}) and the conformal case  $\mathsf H_\alpha(z;m;j)$ 
(\cref{eq:sfHalpha}), the functions belong to a general class of 
hypergeometric functions over Grassmannians \cite[\S3]{MR2799182}. 

Back to the technical  question raised at the beginning, 
we briefly outline the algorithm for computing heat coefficients in 
\S\ref{subsec:symbol-cal}-\S\ref{subsec:rearr-lem}. 
The remaining sections 
\S\ref{subsec:gen-form-V2}-\S\ref{subsec:diagonal-case} are devoted to 
the main results: the functional density of the
$V_2$-term (as in \cref{eq:heatop-P-defn}). 
Several versions of the explicit local expressions of $v_2(P)$ are recorded. 
The first one (\cref{thm:v2P-F(A)}) is presented in a compact form which has the
 merit for communicating the formulas. 
To get a precise and detailed understanding of the notations, one should look at 
\cref{thm:v2P-I-components,thm:v2P-II-components}, 
in which rearrangement operators $F_\alpha(\mathbf A)$ are fully expanded into
components. 
Despite the complexity of the formulas, 
the result 
can be directly applied for all potential applications.
At last, we examine, in \S\ref{subsec:diagonal-case},  a situation in which
 the coefficient matrix $\mathbf A$ is diagonal. 
 By restricting to ``eigenvalues'' of the general form,
 the simplified formulas make  
the underlying geometric features  more transparent, which is the next step 
of our exploration of curvature  beyond conformal geometry.

 In the last two sections, we  focus on the Laplacian $\Delta_\varphi$
 representing the simplified model of conformal geometry studied in
 \cite{Liu:2018aa}.
We are able to derive some functional relations 
(cf. \cref{cor:Gdelk-vs-Gdelvfi-fun-relations}) 
by  computing the associated 
$V_2$-term in  two ways based on the general result in \S\ref{subsec:diagonal-case}.
The functions on two sides of the equations are quite different at the first
glance. The complete cancellation with each other yields strong support to the
validation  of our calculation in \S\ref{sec:heat-coef-via-pscal}.
A notable feature of our computation in \S\ref{sec:Verification} is that, 
instead of invoking the lengthy algebraic expressions, we present and manipulate
the spectral functions through a basis of functions consisting of 
the hypergeometric family and  additional  $\mathsf
G_{\mathrm{pow}}^{(1)}$ and $G_{\mathrm{pow}}^{(1,1)}$ 
(see \cref{lem:delk1/2-to-delk}). 
The cancellation can be seen through the two features among the basis
 discussed  in \cite{Liu:2018ab}:
 differential and recursive relations and the action of the variational
 operators (especially the cyclic permutations \cref{eq:tau1-and-2-defn}).


\subsection*{Acknowledgement}

The author would like to thank Alain Connes, Henri Moscovici and Matthias Lesch
for their comments and inspiring conversations.

The paper is written during the author's postdoctoral fellowships at Max
Planck Institute for Mathematics, Bonn and
SISSA, Trieste, as well as short term visits at Leibniz University (Hannover),
Sichuan University (Chengdu), Shanghai Center for Mathematics Sciences and
University of New South Wales (Sydney).  All institutes are greatly appreciated 
for provding  marvelous working environment during the author's stay.

\section{Preliminaries}

\subsection{Noncommutative $m$-tori  $C^\infty(\T^m_\theta) $}

Let $\theta = (\theta_{ij}) \in  M_{n \times n}(\R)$ be a  skew-symmetric
matrix. The smooth
noncommutative $m$-torus $C^\infty(\T^m_\theta) = (C^{\infty}(\T^m),
\times_\theta )$ (viewed as smooth coordinate
functions on $\T^m_\theta$) is identical to $C^{\infty}(\T^m)$ as a topological
vector space with a deformed multination. 
Similar to the existence of Fourier expansion for functions on tori, the
generators of $C^\infty(\T^m_\theta)$ consists of
$m$ unitary elements: $U_s U_s^* = U_s^* U_s =1$,
$s=1, \ldots, m$ and for $\bar l = (l_1, \ldots, l_m) \in  \Z^n$, put  
\begin{align}
    \label{eq:a-bar-l-U-bar-l}
    a_{\bar l} U^{\bar l} \defeq
    a_{l_1, \ldots, l_m} U_1^{l_1} \cdots U_m^{l_m}, 
    \, \, \, a_{\bar l} \in  \mathbb{C},
\end{align}
then
\begin{align*}
    C^\infty(\T^m_\theta) = \set{
    \sum_{\bar l \in  \Z^n} a_{\bar l} U^{\bar l}  \mid 
    \text{ 
        $a_{\bar l} \in  \mathcal S(\Z)$ is a Schwartz function in $\bar l$. 
}
} .
\end{align*}
The deformed multiplication is given in terms of the generators: 
\begin{align*}
    U_s U_l  = e^{2\pi \mathrm i \theta_{l s} } U_l U_s,
    \, \, \, 1 \le  s, l \le m.
\end{align*}
There exist a canonical trace $\varphi_0: C^\infty(\T^m_\theta) \to
\mathbb{C}$ taking the constant term of an element:
\begin{align*}
    \varphi_0\brac{ \sum_{\bar l} a_{\bar l} U^{\bar l}} = a_0.
\end{align*}
We denote by $\mathcal H$ the corresponding GNS representation 
obtained by completing $C^\infty(\T^m_\theta)$ with respect to the inner product
\begin{align}
    \label{eq:GNS-dotprod-defn}
    \abrac{f, \tilde f} \defeq \varphi_0( \tilde f^* f), \, \, \,
    \forall f, \tilde f \in  C^\infty(\T^m_\theta).
\end{align}

The noncommutative $m$-torus $C(\T^m_\theta)$ 
(playing the role of all continuous function on $\T^m_\theta$) is the
$C^*$-algebra  $C(\T^m_\theta) = \overline{C^\infty(\T^m_\theta)}$ 
in which the completion is taken with respect to 
the operator norm 
from the following representation:
$\rho: C^\infty(\mathbb{T}^m_\theta) \to  B(L^2(\mathbb{T}^m))$, 
 with $\T^m \cong \R^m/(2\pi \Z)^{m}$:  
\begin{align}
    \label{eq:repTmtheta-on-generators}
    ( \rho(U_s)f)(x) \defeq e^{i x_s} f(x+\pi \bar \theta_s), \, \, \,
    \forall f \in  L^2(\T^m),
\end{align}
where $\bar \theta_s$ is the $s$-column of $\theta$.


The differential calculus ( moreover pseudo-differential calculus) is built upon 
a $C^*$-dynamical  system $(C(\T^m_\theta), \R^m, \sigma)$. The action $\sigma$
is periodic, namely, a lift of a $\T^m$ action attached to the $\Z^m$ grading.
In other words, $U^{\bar{l}}$ in \cref{eq:a-bar-l-U-bar-l} are the eigenvectors:
\begin{align}
    \label{eq:sigmas-Ualpha}
    \sigma_r( U^{\bar{l}}) = e^{i r \cdot \bar l} U^{\bar{l}},\, \, \,
    r \in  \R^m, \, \, \, \bar l\in \Z^m.
\end{align}
The representation $\rho$  in Eq.  \eqref{eq:repTmtheta-on-generators}
and the translation action of $\R^m$ 
\begin{align*}
    V_r(f)(x) \defeq f(x+r), \, \, \, r,x \in  \R^m.
\end{align*}
form a covariant representation of the $C^*$-dynamical system:
$ \rho (\sigma_r(T)) = V_{-r} \rho(T) V_r$. 
The smooth noncommutative
torus $C^\infty(\T^m_\theta)$ consists of exactly those smooth elements in the
$C^*$-dynamical  system, that is, all $a \in  C(T^m_\theta)$ such that $r \to
\sigma_r (a)$ is a smooth function in $r \in  \R^m$.  

If we identify $\R^m \cong \mathrm{Lie}(\R^m)$, the Lie derivatives of 
the standard basis
give rise to  the basic derivations $\{ \delta_l\}_{l=1}^m$ acting as
$-i\partial_{x_l}$  on $C^\infty(\T^m)$ regarding to the coordinate 
$(x_1, \ldots, x_m)$. More precisely,  we have
on the generators:
\begin{align}
    \label{eq:basic-derivs-defn}
    \delta_l(U_j) = \mathbf 1_{ij} U_j, \, \, \,
    1\le  i,j \le  m.
\end{align}

The basic derivations  are the generators of the algebra of differential operators. 
In particular, a second
order differential operator is always of the form:
\begin{align}
    \label{eq:P-gen-form-defn}
    P = \sum_{1\le s,l \le m} k_{sl} \delta_s \delta_l + \sum_{s=1}^m r_s
    \delta_s + p_0,
\end{align}
where the coefficients $ k_{sl}, r_s, p_0 \in  C^\infty(\T^m_\theta)$.
As in the commutative case, the ellipticity concerns only the coefficients of
the leading term:
\begin{defn}
    \label{defn:conds-for-A}
 The differential operator 
 $P: C^\infty(\T^m_\theta) \to  C^\infty(\T^m_\theta)$ given in  
 \cref{eq:P-gen-form-defn} is called elliptic if
the coefficient matrix $\mathbf A = (k_{ij})_{1\le i,j \le m}$ admit
a self-adjoint log.
Precisely, we require that $\mathbf A = \exp( \log \mathbf A)$ where the matrix 
$\log \mathbf A = (h_{ij})$ is symmetric with self-adjoint and mutually commute
entries:   
\begin{align*}
  h_{ij} = h_{ji} = h_{ij}^* = h_{ji}^*,\, \, \,
  [h_{ij}, h_{st}] =0, \, \, \,
  1\le  i,j,s,t \le m.
\end{align*}
As a consequence, $\mathbf A$ is positive invertible with mutually commute entries.   
\end{defn}

\subsection{Spectral geometry}
If $P$ fulfills ellipticity above,  
 Connes's pseudo-differential calculus implies that  
there exists a order one pseudo-differential operator $Q$ such
that  $P - Q^* Q$ is of order one. Arguments in Gilkey's book
\cite{gilkey1995invariance} can be applied to
show that $P$   has discrete spectrum  contained  in a conic region of
$\mathbb{C}$. In our examples, the spectrum of $P$ is contained  in
$[c,\infty)$ for some $c \in  \R$.  
The heat operator can be defined via holomorphic functional calculus:
\begin{align}
    \label{eq:heatop-P-defn}
    e^{-tP} = \frac{1}{2 \pi i}
    \int_{\mathcal C} e^{-t\lambda} (P -\lambda)^{-1} d\lambda
    ,
\end{align}
where $\mathcal C$ is a suitable contour winding around the spectrum of $P$. 

Our primary interest is the spectral geometry of $C^\infty(\mathbb{T}^m_\theta)$ 
in which $P$ plays the role of a geometric differential operator.
In the conformal case on $C^\infty(\mathbb{T}^2_\theta)$
\cite{MR3540454,MR3194491}, $P$ is modeled on the Dolbeault Laplacian while on
toric noncommutative manifolds $C^\infty(M_\theta)$ \cite{LIU2017138},
$P$ comes from the squared Dirac operator.  
On $\mathbb{T}^m_\theta$, 
Even for noncommutative tori, the notion of general metrics and the associated 
construction of $P$ is still widely open. 
A recent  proposal  of  constructing Laplace-Beltrami operators 
can be found in \cite{ha2019laplace}.
Among the mentioned examples, a common feature inherited from Riemann geometry
is the property 
that the metric tensor $g = (g_{ij})$ is implemented as
the coefficient matrix $\mathbf A$ of the leading term while 
for lower order terms, $r_s$ and $p_0$
consist of the first and second derivatives of entries of $\mathbf A$ respectively. 

Once the metric $P$ is chosen,  the corresponding local invariants can be
extracted from the small time asymptotic of the heat trace functional 
$a \to  \Tr(a e^{-tP})$: 
\begin{align}
    \label{eq:e^tP-asym-defn}
    \Tr(a e^{-tP}) \backsim_{t \searrow 0}
    \sum_{j=0}^\infty
    t^{\frac{j-m}{2}}
    V_j(a, P), \, \, \, a \in C^\infty(\T^m_\theta),
\end{align}
where $\Tr$ is the operator trace with respect to the Hilbert space $\mathcal
H$ defined in \cref{eq:GNS-dotprod-defn}. 
Each coefficient is absolutely continuous with respect to the canonical trace
$\varphi_0$ with Radon-Nikodym derivative $v_j(P) \in C^\infty(\T^m_\theta)$
(also referred as functional densities in the paper):  
\begin{align}
    \label{eq:V_j-defn}
    V_j(a, P) = \varphi_0 \brac{ a v_j(P) }, \, \, \,
    \forall  a \in  C^\infty(\T^m_\theta).
\end{align}

If $P$ is the scalar Laplacian $\Delta$ on a closed Riemann manifold $(M,g)$, 
the invariants $v_j(\Delta) \in  C^\infty(M)$ are known as
Minakshisundaram-Pleijel coefficients which can written, in principle, as
polynomial functions  in the derivatives of metric tensor.
In particular, the first non-trivial one is  $v_2(\Delta) = \mathcal S_g /6$
proportional to  the scalar curvature function. It gives a geometric interpretation 
to  our main results in \S\ref{subsec:gen-form-V2}
as a model of the scalar curvature from the spectral geometry perspective. 
The spectral paradigm has also been implemented in great detail, known as
spectral action principle, in the noncommutative geometry approach to standard
model, see \cite[\S11]{connes2008noncommutative} for further references.

The pseudo-differential calculus is able to, not only establish the existence
of the expansion, but also provide a efficient algorithm for the computation of
$v_j(P)$. We shall see in later sections that they can be written as finite sums of
 differential expressions in the coefficients of $P$: $k_{sl}$, $r_s$ and $p_0$.   
 Nevertheless, the length of of $v_j(P)$  grows substantially as $j$ going up,
 cf. \cite{2016arXiv161109815C} for an impression of the complexity of
 $V_4$-term even in the conformal case.
 Sum up, a challenging task in the spectral geometry is to explore
 universal structures behind the intimidating differential expressions, 
so that further applications, such as related variational problems, can be
carried out.  

 \section{Hypergeometric Functions in the Rearrangement Lemma}
 \label{sec:hyperfun-in-Rlemma}


\subsection{Continuous functional calculus}
\label{subsec:continuous-funcal}

Let $A$ be a unital commutative $C^*$-algebra and $\mathcal M_A$ be the spectrum
(the space of maximal ideals).
The Gelfand-Naimark theorem asserts that $A$ is isomorphic to the
$C^*$-algebras of continuous functions on $\mathcal M_A$. We denote  the
$*$-isomorphism (the inverse of the Gelfand map), by:
\begin{align}
    \label{eq:Psi-on-CM_A}
   \Psi: C(\mathcal M_A) \to  A.
\end{align}

Let $\bar a =(a_1, \ldots, a_J)$ be a tuple
of mutually commutative self-adjoint elements in $A$. They generate a commutative
unital $C^*$-algebra $C(1,\bar a) \subset A$  whose spectrum
$\mathcal M $ can be identified with a compact subset in $\R^J$: 
\begin{align}
    \label{eq:mathcal-M-in-R^J}
  \mathcal M \subset \prod_{j=1}^n X_{a_j}\subset \R^{J}, \, \, \, \, 
  \mathfrak m \in  \mathcal M \to 
  \brac{ 
  \Psi_{\bar a}^{-1}(a_1)(\mathfrak m), \ldots,
  \Psi_{\bar a}^{-1}(a_J)(\mathfrak m)},
\end{align}
where $\Psi^{-1}_{\bar a}:C(1,\bar a) \to  C(M)$ is the Gelfand map and 
the evaluation $\Psi_{\bar a}^{-1}(a_j)(\mathfrak m)$ belongs to  the
spectrum $X_{a_j}$ of $a_j$, $1\le  j \le J$. 
For any $f(\bar x) \in  C(\mathcal M)$, we have several notations for  the
functional calculus:  
\begin{align*}
    f(a_1, \ldots, a_J) \defeq  
    \int_{\mathcal M} f(\bar x) dE^{\bar a}(\bar x) \defeq \Psi_{\bar a}(f)
    \in  C(1, \bar a),
\end{align*}
where $dE^{\bar a}$  and the map $\Psi_{\bar a}$ are both referred to as the
spectral measure when no confusion arises. 
In the case of $J=1$ with a normal element, that is, $\bar a = (a_1)$ with 
$[a_1,a_1^*]=0$, the $C^*$-algebra  $C(1,a_1)$ is well-defined and the space of
maximal ideals is identical to the spectrum: $M \cong X_{a_1} \subset
\mathbb{C}$.    

\subsection{Smooth functional calculus}
\label{subsec:smooth-funcal}
Let $A$ be a unital $C^*$-algebra as before.
We shall briefly review the construction in \cite[\S3]{leschdivideddifference}.
Consider the (algebraic) contraction map
$\cdot: A^{\otimes n+1} \times A^{\otimes  n} \to  A$,
on elementary tensors, it reads: 
\begin{align}
    \label{eq:contraction-dot-defn}
    (a_0 \otimes  \cdots \otimes  a_n) \cdot 
    (\rho_1 \otimes \cdots  \otimes  \rho_n) 
    = a_0 \rho_1 a_1 \cdots  \rho_n a_n.
\end{align}
It makes elements of $A^{\otimes n+1}$ into linear operators from
$A^{\otimes n}$ to $A$. We denote the induced map by 
\begin{align}
    \label{eq:iota-defn-induced-contraction}
    \iota: A^{\otimes  n+1} \to  L(A^{\otimes  n}, A).
\end{align}
For any $a \in  A$, for $0\le  j \le n$, depending on the context, we denote by
$a^{(j)}$ either the elementary tensor 
\begin{align}
    \label{eq:a(j)-as-tensors-no-iota}
    a^{(j)} \defeq 1\otimes \cdots  \otimes  a \otimes  \cdots  \otimes  1
    , \, \, 
    \text{$a$ occurs at the $j$-the factor,}
\end{align}
or the operator
\begin{align}
    \label{eq:a(j)-as-ops-iota}
    a^{(j)} \defeq 
    \iota(1\otimes \cdots  \otimes  a \otimes  \cdots  \otimes  1)
    \in  L(A^{\otimes n}, A).
\end{align}
The left  and right multiplications correspond to $a^{(0)},a^{(1)} \in L(A,A)$.
Generally, the superscript $(j)$ simply indicates that, the multiplication
occurs at the $j$-slot of elementary tensors in $A^{\otimes n}$.
For self-adjoint $a = a^* \in  A$
and $k = e^a$,  we put $\mathbf x_a = -\mathrm{ad}_a=[\cdot,a]$ and 
$\mathbf y_a =\mathrm{Ad}_{k^{-1}}= k^{-1}(\cdot )k$, 
then the associated lifted operators in $L(A^{\otimes n},A)$ are given by:
\begin{align}
    \label{eq:mathbf-x=a-and-mathbf-y=k}
    \mathbf x_a^{(j)} = -a^{(j-1)} + a^{(j)},\, \, 
    \mathbf y_a^{(j)} = k^{(j-1)} k^{(j)} = e^{\mathbf x^{(j)}}, \, \, 
    1\le  j \le  n,
\end{align}
in which the superscript $(j)$ indicates the commutator or conjugation operator
acts only  on the $j$-th factor on   elementary tensors.
We shall make use of the inverse of the relations in
\cref{eq:mathbf-x=a-and-mathbf-y=k} in later sections:
\begin{align}
    \label{eq:a=mathbf-x-and-k=mathbf-y}
    a^{(j)} = -a^{(0)} + \mathbf x_a^{(1)}+ \cdots  +\mathbf x_a^{(j)}, 
    \, \, \, \, 
    k^{(j)} = (k^{(0)})^{-1} \mathbf y_a^{(1)} \cdots \mathbf y_a^{(j)}
    ,\, \, \, \, 
    1\le  j \le  n,
\end{align}

The functional calculus requires different completions of the algebraic tensor
products of $A$. 
The  projective tensor product $A^{\otimes_\gamma n}$ is the
norm completion with respect to
\begin{align*}
    \norm{a}_\gamma \defeq 
    \inf \sum_s \norm{\alpha_0^{(s)}} \cdots  \norm{\alpha_n^{(s)}},
\end{align*}
where the infimum runs over all possible decomposition of $a$  as elementary
tensors $a = \sum_s \alpha_0^{(s)} \otimes  \cdots  \otimes  \alpha_n^{(s)}$. 
The signature property of the   projective tensor product 
is that the multiplication map: $m: A^{\otimes n}\to  A$
extends continuously. 
In particular, the algebraic map defined in 
\cref{eq:iota-defn-induced-contraction}   induces a continuous map
\begin{align}
    \label{eq:iota-cont-defn}
    \iota: A^{\otimes_\gamma n+1} \to  
   L_{\mathrm{cont}}(A^{\otimes_\gamma n}, A) .
\end{align}


Now let us consider a tuple of mutually commuting self-adjoint elements $\bar
a =(a_1, \ldots,
a_J)$ with spectra $X_{a_j}$, $1\le  j \le  J$. As in \cref{eq:mathcal-M-in-R^J},
the space of maximal ideals $\mathcal M_{\bar a} $ of $C(1, \bar a)$ is
a subset of $\prod_{j=1}^J X_{a_j} \subset \R^J$. Now let $U \subset \R^J$ be
an open subset containing $\mathcal M_{\bar a}$. Followed by the restriction map
$ C^\infty(U) \to  C(\mathcal M_{\bar a})$, the continuous functional calculus
$\Psi_{\bar a}$  in \cref{eq:Psi-on-CM_A} leads to a map, 
\begin{align*}
    \Psi_{U}
    : C^\infty(U) \to  C(1, \bar a), \, \, \, \, 
  f \in  C^\infty(U) 
  \to   f(a_1, \ldots, a_J) \defeq  \Psi_{\bar a} (f|_{\mathcal M_{\bar a}}). 
\end{align*}

The smooth functional calculus relies on the nuclearity of the Fr\'echet
topology of $C^\infty(U)$, which states that the projective $\otimes_\gamma$ and
injective $\otimes_\epsilon$ tensor product agree and they are both isomorphic
to the smooth functions on the Cartesian product:
\begin{align*}
    C^\infty(U)^{\otimes_\gamma n+1} \cong C^\infty(U^{n+1}) 
    \cong C^\infty(U)^{\otimes_\epsilon n+1}. 
\end{align*}
The injective side allows us to approximate multivariable functions 
$f(x_0, \ldots, x_n)$ by those of the form of separating variables. More
precisely, the algebraic map $ C^\infty(U)^{\otimes n+1} \to  C^\infty(U^{n+1})$: 
\begin{align*}
  f_0 \otimes  \cdots  \otimes  f_n \to  f, \, \, \, \, 
  \text{with $f(x_0, \ldots, x_n) \defeq f_0(x_0)\cdots  f_n(x_n)$}
\end{align*}
 extends by continuity to an
isomorphism of $ C^\infty(U)^{\otimes_\epsilon n+1} \to C^{\infty}(U^{n+1})$.  
The projective feature implies that, after the projective completion, 
the algebraic map
\begin{align*}
    \Psi_{U}^{\otimes n+1}: C^\infty(U)^{\otimes  n+1}\to  A^{\otimes n+1}:
    f_0\otimes  \cdots  \otimes f_n \to 
    f_0(\bar a) \otimes \cdots  \otimes  f_n(\bar a)
\end{align*}
induces a continuous map
\begin{align}
    \label{eq:Psi-gamma-defn}
    \Psi_\gamma: C^\infty(U^{n+1}) \to  A^{\otimes_\gamma n+1} .
\end{align}
As in \cref{eq:a(j)-as-ops-iota}, we denote 
 $\bar a^{(j)} = (a_1^{(j)}, \ldots, a_J^{(j)})$ and  
$\mathbf a = (a_{i}^{(j)})_{J \times  n}$, with $1\le i\le J$ and $0\le j\le n$. 
The operators $a_i^{(j)} \in  L(A^{\otimes n},A)$ defined in 
\cref{eq:a(j)-as-ops-iota} are the images of the
coordinate functions under the functional calculus:
\begin{align}
    \label{eq:uij-sent-aij}
(\iota \circ \Psi_\gamma) ( u_i^{(j)}) = a_i^{(j)}.
\end{align}
where $\bar u = (\bar u^{(1)}, \ldots, \bar u^{(n)}) \in U^{n+1} \subset \R^{J
\times n}$, with $\bar u^{(j)} = (u_1^{(j)}, \ldots, u_J^{(j)})$.

\begin{defn}[Smooth functional calculus]
    \label{defn:smooth-funcal}
   Keep the notations as above.
  For any $f \in  C^\infty(U^{n+1}) $, 
  we define the smooth functional calculus in the following way:
  \begin{align*}
      f(\mathbf a) = 
      f(\bar a^{(0)}, \ldots, \bar a^{(n)}) \defeq (\iota \circ \Psi_\gamma)(f) \in 
      L_{\mathrm{cont}}(A^{\otimes_\gamma n}, A),
  \end{align*}
  where $\iota$ and $\Psi_\gamma$ are defined in
  \cref{eq:iota-cont-defn,eq:Psi-gamma-defn} respectively. 
\end{defn}

\begin{prop}
    \label{prop:Fubini-type-result}
    Consider functions given via  integral representations:
    \begin{align*}
    f(\bar u) = 
    \int_{\mathcal B} F(p,\bar u) d\mu_p 
    \end{align*}
where  $(\mathcal B, \mu)$ is a Borel space and  
$F(p,u): \mathcal B \times  U^{n+1} \to  \mathbb{C}$ is continuous in $p$ and
smooth in $\bar u$. With the integrability condition: for any compact set $K
\subset U^{n+1}$ and multiindex $\alpha$, 
\begin{align}
    \label{eq:integrability-condition}
    \int_{\mathcal B}  \sup_{ \bar u \in K}
    \abs{\partial^\alpha_{\bar u} F(p,\bar u)} d\mu_p <\infty ,
\end{align}
we have the Fubini type result
\begin{align}
    \label{eq:fubibi-type-result}
    f(\bar a^{(0)}, \ldots, \bar a^{(n)}) \defeq \Psi_\gamma 
    \brac{ \int_{\mathcal B} F(p, \cdot )d\mu_p} =
    \int_{\mathcal B} 
    \Psi_\gamma
    \brac{F(p, \cdot )}d\mu_p,
\end{align}
where the last integral is a Bochner one taking values in $A^{\otimes_\gamma n+1}$.  
\end{prop}
\begin{proof}
    The integrability for all derivatives shown in
    \cref{eq:integrability-condition} implies the converges of the integral 
    $\int_{\mathcal B} F(p,\cdot )d\mu_p$ with respect to the Fr\'echet
    topology of $C^\infty(U^{n+1})$, so that $f(\bar u)$ is smooth in $\bar u$.
    We refer to \cite[Theorem 3.4]{leschdivideddifference}  for more details.
\end{proof}

The integral form mentioned in the proposition above leads to a  more
elementary construction of the functional calculus making use of  Fourier
transform. For any $f \in  C^\infty(U^{n+1})$, taking any extension $ \tilde
f \in  \mathscr S(\R^{J \times  (n+1)})$ to a Schwartz function so that it can
be written as a Fourier transform, 
with $\bar{\mathsf u} = ( u_l^{(j)})$ and $\xi =(\xi_l^{(j)})$ , $1\le  l \le
J$ and $0\le  j\le  n$:   
\begin{align*}
    \tilde f(\bar{\mathsf u}) = 
    \int_{\R^{J \otimes (n+1)}}
    ( \tilde f)^{\vee}  (\xi) \exp\brac{
    \sum_{l,j}
    i \xi^{(j)}_i u^{(j)}_i} d\xi,
\end{align*}
where $( \tilde f)^{\vee}$ denotes the normalized Fourier transform. 
Now we are ready to define   the Schwartz functional calculus 
    \begin{align*}
        \Psi_{\mathscr S}:\mathscr S(\R^{n}) \to  L(A^{\otimes  n},A):
        f \to  \Psi_{\mathscr S}(f) \defeq
        f_{\mathscr S}   ( a^{(0)}, \ldots, a^{(n)}).
    \end{align*}
    by substituting $u_i^{(j)} \to  a_i^{(j)}$ into the integral form above.
    More precisely, we have, on elementary tensors:
\begin{align}
    \label{eq:schwartz-fcal-in-a}
    &\, \,  f_{\mathscr S}  
     (\bar a) ( \rho_1 \otimes \cdots  \otimes  \rho_n) \\
    =&\, \, 
    \int_{\R^{J \otimes (n+1)}}
    ( \tilde f)^{\vee}  (\xi) \exp\brac{
    \sum_{l,j}
    i \xi^{(j)}_i u^{(j)}_i}
    \brac{
     \rho_1 \otimes \cdots  \otimes  \rho_n
    }
    d\xi,
   \nonumber \\
     =&\, \,  
     \int_{\R^{J \otimes  (n+1)}} ( \tilde f)^{\vee} (\xi)
     \brac{
         e^{i \sum_{l=1}^J\xi_{l}^{(0)} a_l} \rho_1
         e^{i \sum_{l=1}^J\xi_{l}^{(1)} a_l} \cdots \rho_n
e^{i \sum_{l=1}^J\xi_{l}^{(n)} a_l}
     } d\xi.
     \nonumber
\end{align}


Finally, following the substitution in \cref{eq:mathbf-x=a-and-mathbf-y=k}, 
we obtain the corresponding functional calculus for
the  commutator and conjugation operators given in 
\cref{eq:mathbf-x=a-and-mathbf-y=k}: 
$\{\mathbf x_{a_i}^{(j')}, \mathbf y_{a_i}^{(j')}\} \subset L(A^{\otimes  n},A)$, 
with $1\le i\le J$ and $ 1\le j'\le n$.
In more detail, let $\bar u = (u_i^{(j)})_{J \times
(n+1)}$ be the coordinate function  in \cref{eq:uij-sent-aij}, 
we denote matrices, viewed as maps:
\begin{align*}
    \bar{\mathsf x} 
    \defeq \bar{\mathsf x}(\bar{\mathsf u})
  = (x_i^{(j')})_{J \times  n}, \, \, \, \,   
\bar{\mathsf y} 
    \defeq \bar{\mathsf x}(\bar{\mathsf u})
= (y_i^{(j')})_{J \times  n} :
U^{n+1} \subset \R^{J \times (n+1)} \to \R^{{J \otimes  n}}
,
\end{align*}
in which the entries come from the change of variables in
\cref{eq:mathbf-x=a-and-mathbf-y=k}: 
\begin{align*}
    x_i^{(j')} \defeq  x_i^{(j')} (\bar{\mathsf u}) 
    = -u_i^{(j'-1)} + u_i^{(j')}, \, \, \, \, 
    y_i^{(j')} \defeq  y_i^{(j')} (\bar{\mathsf u} )
    = e^{-u_i^{(j'-1)}} e^{u_i^{(j')}}, \, \, \, \, 
1\le i\le J, \, \, 1\le j'\le n.
\end{align*}
Let $\bar{\mathbf x} = (\mathbf x_{a_i}^{(j')})$ and 
$\bar{\mathbf y} = (\mathbf y_{a_i}^{(j')})$, we define  
\begin{align}
    \label{eq:schwartz-fcal-in-x-and-y}
    f_{\mathscr S}
    \brac{(\bar{\mathbf x})} 
    \defeq (f\circ \bar{\mathsf x})_{\mathscr S}\brac{
    \mathbf a
    }, \, \, 
f_{\mathscr S}
    \brac{(\bar{\mathbf y})} 
    \defeq (f\circ \bar{\mathsf y})_{\mathscr S}\brac{
    \mathbf a
    }, \, \, 
\end{align}
whenever the right hand sides make sense as Schwartz functional calculus
(cf. \cref{eq:schwartz-fcal-in-a}) in
$\mathbf a = (\bar a^{(0)}, \ldots, \bar a^{(n)})$.
\subsection{Hypergeometric integrals for the rearrangement lemma}
\label{subsec:gen-case}
%

We now  describe the spectral functions required in
the rearrangement lemma (\cref{prop:rearg-lem}).
For a multiindex 
$\alpha = (\alpha_0, \alpha_1, \ldots, \alpha_n) \in  \Z_{>0}^{n+1}$ and a point
$u = (u_1, \ldots, u_n) \in \triangle^n$ in the standard $n$-simplex: 
\begin{align*}
    \triangle^n  = \set{\sum_1^n u_j \le  1, u_1\ge 0, \ldots , u_n\ge 0}
\subset \R^n,  
\end{align*}
  we set:
\begin{align}
    \label{eq:omega-alpha-u-defn}
    \omega_\alpha(u) =
    \brac{\prod_1^n \Gamma(\alpha_l)}^{-1}
        \brac{ 1- \sum_1^n u_l }^{\alpha_0-1} 
        \left( \prod_1^n u_l^{\alpha_l-1}\right)
    \end{align}
where $\Gamma(z)$ is the standard Gamma function. 
Observe that $\omega_\alpha(u)$ couples all the boundary hyperplanes of
the standard $n$-simplex $\triangle^n$ with the index $\alpha$ in
a multiplicative fashion. 
Similarly, for a tuple $\bar A = (A_0, \ldots, A_n)$ of 
positive invertible $m\times m$  matrices, we denote   
\begin{align}
    \label{eq:B(A-u)-scalar-version}
    B_n(\bar A,u) = A_0 (1-\sum_1^n u_l) + \sum_1^n A_l u_l = 
    A_0 (1-\sum_1^n Z_l u_l) \in M_{m \times  m}(\R),
\end{align}
where, in the last equal sign, we have substituted, 
$A_l = A_0 Y_1 \cdots Y_l$ with $Y_l = A_{l-1}^{-1} A_l$, so that 
\begin{align*}
    Z_l = 1- A_0^{-1} A_l = 1- Y_1 \cdots  Y_l ,\, \, \, \,  
    1\le l\le n.
\end{align*}
By assembling the notations together with $\xi = (\xi_1, \ldots, \xi_m) \in
\R^m$, we introduce a family of hypergeometric integrals as below: 
\begin{align}
    \label{eq:Falpha(A-xi)-scalar}
    F_\alpha (\bar A,\xi)
    & = 
    \int_{\triangle^n} \omega_\alpha(u) 
    \frac{ (2\pi)^{m/ 2} }{\sqrt{\det B_n(\bar A,u)} }
    \exp \brac{
            \frac{1}{4} \sum_{1\le  i,j \le n}
            ( B_n(\bar A,u))_{ij}  \xi_i \xi_j
} 
    du . 
\end{align}

Let us apply the functional calculus in \S\ref{subsec:smooth-funcal}
to the hypergeometric family given above.  
The $J$-tuple (with $J=m^2$)  of commuting elements  comes from the entries of
the coefficient matrix
$\mathbf A = (k_{ij}) \in \op{GL}_m(C^\infty(\T^m_\theta))$ 
of the leading term of the differential operator $P$ in Eq.
\eqref{eq:P-gen-form-defn}.
It gives rise to 
$\mathbf{\bar A} = (\mathbf A^{(1)}, \ldots, \mathbf A^{(n)})$ with 
$\mathbf A^{(l)} = (k_{ij}^{(l)})$, $l=0,\ldots,n$, where 
$k_{ij}^{(l)} \in  L(C^\infty(\mathbb{T}^m_\theta)^{\otimes  n},
C^\infty(\mathbb{T}^m_\theta))$ as in \cref{eq:a(j)-as-ops-iota}.
Furthermore, matrix in \cref{eq:B(A-u)-scalar-version} becomes
\begin{align}
    \label{eq:Bu_n-defn}
    \mathbf B_n(u) \defeq
    \mathbf B_n (u, \mathbf{\bar A}) & =   \mathbf A^{(0)}( 1- \sum_1^n u_l) +  
    \sum_1^n \mathbf A^{(l)} u_l  = 
        \mathbf A^{(0)}(1- \sum_1^n \mathbf Z_l u_l ) 
    \end{align}
 with $\mathbf A^{(l)} = \mathbf A^{(0)} \mathbf Y^{(1)} \cdots \mathbf Y^{(l)}$ 
 and $\mathbf Y^{(l)} = (\mathbf y_{ij}^{(l)})
 = (\mathbf A^{(l-1)})^{-1} \mathbf A^{(l)}$  so that:
\begin{align}
    \label{eq:bfZ-defn}
    \mathbf Z^{(l)} = (\mathbf z_{ij}^{(l)}) =
    1- (\mathbf A^{(0)})^{-1} \mathbf A^{(l)} =
    1- \mathbf Y^{(1)} \cdots \mathbf Y^{(l)}, 
    \, \, \,
    l=1,\ldots,n.
\end{align}

Finally, we are ready to describe the generalization of the rearrangement operators
appeared in the conformal setting. They are formal differential operators
acting on the algebra of polynomial symbols $C^{\infty}(\T^m_\theta)[\xi]$. 
By setting $\bar A \to  \mathbf{\bar A}$ and $\xi \to  \partial_\xi$ in
\cref{eq:Falpha(A-xi)-scalar}, we obtain
\begin{align*}
  F_\alpha (\mathbf A) \defeq 
  F_\alpha (\mathbf{\bar A},\partial_\xi) 
  : ( C^{\infty}(\T^m_\theta)[\xi])^{\otimes n} \to
 C^{\infty}(\T^m_\theta)[\xi].  
\end{align*}
In fact, we need
$F_\alpha (\mathbf A) |_{\xi=0}:
 ( C^{\infty}(\T^m_\theta)[\xi])^{\otimes n}\to  C^\infty(\T^m_\theta) $
 with the evaluation map
 $|_{\xi=0}: C^\infty(\mathbb{T}^m_\theta)[\xi] \to  C^\infty(\T^m_\theta)$:
\begin{align}
    \label{eq:Falpha(A)-defn}
    F_\alpha (\mathbf A)
    & = 
    \int_{\triangle^n} \omega_\alpha(u) 
    \frac{ (2\pi)^{m/ 2} }{\sqrt{\det \mathbf B_n(u)} }
    \exp \brac{
            \frac{1}{4} \sum_{1\le  i,j \le n}
            ( \mathbf B^{-1}_n (u))_{ij} \partial_{\xi_{i} } \partial_{\xi_j}
} 
    du . 
\end{align}
Since the domain involves  only polynomial symbols (in $\xi$), 
the exponential of differential operators make sense as the
power series:
\begin{align*}
    \exp
    \brac{
   \sum_{ij} B^{-1}_{ij} \partial_{\xi_i}  \partial_{\xi_j}
    }
&=
    \sum_{N=1}^\infty \frac{1}{N!}
\brac{
   \sum_{ij} B^{-1}_{ij} \partial_{\xi_i}  \partial_{\xi_j}
    }^N \\
    &=  
    \sum_{N=1}^\infty \frac{1}{N!}
    \sum_{l_1, \ldots, l_{2N}}
    \brac{
    B^{-1}_{l_1 l_2 } \partial_{\xi_{l_1}} \partial_{\xi_{l_2}}
    }
    \cdots 
    \brac{
    B^{-1}_{l_{2N-1} l_{2N} } \partial_{\xi_{l_{2N-1}}} \partial_{\xi_{l_{2N}}}
    },
\end{align*}
where $i,j$ and $l$'s are summed over $1$ to $m$.
In practice, we often have to expand $F_\alpha (\mathbf A)$ into components: 
\begin{align}
    \label{eq:FaplhaA-expands-to-components}
    F_\alpha (\mathbf A)
    =&
    \sum_{N=1}^\infty 
    \sum_{l_1, \ldots, l_{2N}}
    F_\alpha (\mathbf A)_{(l_1l_2)\cdots (l_{2N-1}l_{2N})}
\partial_{\xi_{l_1}}  \cdots \partial_{\xi_{l_{2N}}} 
\end{align}
with  $F_\alpha (\mathbf A)_{(l_1l_2)\cdots (l_{2N-1}l_{2N})}:
C^\infty(\T^m_\theta)^{\otimes  n} \to  C^\infty(\T^m_\theta)$:
\begin{align}
    \label{eq:FaplhaA-components}
    \begin{split}
        F_\alpha (\mathbf A)_{(l_1l_2)\cdots (l_{2N-1}l_{2N})}
     & = 
  \frac{1}{4^N N!}  \int_{\triangle^n} \omega_\alpha(u) 
    \frac{ (2\pi)^{m/ 2} }{\sqrt{\det \mathbf B_n(u)} }
\mathbf B^{-1}_n (u)_{l_1 l_2} \cdots 
\mathbf B^{-1}_n (u)_{ l_{2N-1} l_{2N}}
du, \\
F_\alpha (\mathbf A)_{\emptyset} 
     & = 
  \frac{1}{4^N N!}  \int_{\triangle^n} \omega_\alpha(u) 
    \frac{ (2\pi)^{m/ 2} }{\sqrt{\det \mathbf B_n(u)} } du
    ,\, \, \, \text{for $N=0$.}
    \end{split}
\end{align}
Since $\mathbf B \defeq \mathbf B_n(u)$ is symmetric and has mutually commuting
entries, the inverse can be computed from the adjugate matrix\footnote{
  The transpose of the cofactor matrix.}
$\mathbf E_{\mathbf B}$  of $\mathbf B$, which is also symmetric:
\begin{align*}
    (\mathbf B)^{-1} = (\det \mathbf B)^{-1} \mathbf E_{\mathbf B}
    ,\, \, \,
    (\mathbf E_{\mathbf B})_{ij} = 
    (\mathbf E_{\mathbf B})^T_{ij} =
\partial_{\mathbf B_{ij}} (\det \mathbf B).
\end{align*}
As a result, we can replace the inverse in   Eq. \eqref{eq:FaplhaA-components} 
by derivatives of the determinant:
\begin{align*}
     &\, \,    F_\alpha (\mathbf A)_{(l_1l_2)\cdots (l_{2N-1}l_{2N})}\\
    = &\, \, 
     \frac{(2\pi)^{m/ 2}}{4^N N!}  \int_{\triangle^n} \omega_\alpha(u) 
     \sbrac{
  (\det \mathbf B)^{-N-1/2} 
  \brac{
      \partial_{\mathbf B_{l_1 l_2}}(\det \mathbf B)  \cdots 
      \partial_{\mathbf B_{l_{2N-1} l_{2N}}}(\det \mathbf B)  
  }      } \bigg |_{(u,n)}
du.
\end{align*}

We remind the reader three parameters $m$, $n$ and $N$ which has been
frequently used in the construction of $F_\alpha(\mathbf A)$  above:
\begin{enumerate}
    \item $m$ denotes the dimension of the noncommutative tori. It determines the
        length of $\xi \in  \R^m$ and the dimensions  of the matrices such as:
        $\mathbf A^{(j)}$, $\mathbf Y^{(j)}$, $\mathbf Z^{(j)}$   and $\mathbf
        B(u,n)$.  
    \item $n$ comes from the length of $\alpha$: for  $\alpha \in
        \Z_{\ge 1}^{n+1}$, $F_\alpha(\mathbf A)$ acts on $A^{\otimes  n}$ 
       and the integration is taken over the standard $n$-simplex: 
       $u \in \triangle^n$.  
    \item $N$ appears when expanding the formal differential operator 
$F_\alpha(\mathbf A)$  (in $\xi$) 
         in \cref{eq:FaplhaA-expands-to-components}. In applications,
         $N$ is subject to a relation \cref{eq:homogeneq-for-xi}.
\end{enumerate}

\subsection{The diagonal case}
\label{subsec:the-diag-case-spec-funs}
We now assume that $\mathbf A = \mathrm{diag}(k_{11}, \ldots ,k_{mm})$ is diagonal,
that is the leading symbol reads $p_2(\xi) = \sum_{s=1}^m k_s \xi_s^2$ with
abbreviation 
$k_s \defeq k_{ss}$. All the matrices defined in the previous section are
diagonal: $\{\mathbf A^{(l)}\}_{l=0}^n  $, $\{\mathbf Y^{(l)}\}_{l=1}^n $,
$\{\mathbf Z^{(l)}\}_{l=1}^n $, and $\mathbf B_n(u)$. For $s=1, \ldots, m$:  
\begin{align}
    \label{eq:A-Y-Z-defn-diag}
    (\mathbf A^{(l)})_{ss} = k^{(l)}_s, \, \, \,
    (\mathbf Y^{(l)})_{ss} = \mathbf y_s^{(l)}=(k_s^{(l-1)})^{-1} k^{(l)}_s,
    \, \, \,
    (\mathbf Z^{(l)})_{ss} = \mathbf z_s^{(l)} =
    1- \mathbf y_s^{(1)} \cdots \mathbf y_s^{(l)},
\end{align}
and for $u= (u_1, \ldots, u_n)$ 
\begin{align*}
    (\mathbf B_n)_{ss} (u) = k_s^{(0)} 
    (1- \sum_{l=1}^{n} \mathbf z_s^{(l)} u_l).
\end{align*}
Therefore the components of $F_\alpha(\mathbf A)$ are ``diagonal'' in the sense
that
\begin{align*}
    F_\alpha(\mathbf A)_{(l_1 l_2) \cdots  (l_{2N-1} l_{2N})} = 
    (\mathbf 1_{l_1 l_2} \cdots \mathbf 1_{l_{2N-1} l_{2N}} ) 
    F_\alpha (\mathbf A)_{(l_2 l_2) \cdots  (l_{2N} l_{2N})} , 
\end{align*}
where $\mathbf 1_{ij}$ is the $(i,j)$-entry of the identity matrix. 
The components of $F_\alpha(\mathbf A)$ can be written as a product 
in which the first factor collects the contribution from the 
the left multiplications (entries of $\mathbf A \defeq \mathbf A^{(0)}$),
while the second factor consists of the action of the conjugation operators
generated by $\mathbf Z^{(l)}$ or $\mathbf Y^{(l)}$:
\begin{align}
    \label{eq:FaplhaA-vs-mathsfFaA}
    F_\alpha (\mathbf A)_{(s_1 s_1) \cdots  (s_N s_N)} = 
    \frac{1}{4^N N!}
    \det(\mathbf A)^{-\frac{1}{2}}
    \brac{ \prod_{l=1}^N k_{s_l} }^{-1} 
    \mathsf F_\alpha(\mathbf z)_{s_1, \ldots, s_N}
\end{align}
where $\mathbf z$ denotes the matrix $\mathbf z = ( \mathbf z_j^{(l)})$, 
$1\le  j \le m$ and $1\le  l \le  n$. 
The spectral   functions of $\mathsf F_\alpha(\mathbf z)_{s_1, \ldots, s_N}$ 
are of $m \times  n$ variables: $z \defeq z(m,n) = (z_j^{(l)})$, 
with $1\le  j \le m$ and  $1\le l \le  n$: 
\begin{align}
    \mathsf F_\alpha(z)_{(s_1, \ldots, s_N)} =
  (2\pi)^{\frac{m}{2}}  \int_{\triangle_n} 
    \omega_\alpha(u) (\det Z_n(z,u))^{-\frac{1}{2}}
    \brac{
       \prod_{l=1}^N Z_n(z,u)_{s_l s_l} 
    }^{-1} 
    d u,
    \label{eq:sfFa-defn}
\end{align}
where $Z_n(z,u)_{s_l s_l}$ is the $s_l$-th diagonal entry of the $m \times  m$
matrix: 
\begin{align*}
Z_n(z,u) =\mathrm{diag}(
1- \sum_{l=1}^n z_1^{(l)}u_l ,\cdots,                                     
1- \sum_{l=1}^n z_m^{(l)}u_l 
)   .
\end{align*}
When $N=0$,
\begin{align*}
    \mathsf F_\alpha(z)_{\emptyset} = 
  (2\pi)^{\frac{m}{2}}
  \int_{\triangle_n} 
    \omega_\alpha(u) 
    (\det Z_n(z,u))^{-\frac{1}{2}}
   du. 
\end{align*}

\subsection{The conformal case}
\label{subsec:the-conformal-case-Halpha}

We now further assume that
$k =k_1 = \cdots =k_m $,  that is, there 
is only one noncommutative coordinate $k \in C^\infty(\T^m_\theta)$
positive and invertible and $\mathbf A = k I$ is a scalar matrix. 
The $\mathbf z$-variables in Eq.
\eqref{eq:A-Y-Z-defn-diag} becomes: 
$\mathbf Z^{(l)} = \mathbf z^{(l)} I$ with $\mathbf y = k^{-1}(\cdot) k$ and 
\begin{align*}
  \mathbf  z^{(1)} = 1-\mathbf y, \, \, \,
  \mathbf  z^{(2)} = 1- \mathbf y^{(1)} \mathbf y^{(2)}, \, \, \,
    \cdots, \, \, \,
   \mathbf z^{(n)} = 1- \mathbf y^{(1)} \cdots \mathbf y^{(n)}.
\end{align*}
With $z^{(l)}_1 = \cdots =z^{(l)}_m = z^{(l)}$, 
the integrand of in Eq. \eqref{eq:sfFa-defn} is reduced to:
\begin{align*}
    (\det Z_n(z,u))^{-\frac{1}{2}}
    \brac{
       \prod_{l=1}^N Z_n(z,u)_{s_l s_l} 
    }^{-1} 
    =\brac{
        1- \sum_{l=1}^n z^{(l)} u_l
    }^{-\frac{m}{2}-N}.
\end{align*}
The components $\mathsf F_\alpha(\mathbf z)_{(s_1, \ldots, s_N)}$ are all
identical, in other words, only the length $N$ matters. 
    In later computation, $F_\alpha(\mathbf A)$ turns up when integrating
    the resolvent approximations $\{b_j(\xi,\lambda)\}_{j=0}^\infty $ (cf.
    \S\ref{subsec:resolvent-appox}). Each $b_j(\xi,\lambda)$ is of              
    of degree $-j-2$ (in $\xi$) and consists of summands of the form 
        \begin{align}
            \label{eq:N-remarks}
            b_0^{\alpha_0} \rho_1  b_0^{\alpha_1} \cdots \rho_n b_0^{\alpha_n}  
            = (b_{0}^{(0)})^{\alpha_0} \cdots (b_{0}^{(n)})^{\alpha_n} 
            \cdot \brac{
            \rho_1 \otimes  \cdots \otimes  \rho_n
            }
        \end{align}
        which contributes to a term $F_{\alpha}(\mathbf A)(\rho_1 \otimes
        \cdots \otimes  \rho_n)$ in the final result of the heat coefficient. 
In \cref{eq:N-remarks},
$b_0 = (p_{2}(\xi) - \lambda)^{-1}$ is of degree $-2$ in $\xi$, while           
$\rho_1(\xi) \otimes \cdots \otimes  \rho_n(\xi)$ is of degree $2N$
\footnote{If $\rho_1(\xi) \otimes \cdots \otimes  \rho_n(\xi)$ is of odd degree
in $\xi$, it automatically killed by $F_\alpha(\mathbf A) |_{\xi = 0}$
according to \cref{eq:FaplhaA-expands-to-components}. This observations also
explains the vanishing of all odd heat coefficients.}.   
        Hence $N$, $\alpha$ and $j$ are subjected to the  condition:
        \begin{align}
            -2(\sum_{s=0}^{n} \alpha_s )+ 2N = -j-2, \, \, \,
            \text{or} \, \, \,
            N= \sum_{s=0}^{n} \alpha_s -j/2 -1.
            \label{eq:homogeneq-for-xi}
        \end{align}
Let  $\bar z = (z^{(1)}, \ldots, z^{(n)})$, $\alpha = (\alpha_1, \ldots,
\alpha_n)$,  all the functions $\mathsf F_\alpha((z^{(j)}_l))_{s_1,
\ldots, s_N}$ in Eq. \eqref{eq:sfFa-defn}  are equal to  
\begin{align}
    \label{eq:sfHalpha}
    \mathsf H_\alpha(\bar z;m;j) =
(2\pi)^{\frac{m}{2}}
  \int_{\triangle_n} 
    \omega_\alpha(u) 
    \brac{
        1- \sum_{l=1}^n z^{(l)} u_l
    }^{-\frac{m}{2}-\sum_{s=0}^{n} \alpha_s +j /2+1} du.
\end{align}
We set a default value for  $j$: $\mathsf H_\alpha(\bar z;m)  \defeq \mathsf
H_\alpha(\bar z;m;2)$ when dealing with the second heat coefficient. 
Compared to the hypergeometric family $H_{\alpha}(\vec
z;m)$  used in \cite{Liu:2018aa,Liu:2018ab}, we have
\begin{align}
    \label{eq:Hnow-vs-Hbefore}
    \mathsf H_{\alpha}(\vec z;m) =  
    \frac{ (2\pi)^{ m / 2}}{ \Gamma(d(\alpha,m)) } 
    H_{\alpha}(\vec z;m)  = C_m 
    \frac{\Gamma(m / 2)}{ \Gamma(d(\alpha,m)) } 
    H_{\alpha}(\vec z;m)  ,
\end{align}
where $\alpha = (\alpha_0, \ldots , \alpha_n)$, and 
$ d(\alpha,m) \defeq  d(\alpha,2,m) $ with 
$d(\alpha,j,m) = \sum_0^n \alpha_l +  m /2 -j$. The constant $C_m$ is the
overall factor used in \cite{Liu:2018aa,Liu:2018ab}: 
\begin{align}
    C_m = 2^{m / 2} \frac{ \pi^{m / 2} }{\Gamma(m / 2)} = 2^{ m /2 -1} 
    \mathrm{vol}( \mathbb S^{m-1}) .
    \label{eq:overall-factor-before}
\end{align}


\section{Heat coefficients via Pseudo-differential Calculus} 
\label{sec:heat-coef-via-pscal}
We assume, in the section, the reader's  acquaintance with Connes's
pseudo-differential calculus attached to a $C^*$-dynamical system
see \cite{connes1980c} and \cite{MR967366,MR967808}.
In recent papers, \cite{MR3985230,MR3985231} give detailed discussions on  
on pseudo-differential operators on arbitrary noncommutative tori and 
\cite{MR3540454,lesch2018modular}  deal with pseudo-differential calculus 
acting on Heisenberg modules. 
The author 


 We only consider (pseudo-differential) operators acting on functions. 
 Without extra indication,  $P: C^\infty(\mathbb{T}^m_\theta) \to
 C^\infty(\mathbb{T}^m_\theta)$ will always denote
 an elliptic self-adjoint second order differential operator of the form in 
 \cref{eq:P-gen-form-defn} whose coefficient matrix $\mathbf A$ of the lead
 term fulfills  the conditions in \cref{defn:conds-for-A}.  
We would like to outline the computation of the heat coefficients 
in the small time asymptotic of $\Tr(a e^{-tP})$ in 
\cref{eq:e^tP-asym-defn} with special focus on the $V_2$-term.

\subsection{Symbol calculus}
\label{subsec:symbol-cal}
The space of parametric symbols is contained in  $p(\xi,
\lambda) \in C^\infty(\R^m \times \Lambda, C^\infty(\T^m_\theta))$,
the analog of functions on the cotangent bundle of $\T^m_\theta$, where the
domain of the resolvent parameter $\Lambda \subset \mathbb{C}$ is a conic
subset.  We will encounter only homogeneous symbols in the paper,  on which
the filtration (or the graded structure)  is reflected on the homogeneity
condition:
a degree $d$ symbol satisfies, with $d \in  \R$,   
\begin{align}
    p( c \xi, \sqrt{c} \lambda) = c^d p(\xi, \lambda), \, \, \,
    \forall  c>0.
    \label{eq:p-homogeniety-defn}
\end{align}
Notice that $\lambda$ is treated as a degree two symbol since 
the elliptic operator in question is of second order.
The key ingredient of the symbolic calculus is
the formal star product $\star$ represents the symbol of the
composition of two  pseudo-differential operators $P$ and $Q$:
\begin{align}
    \label{eq:star-prod-formal}
    \pmb\sigma(PQ) = p \star q   \backsim \sum_{j=0}^\infty a_j(p,q), \, \, \,
    p = \pmb\sigma(P), q = \pmb\sigma (Q),
\end{align}
where $a_j(\cdot, \cdot)$ are bi-differential operators lowering the total degree
by $j$.

To incorporate the notations in \cite{Liu:2015aa,LIU2017138},
we put 
\begin{align*}
  \nabla_j = -i \delta_j: C^\infty(\R^m, C^\infty(\T^m_\theta)) \to
C^\infty(\R^m, C^\infty(\T^m_\theta))  
\end{align*}
which are the horizontal covariant differentials if we think $C^\infty(\R^m,
C^\infty(\T^m_\theta))$ as the smooth functions on the cotangent space of
$\T^m_\theta$. 
The vertical differentials $D$ is simply the derivatives in $\xi \in  \R^m$:
$D_s = \partial_{\xi_s}$. 
Notations, for higher derivatives, such as:
\begin{align*}
    (D^2 p)_{st} = D^2_{st} p \defeq \partial_{\xi_s}  \partial_{\xi_t} p, \, \, \,
    (\nabla ^2 p)_{st} = \nabla^2_{st} p \defeq 
    \nabla_s \nabla _t (p)=
    - \delta_s \delta_t (p)
\end{align*}
are freely used in later calculations.
In particular,the bi-differential operators in the $\star$-product 
(\cref{eq:star-prod-formal}) are given by: 
\begin{align}
    \label{eq:conness-a_j-defn}
    a_j(p,q) = \frac{(-i)^j}{j!} (D^j p)(\nabla^j q) \defeq
\frac{(-i)^j}{j!} 
    \sum_{ 1\le  l_1, \ldots, l_j \le m}
    (D^j p)_{l_1 \cdots l_j}
    (\nabla ^j p)_{l_1 \cdots l_j},
\end{align}
where $(D^j p)(\nabla^j q)$ is the contraction of two rank $j$ tensor:
contravariant $D^j p$ and covariant $\nabla^j q$ respectively.    
Of course, the multiplication among the summands is the one induced from 
$C^\infty(\mathbb{T}^m_\theta)$. 

For differential operators, computation of symbols is similar to the classical
counterpart. 
In detail, they are polynomials in $\xi$ constructed on generators in   the
following way: 
for the basic derivations $\{\delta_j\}_{j=1}^m$ and coordinate functions $f \in
C^\infty(\T^m_\theta)$ (via left-multiplication), we have:
\begin{align*}
   \pmb\sigma(\nabla_j) = \pmb\sigma(-i\delta_j) = -i \xi_j, \, \, \,
    \pmb\sigma(f) = f.
\end{align*}
The rest is determined via the $\star$-product Eq. \eqref{eq:star-prod-formal},
which is a finite sum when $p,q$ are polynomials in $\xi$.   


\subsection{Resolvent approximation}
\label{subsec:resolvent-appox}

Let $P$ be such an elliptic operator in of the form in Eq.
\eqref{eq:P-gen-form-defn} with the heat operator given in 
\cref{eq:heatop-P-defn} via holomorphic functional calculus which 
suggests that one shall start with the resolvent $(P-\lambda)^{-1}$.
We write the symbol $\pmb\sigma(P - \lambda) 
= p_2(\xi, \lambda) +p_1 (\xi)+p_0$, 
so that $p_l$ is homogeneous of degree (cf. \cref{eq:p-homogeniety-defn})
$l$, $l=0,1,2$.  
Note that  the resolvent parameter $\lambda$ is
grouped with the leading terms $p_2(\xi,\lambda) = p_2(\xi) -\lambda$ 
as they are both of degree $2$. 
We assume formally $ \pmb\sigma((P
-\lambda)^{-1}) \sim \sum_{j=0}^\infty b_j(\xi, \lambda)$ with $b_j$ of degree
$-2-j$.
The inverse is taken
with respect to the $\star$-product:
\begin{align*}
     (b_0+b_1+ \cdots) \star   \sigma(P) 
        = \sum_{0\le j,r\le \infty} \sum_{l=0}^2 a_j(b_r,p_l )
        \sim 1.
    \end{align*}
We compare two sides according to the homogeneity.     
Since summand $a_j(b_r,p_l )$ is of degree $l-2-r-j$, we get, by collecting 
terms of degree $i=0,-1,-2,\ldots$: 
\begin{align}
    \label{eq:deg-0-resolvent-aprx}
    a_0(b_0,p_2) &= 1, \\\, \, \,
    \sum_{l=0}^2 \sum_{r=0}^{N} a_{l-2-r+N}(b_r,p_l) &= 0, \, \, \,
   N=1,2,\ldots.
    \label{eq:deg-N-resolvent-aprx}
\end{align}
Since $a_0(p,q) = pq$, 
the first approximation $b_0$ is simply the resolvent of the
leading symbol $p_2$: 
$b_0 = (p_2(\xi,\lambda))^{-1} = (p_2(\xi) - \lambda)^{-1}$, where 
the inverse  is taken in
$C^\infty(\R^m, C^\infty(\T^m_\theta))$, $\forall  \xi \neq 0$ whose existence is 
provided by the ellipticity of $P$.
By solving equation in Eq. \eqref{eq:deg-N-resolvent-aprx} one by one (for
$N=1,2, \ldots$), we obtain the recursive formulas of $b_N$:
\begin{align}
    \label{eq:bN-genernal}
    b_N =   \brac{
        \sum_{l=0}^2 \sum_{r=0}^{N-1}
        a_{l-2-r+N}(b_r,p_l)
    }(-b_0).
\end{align} 
For example, the first two terms are given by:
\begin{align}
    \label{eq:b2cal-b1term-gen}
 b_1 &= [ 
 a_0\left(b_0,p_1\right)
 + a_1\left(b_0,p_2\right)]
 (-b_0)
\\
 b_2 &= 
 [a_0\left(b_0,p_0\right)+a_0\left(b_1,p_1\right)+
 a_1\left(b_0,p_1\right)+a_1\left(b_1,p_2\right)+
 a_2\left(b_0,p_2\right)]
 (-b_0),
\label{eq:b2cal-b2term-gen}
\end{align}
with 
\begin{align}
\label{eq:a_0-to-a_2}
    a_0(p,q) = p q, \, \, \,
    a_1(p,q) = -i (D p) (\nabla q), \, \, \,
    a_2(p,q) = -\frac{1}{2} ( D^2 p) ( \nabla^2  q)
    .
\end{align}
By carefully expand the right hand sides of
\cref{eq:b2cal-b1term-gen,eq:b2cal-b2term-gen} according to 
\cref{eq:a_0-to-a_2}, we get
\begin{align*}
    b_2 = (b_2)_{\mathbf{I}}  + (b_2)_{\mathbf{II}},
\end{align*}
where the terms are grouped apropos to the number of $b_0$-factors.   
Here is part I:
\begin{align}
    \label{eq:b2I-gen}
\begin{split} 
    (b_2)_{\mathbf{I}} &= -\frac{1}{2} 
    b_0^{2}
     ( D^2_{st} p_2) \cdot (\nabla^2_{st} p_2) b_0 +
     b_0^3 ( D_s p_2) \cdot ( D_t p_2) \cdot (\nabla^2_{st} p_2) b_0
     \\
& -ib_0^2 (D_s p_2) (\nabla_s p_1) b_0 - b_0 (p_0) b_0,                     
\end{split}
\end{align}
and then part II:
\begin{align}
    \label{eq:b2II-gen}
        \begin{split}
       (b_2)_{\mathbf{II}} 
    & =
     b_0^2 (D_s p_2) (D_t (\nabla_s p_2)) b_0 (\nabla_t p_2) b_0  
     - b_0^2 (D_s p_2) (\nabla_s p_2) b_0^2 (D_t p_2) (\nabla_t p_2) b_0 
     \\ &+
   b_0^2 (D^2_{st} p_2) (\nabla_s p_2) b_0 (\nabla_t p_2) b_0 
   -2  b_0^3 (D_s p_2) (D_t p_2) (\nabla_s p_2) b_0 (\nabla_t p_2) b_0
   \\
        &+i b_0^2 ( D_s p_2) ( \nabla_s p_2)b_0 (p_1) b_0 
        -i b_0 (D_s p_1)b_0 (\nabla_s p_2)b_0
        + i b_0 (p_1)b_0^2 (D_s p_2)(\nabla_s p_2) b_0    \\ 
        & + b_0 (p_1) b_0 (p_1)b_0,
        \end{split}
\end{align}
where summations are taken over repeated indices $s,t$ from $1$  to $m$.

\subsection{Rearrangement lemma}
\label{subsec:rearr-lem}
 The resolvent approximation $\{b_j\}_{j=0}^\infty$ determines the heat
 coefficients in the following way.
\begin{prop}
    \label{prop:bj-to-vjP}
    In the light of Eq. \eqref{eq:heatop-P-defn},
    the heat coefficient $v_j(P)$, $j=0,1,2,\ldots$,  is completely  determined
    by $b_j$ in the
    following way:
    \begin{align}
    \label{eq:bj-to-vjP}
        v_j(P) = 
        \frac{1}{2\pi i}
        \int_{\R^m}\int_{\mathcal C}
        e^{-\lambda} b_j(\xi,\lambda)d\lambda d\xi
        .
    \end{align}
\end{prop}
\begin{proof}
    This is a standard result in pseudo-differential calculus. One first 
    establishes a trace formula linking the operator 
 trace of a pseudo-differential operator with its symbol
 (cf. \cite[Theorem 5.7]{MR538027} for instance)
and then follows the argument in
\cite[\S1.8]{gilkey1995invariance} to reach the heat coefficients.
\end{proof}

Similar to the $b_2$-term given in  \cref{eq:b2I-gen,eq:b2II-gen}, we have, 
in general, that $b_j$ consists of finite sums of the form:
\begin{align}
      b_j = \sum 
      b_0^{\alpha_0} \rho_1  b_0^{\alpha_1} \cdots \rho_n b_0^{\alpha_n}  
            \defeq \sum 
            (b_{0}^{(0)})^{\alpha_0} \cdots (b_{0}^{(n)})^{\alpha_n} 
            \cdot \brac{
            \rho_1 \otimes  \cdots \otimes  \rho_n
            },
            \label{eq:bj-summands}
        \end{align}
where the $\rho$'s are the derivatives of $p_j$, $j=2,1,0$, which 
are polynomial in $\xi$ and have no dependence on $\lambda$. 
Therefore we factor out the $\rho$'s by making used of the notations in 
\cref{eq:a(j)-as-ops-iota} which leads to the rearrangement operator 
from $C^\infty(\mathbb{T}^m_\theta)[\xi]^{\otimes  n}$ to
$C^\infty(\mathbb{T}^m_\theta)$:  
\begin{align*}
    \rho_1 \otimes  \cdots \otimes  \rho_n
\to  \frac{1}{2\pi i}
        \int_{\R^m}\int_{\mathcal C}
        e^{-\lambda} (
            (b_{0}^{(0)})^{\alpha_0} \cdots (b_{0}^{(n)})^{\alpha_n} 
            )d\lambda \brac{
    \rho_1 \otimes  \cdots \otimes  \rho_n
            } d\xi.
\end{align*}
The rearrangement lemma (\cref{prop:rearg-lem}) asserts that it is exactly
the operator $F_\alpha(\mathbf A)  \big |_{\xi=0}$ 
described in Eq. \eqref{eq:Falpha(A)-defn}, 
where $\mathbf A = (k_{ij}) \in  \op{GL}_{m }(C^\infty(\T^m_\theta))$ 
is the coefficient matrix of $p_2(\xi)$.

We first deal with the contour integral in \cref{eq:bj-to-vjP}, for which
the origin is no longer a singularity\footnote{compared to the one in 
\cref{eq:heatop-P-defn}, in which the elliptic operator $P$ might have
non-trivial kernel}.  
We fix  $\mathcal C$ to be the imaginary axis $\lambda
= i x$, with  $x \in \R $ oriented from $-\infty$ to $\infty$. 
By replacing $b_j$ with the summands shown in \cref{eq:bj-summands}, we have
arrived at the integral in \cref{eq:G-bar-l-bar-A-defn}, which turns out to be
an hypergeometric integral (\cref{eq:G-bar-l-bar-A-lemma}).
\begin{lem}
        \label{lem:step1-contour-int}
    Let $\overline A = (A_0, \ldots, A_n) \in  \R^{n+1}$, $l_0, \ldots , l_n
    \in \N_+$, denote 
    \begin{align}
        \label{eq:G-bar-l-bar-A-defn}
        G_{l_0, \ldots , l_n}(\overline A) = 
        \frac{1}{2 \pi } \int_{-\infty}^\infty 
        e^{-ix} (A_0 - ix)^{-l_0} \cdots (A_n  - ix)^{-l_n}
       dx
    \end{align}
    Then $G_{l_0, \ldots , l_n}(\overline A)$  is equal to the following  confluent 
    type hypergeometric integral:
    \begin{align}
        & \, \, 
        G_{l_0, \ldots , l_n}(\overline A)  \nonumber \\
        = & \, \, 
        \brac{  \prod_{0}^n \Gamma(l_j) }^{-1}
     \int_{\triangle^n}
        (1 - \sum_1^n u_j)^{l_0 -1} \brac{  \prod_0^n u_j^{l_j -1}  }
        e^{ - \brac{ A_0(1- \sum_1^n u_l)
+ \sum_1^n  A_l u_l
        }
        }         du_1 \cdots du_n \nonumber \\
        =&\, \, 
        \int_{\triangle^n}
     \omega_l(u) 
        e^{ - \brac{  A_0(1- \sum_1^n u_l) + \sum_1^n  A_l u_l }}
        du
        \label{eq:G-bar-l-bar-A-lemma}
    \end{align}
where $\triangle^n  = \{ \sum_1^n u_j \le  1, u_1\ge 0, \ldots , u_n\ge 0\}
\subset \R^n$  is the standard simplex and  $\omega_l(u)$ is defined in 
Eq. \eqref{eq:omega-alpha-u-defn}.
\end{lem}
\begin{proof}
    We begin with rewriting each $ (A_j - ix)^{-l_j }$, $j=0,\ldots,n$ as a 
    Mellin transform: 
    \begin{align*}
        (A_j - ix)^{-l_j } = \frac{1}{\Gamma(l_j)} \int_0^\infty
        s_j^{l_j - 1} e^{ - (A_j - ix) s_j } ds_j
    \end{align*}
    so that $ G_{l_0, \ldots , l_n}(\overline A)$ becomes:
    \begin{align*}
& \, \, 
        \brac{  \prod_{0}^n \Gamma(l_j) }
        G_{l_0, \ldots , l_n}(\overline A)  \\
        = & \, \,         
        \frac{1}{2 \pi } \int_{-\infty}^\infty 
        \int_{[0,\infty)^{n+1}}
        e^{ ( \sum_0^n s_j - 1)ix}  \prod_{0}^n
        \brac{
        s_j^{l_j -1} e^{-A_j s_j}
        }
        ds_0 \cdots  ds_n  dx \\
        = & \, \,         
        \int_{[0,\infty)^{n}}  \frac{1}{2 \pi }
        \int_{\R} \int_{\sum_1^n u_j -1}^{\infty}
        e^{u_0 i x} (u_0 + 1- \sum_1^n u_j )^{-l_0-1}     
        e^{-A_0 (u_0 +1 - \sum_{1}^n u_j)} 
         \\
          & \, \,         
        \prod_{0}^n
        \brac{
        u_j^{l_j -1} e^{-A_j u_j}
        }
        du_0 dx (du_1 \cdots  du_n), 
    \end{align*}
    where the last line is obtained by the substitution:
    \begin{align*}
    u_0 = \sum_0^n s_j -1, \, \, \, u_1 = s_1, \ldots,
    u_n = s_n.
    \end{align*}
  We denote 
  \begin{align*}
      f(u_0, \ldots , u_n) = (u_0 + 1- \sum_1^n u_j )^{-l_0-1}     
        e^{-A_0 (u_0 +1 - \sum_{1}^n u_j)} 
        \prod_{0}^n
        \brac{
        u_j^{l_j -1} e^{-A_j u_j}
        }
  \end{align*}
  and view it as  a function in $u_0$. Set $ \tilde f (u_0) = 
  \mathbf 1_{ \{ u_0\ge  \sum_1^n u_j -1 \} } (u_0)
  f(u_0, \ldots, u_n) $, where $\mathbf 1_{ \{ u_0\ge \sum_1^n u_j -1 \}
  } (u_0) $ is the characteristic function in $u_0$ 
  of the set $ \{ u_0\ge  \sum_1^n u_j -1 \} \subset \R$.
  The Fourier inversion theorem with respect to $du_0 dx$ gives:
  \begin{align*}
      &  \int_{\R} \int_{\sum_1^n u_j -1}^{\infty}
        e^{u_0 i x} (u_0 + 1- \sum_1^n u_j )^{-l_0-1}     
        e^{-A_0 (u_0 +1 - \sum_{1}^n u_j)} 
        \prod_{0}^n
        \brac{
        u_j^{l_j -1} e^{-A_j u_j}
        }
        du_0 dx  \\ 
        = & \, \, 
        \int_{\R} \int_{\sum_1^n u_j -1}^{\infty}
        e^{u_0i x} f(u_0, \ldots , u_n) d u_0 dx =
        \int_{\R}  \int_{\R} 
        e^{u_0 i x} \tilde f(u_0, \ldots , u_n) d u_0 dx  \\
        = & \, \, 
\tilde f(0, u_1,\ldots , u_n) = f(0, u_1,\ldots , u_n)
\mathbf 1_{ \{ \sum_1^n u_j \le  1\} }(u_1, \ldots, u_n) .
  \end{align*}
We conclude the proof by observing 
that $f(0, u_1,\ldots
, u_n)$ is exactly the integral (including the factor $ e^{-A_0}$) appeared
on the right hand side of \eqref{eq:G-bar-l-bar-A-defn} and 
$\triangle^n = [0,\infty)^n \cap  \{ \sum_1^n u_j \le  1\}$. 
\end{proof}
The next step is the integration in $\xi$ in \cref{eq:bj-to-vjP}, which will be 
handled by the well-known Gaussian integral for polynomials.
\begin{lem}
    [Gaussian Integral]
    \label{lem:gaussian-int}
    Let $f(\xi)$ be a polynomial in $\xi = (\xi_1, \ldots, \xi_m)$ and  $B
    = (B_{ij})$ be a symmetric positive-definite $m \times m$ matrix
 with the inverse denoted by $(B^{-1})_{ij}$, then  
\begin{align*}
    \int_{\R^m} e^{ -\sum_{1\le  i,j \le m}
    B_{ij} \xi_{i}\xi_j} f( \xi) d  \xi =
        \frac{ (2\pi)^{m/ 2} }{\sqrt{\det B} }
    \exp \brac{
            \frac{1}{4} \sum_{1\le  i,j \le m}
( B^{-1})_{ij} \partial_{\xi_{i} } \partial_{\xi_j}
    } f(\xi)  \bigg |_{\xi = 0}
\end{align*}
where the exponential over differential operators is interpreted as power series.
\end{lem}


Consider the elliptic operator $P$ given in Eq. \eqref{eq:P-gen-form-defn} 
with leading symbol 
$p_{2}(\xi) = \sum_{s,t=1}^m k_{st}\xi_s \xi_t$. Denote by  $\mathbf
A = (k_{ij})_{m \times  m}$ the coefficient matrix and put 
$\mathbf A^{(l)} = (k_{ij}^{(l)})$, $l=1,\ldots,n$.
For each $u = (u_0, \ldots, u_n) \in \triangle^n \subset \R^n$ in the standard
simplex,  
let $\mathbf B_n (u)$ be the coupled matrix defined in \cref{eq:Bu_n-defn}:
    \begin{align*}
        \mathbf B_n (u) =   \mathbf A^{(0)}( 1- \sum_1^n u_l) +  
        \sum_1^n \mathbf A^{(l)} u_l  = 
        \mathbf A^{(0)}(1- \sum_1^n \mathbf Z^{(l)} u_l )
        .
    \end{align*}
    
\begin{prop}
    \label{prop:rearg-lem}
    Let $\alpha = (\alpha_0, \ldots, \alpha_n) \in  \R_{>0}^{n+1}$, 
    and $ \rho_1 \defeq \rho_1(\xi), \ldots, \rho_n \defeq \rho_n(\xi) 
    \in   C^\infty(\T^m_\theta)[\xi]$ are polynomial symbols.
    Put $\bar \rho (\xi) = \rho_1 (\xi)\otimes \cdots
    \otimes  \rho_n (\xi)$. If we apply the integral onto a typical summand
    appeared in the resolvent approximation, the result is 
    The operator $F_\alpha(\mathbf A) |_{\xi = 0}$ defined in
    \eqref{eq:Falpha(A)-defn} computes the integral in Eq.
    \eqref{eq:bj-to-vjP} applied onto summands of the resolvent approximation
    $($see Eq. \eqref{eq:bj-summands}$):$
    \begin{align*}
    & \, \, 
     \frac{1}{ 2 \pi i} \int_{\R^m} \int_{\mathcal C} e^{-\lambda}
b_0^{\alpha_0} \rho_1 b_0^{\alpha_1} \rho_2 \cdots 
b_0^{\alpha_{n-1}} \rho_n b_0^{\alpha_n} 
           d\lambda d\xi \\
        =&\, \, 
    \int_{\triangle^n}   
    \omega_\alpha (u) 
    \frac{ (2\pi)^{m/ 2} }{\sqrt{\det \mathbf B_n(u)} }
    \exp \brac{
            \frac{1}{4} \sum_{1\le  i,j \le n}
            ( \mathbf B_n^{-1} (u))_{ij} \partial_{\xi_{i} } \partial_{\xi_j}
}
\brac{
\bar \rho (\xi)
}
    \bigg |_{\xi = 0}
    du \\
        = &\, \, 
        F_\alpha(\mathbf A)  
    \big |_{\xi = 0}
        \brac{
\bar \rho (\xi)
        }
    \end{align*}
\end{prop}

\begin{proof}
   We start with Lemma \ref{lem:step1-contour-int} with operator-valued arguments
   $A_j = p_2^{(j)}(\xi) = \sum_{s,t=1}^m k_{st}^{(j)} \xi_s \xi_t$,  
   $j=0,1,\ldots,n$: 
    \begin{align*}
     &\, \, 
     \frac{1}{ 2 \pi i} \int_{\mathcal C} e^{-\lambda}
b_0^{\alpha_0} \rho_1 b_0^{\alpha_1} \rho_2 \cdots 
b_0^{\alpha_{n-1}} \rho_n b_0^{\alpha_n} 
           d\lambda \\
           =&\, \, 
           \frac{1}{ 2 \pi i} \int_{\mathcal C} e^{-\lambda}
     (p_2^{(0)} - \lambda)^{\alpha_0} \cdots      
     (p_2^{(n)} - \lambda)^{\alpha_n} d\lambda 
     \brac{\rho_1 \otimes  \cdots \otimes  \rho_n} 
     \\
           =&\, \, 
 G_\alpha(p_2^{(0)}, \ldots, p_2^{(n)}) 
     \brac{\rho_1 \otimes  \cdots \otimes  \rho_n} \\
           =&\, \,  
           \int_{\triangle^n} \omega_\alpha(u)
           \exp\brac{- \brac{p_2^{(0)}(1- \sum_{l=1}^n u_l) 
           + \sum_{l=1}^n p_2^{(l)} u_l}}    
     \brac{\rho_1 \otimes  \cdots \otimes  \rho_n}. 
    \end{align*}
Notice that we have silently  quoted \cref{eq:fubibi-type-result} (more
than once) in which the
Borel spaces are $\mathcal C$ and $\triangle^n$, the verification of the
integrability  condition \cref{eq:integrability-condition} is straightforward
and left to the reader.

    The coupled matrix $\mathbf B_n(u)$ is the coefficient matrix of the sum:   
\begin{align*}
    p_2^{(0)}(\xi)(1- \sum_{l=1}^n u_l) 
    + \sum_{l=1}^n p_2^{(l)}(\xi) u_l = 
\sum_{0\le j,l \le m} \mathbf B_n (u)_{jl} \xi_j \xi_l.
\end{align*}
The result follows immediately from the Gaussian integral lemma
\ref{lem:gaussian-int}:
\begin{align*}
    &\, \,  \int_{\R^m}
           \exp\brac{- \brac{p_2^{(0)}(1- \sum_{l=1}^n u_l) 
           + \sum_{l=1}^n p_2^{(l)} u_l}} (\bar \rho(\xi))    d\xi \\
           =&\, \, 
    \frac{ (2\pi)^{m/ 2} }{\sqrt{\det \mathbf B_n(u)} }
    \exp \brac{
            \frac{1}{4} \sum_{1\le  i,j \le n}
            ( \mathbf B^{-1} (u,n))_{ij} \partial_{\xi_{i} } \partial_{\xi_j}
} 
\brac{
\bar \rho (\xi)
}
    \bigg |_{\xi = 0}.
\end{align*}

\end{proof}

 \subsection{General form of the second heat coefficients}
 \label{subsec:gen-form-V2}
Now let us work out the rest of computation for the $V_2$-term 
starting with Prop. \ref{prop:bj-to-vjP}, 
\begin{align*}
    v_2(P) = 
        \frac{1}{2\pi i}
        \int_{\R^m}\int_{\mathcal C}
        e^{-\lambda} ( b_2(\xi,\lambda)_{\mathbf I} +
        b_2(\xi,\lambda)_{\mathbf{II}} )
        d\lambda d\xi,
\end{align*}
where $b_2(\xi, \lambda)_{\mathbf{I}}$ 
and $b_2(\xi, \lambda)_{\mathbf{II}}$ are recorded  
 in   \cref{eq:b2I-gen,eq:b2II-gen}.
 We have just shown  that the result of the integration is given by the
 rearrangement operators $F_\alpha(\mathbf A) |_{\xi =0}:
 C^\infty(\T^m_\theta)^{\otimes  n}[\xi] \to C^\infty(\T^m_\theta)$.
 More precisely, we just have to carry out substitutions like:
\begin{align*}
    b_0^{2} ( D^2_{st} p_2) \cdot (\nabla^2_{st} p_2) b_0 &\to 
    F_{2,1}(\mathbf A)( 
    ( D^2_{st} p_2) \cdot (\nabla^2_{st} p_2)
    )
    \\
   b_0^2 (D^2_{st} p_2) (\nabla_s p_2) b_0 (\nabla_t p_2) b_0 
                                                          &\to 
F_{2,1,1}(\mathbf A)
   ( (D^2_{st} p_2) (\nabla_s p_2) \otimes  (\nabla_t p_2))
\end{align*} 
on individual terms  in   \cref{eq:b2I-gen,eq:b2II-gen}.
After that, we have obtained a relatively compact version of the second heat
coefficient.
\begin{thm}
    \label{thm:v2P-F(A)}
    We shall group the summands of $v_2(P)$ in terms of the parameter $\alpha$ in 
    $F_\alpha(\mathbf A)$$:$ 
    \begin{align*}
       v_2(P)=v_2(P)_{\mathbf{I}} + v_2(P)_{\mathbf{II}} .
    \end{align*}
The first part    
$v_2(P)_{\mathbf I} = v_2(P)_{\mathbf I, 2,1} + v_2(P)_{\mathbf I, 3,1}
+ v_2(P)_{\mathbf I, 1,1}:$ 
\begin{align}
 \begin{split}
    v_2(P)_{\mathbf I,2,1} & 
    = 
    \sum_{s,t}
    - F_{2,1}(\mathbf A)
    \brac{
     \frac{1}{2}   ( D^2_{st} p_2) (\nabla^2_{st} p_2)
    + i (D_s p_2) (\nabla_s p_1)
}  \\
    v_2(P)_{\mathbf I,3,1}  
                           &= 
 F_{3,1}(\mathbf A)\brac{
( D_s p_2) ( D_t p_2) (\nabla^2_{st} p_2)
} \\
    v_2(P)_{\mathbf I,1,1}  
  & = - F_{1,1}(\mathbf A)\brac{ p_0
} 
\end{split}   
\end{align}
The second part consists of $:$
\begin{align*}
v_2(P)_{\mathbf{II}} =
v_2(P)_{\mathbf{II},2,1,1}
+v_2(P)_{\mathbf{II},3,1,1}
+v_2(P)_{\mathbf{II},2,2,1}
+v_2(P)_{\mathbf{II},1,2,1}
+v_2(P)_{\mathbf{II},1,1,1} ,
\end{align*}
in which: 
\begin{align}
    \label{eq:v_2P-211}
  &\, \,   v_2(P)_{\mathbf{II},2,1,1}  
    \\ =&\, \, 
 \sum_{s,t}
    F_{2,1,1}(\mathbf A) 
    \left[  
\vphantom{\frac{1}{2}} 
    (D_s p_2) (D_t (\nabla_s p_2)) \otimes (\nabla_t p_2) 
+ (D^2_{st} p_2) (\nabla_s p_2) \otimes  (\nabla_t p_2) 
\right. \nonumber \\
+&\, \, \left.
i ( D_s p_2) ( \nabla_s p_2) \otimes  p_1
+ i ( D_s p_2) p_1 \otimes  ( \nabla_s p_2) 
\vphantom{\frac{1}{2}} 
\right] ,
\nonumber
\end{align}
and
\begin{align}
    \label{eq:v_2P-121}
    v_2(P)_{\mathbf{II},1,2,1} & = 
 i\sum_s
 F_{1,2,1}(\mathbf A) \brac{
       p_1 \otimes  (D_s p_2)(\nabla_s p_2)
  },
\end{align}
and
\begin{align}
    \label{eq:v_2P-311}
    v_2(P)_{\mathbf{II},3,1,1} & =
 -2 \sum_{s,t}
 F_{3,1,1}(\mathbf A) \brac{
 (D_s p_2) (D_t p_2) (\nabla_s p_2) \otimes  (\nabla_t p_2)
  }\\
    \label{eq:v_2P-221}
    v_2(P)_{\mathbf{II},2,2,1} 
                               & =
  -F_{2,2,1}(\mathbf A) \brac{
 (D_s p_2) (\nabla_s p_2) \otimes  (D_t p_2) (\nabla_t p_2)
  },
\end{align}
and
\begin{align}
    \label{eq:v_2P-111}
    v_2(P)_{\mathbf{II},1,1,1} & = 
    F_{1,1,1}(\mathbf A) \brac{ 
      p_1 \otimes  p_1
    } +
  \sum_s
 F_{1,1,1}(\mathbf A) \brac{
  -i     (D_s p_1) \otimes  (\nabla_s p_2) 
  }.
\end{align}
\end{thm}

For piratical purposes, we need the fully expanded version for all the
terms above.
 For the differential operator $P$ given 
in \eqref{eq:P-gen-form-defn}, the symbols reads explicitly:
$p_0 = p_0$ and
\begin{align}
   p_{2}(\xi) = \sum_{i,j} k_{ij}\xi_i \xi_j , \, \, \,
p_1(\xi) = \sum_s r_s \xi_s,
\label{eq:p_2-and-p_1-in-xi}
\end{align}
with derivatives (will be needed later):
\begin{align}
    \label{eq:derivs-p_2-p_1}
    \begin{split}
 &  D_s p_2 = \sum_l k_{sl}\xi_l, \, \, \,
   D^2_{st} p_2 = 2 k_{st}, \, \, \, D_s p_1 = r_s, \, \, \, \\
 & \nabla^2_{jl}(D^2_{st} p_2) =2 (\nabla^2_{jl} k_{st})
 \nabla_l( D_s p_1) = \nabla_l( r_s ).
    \end{split}
\end{align}
To expand $F_\alpha(\mathbf A)$ using 
Eq. \eqref{eq:FaplhaA-expands-to-components},
we recall a general result to compute the partial derivative 
$\partial_{\xi_{l_1}} \cdots  \partial_{\xi_{l_{2N}}}$.
\begin{lem}
    \label{lem:generalized-Lieb-poly}
    Let $\{\rho_l(\xi)\}_{l=1}^n$ be homogeneous polynomial in $\xi$ with
    degree $\beta_l = \deg \rho_l \ge 1$. 
    Then  the partial derivative of the product $\rho_1\cdots \rho_n$ 
\begin{align*}
    \partial^{2N}_{\xi_{l_1} \cdots \xi_{l_{2N}}} \brac{
        \rho_1(\xi) \cdots  \rho_n(\xi)
    }\big |_{\xi =0} 
\end{align*}
is non-zero only when $\beta = (\beta_1, \ldots, \beta_n)$ is partition of
$2N$, that is $2N = \sum_{l=1}^n \beta_l$ and
\begin{align*}
    \partial^{2N}_{\xi_{l_1} \cdots \xi_{l_{2N}}} \brac{
        \rho_1(\xi) \cdots  \rho_n(\xi)
    }\big |_{\xi =0} =
    \partial_{\beta,\bar l}^{(1)} (\rho_1 (\xi)) \cdots  
    \partial_{\beta,\bar l}^{(n)} (\rho_n(\xi)) 
\end{align*}
where we split the partial derivatives according to the partition $\beta$: 
\begin{align*}
&    \partial_{\beta,\bar l}^{(1)} = 
    \frac{1}{\beta_1!}
    \prod_{\nu = 1}^{\beta_1} \partial_{\xi_{\bar l_\nu}}
    , \, \, \,
    \partial_{\beta,\bar l}^{(2)} 
    =
    \frac{1}{\beta_2!}
    \prod_{\nu = \beta_1+1}^{\beta_1+\beta_2} \partial_{\xi_{\bar l_\nu}},
\, \, \, \ldots,
    \partial_{\beta,\bar l}^{(n)} 
    =
    \frac{1}{\beta_n!}
    \prod_{\nu = 1+ \sum_{s=1}^{n-1}\beta_s }^{2N} \partial_{\xi_{\bar l_\nu}}
         . 
\end{align*}
Last but not least, those barred $l$'s indicate  the following
symmetrization $($without dividing by the factorial factor$)$ occurs$:$ 
\begin{align}
    (**)_{\bar l_1, \ldots, \bar l_{2N}} = \sum_{\tau \in S_{2N}}
    (**)_{\tau(l_1) \ldots, \tau(l_{2N})}.
    \label{eq:symmetrization-l_1-l_2N}
\end{align}
\end{lem}
\begin{proof}
    Despite the lengthy notations,  operations  behind is quite simple.
    We distribute $2N$ partial derivatives to the factors $\rho_1, \ldots, \rho_n$
    following the general Leibniz property and then collect the non-zero terms
    after evaluating at $\xi =0$.
    Notice that non-trivial contribution only occurs in the situation in which
    $\rho_l$ is differentiated exactly $\beta_l$ times, $l=1, \ldots,n$.      
\end{proof}

Let us take Eq. \eqref{eq:v_2P-311} for example, $(D_s p_2) (D_t p_2) (\nabla_s
p_2) \otimes  (\nabla_t p_2)$ is a polynomial of degree $6$ with partition 
$\beta = (1,1,2,2)$, thus
\begin{align*}
  &\, \,   F_{3,1,1}(\mathbf A)  \brac{
    (D_s p_2) (D_t p_2) (\nabla_s
p_2) \otimes  (\nabla_t p_2)
} \\ 
    =&\, \,  
    F_{3,1,1}(\mathbf A)_{(l_1 l_2)(l_3 l_4) (l_5 l_6)}(\mathbf A)
    \partial_{\xi_{l_1}} \cdots \partial_{\xi_{l_{6}}}
    \brac{
    (D_s p_2) (D_t p_2) (\nabla_s
p_2) \otimes  (\nabla_t p_2)
    } 
    |_{\xi =0} \\
    = &\, \, 
    F_{3,1,1}(\mathbf A)_{(l_1 l_2)(l_3 l_4) (l_5 l_6)}(\mathbf A)
    \brac{
    \frac{1}{4}
    (D^2_{s \bar l_1} p_2) (D^2_{t \bar l_2} p_2) (\nabla_s
    D^2_{\bar l_3 \bar l_4}(p_2) ) \otimes  (\nabla_t D^2_{\bar l_5 \bar l_6}( p_2 ) )
    } .
\end{align*}
After repeating similar computation above to
all summands of 
$v_2(P)_{\mathbf{I}}$ and $v_2(P)_{\mathbf{II}}$ in 
Theorem \ref{thm:v2P-F(A)}, we have arrived at
 fully expanded version of the $V_2$-term which will be recorded   separately
 into  two theorems.
\begin{thm}
    \label{thm:v2P-I-components}
 With   the components of $F_\alpha(\mathbf A)$
defined in \eqref{eq:FaplhaA-components}:
    \begin{align*}
        F_\alpha (\mathbf A)_{(l_1l_2)\cdots (l_{2N-1}l_{2N})} :
        C^\infty(\T^m_\theta)^{\otimes  n} \to  C^\infty(\T^m_\theta),
    \end{align*}
we can further expand the summands of $v_2(P)_{\mathbf{I}}$ listed in 
Theorem $\ref{thm:v2P-F(A)}:$   
\begin{align*}
    v_2(P)_{\mathbf I,2,1} & =
    \sum_{s,t, l_1, l_2}
    - F_{2,1}(\mathbf A)_{l_1 l_2}
    \brac{
     \frac{1}{2}   ( D^2_{st} p_2) (\nabla^2_{st} (D^2_{l_1 l_2}p_2) )
    + i (D^2_{s \bar l_1} p_2) (\nabla_s (D_{\bar l_2} p_1) )
} , \\
                           &= 
    - 2 \sum_{s,t, l_1, l_2}
    F_{2,1}(\mathbf A)_{l_1 l_2}
    \brac{
    k_{st}(\nabla^2_{st} k_{l_1 l_2})+
  i k_{s \bar l_1} (\nabla_s r_{\bar l_2}) 
    },
\end{align*}
and
\begin{align*}
    v_2(P)_{\mathbf I,3,1}  
                           &= 
                   \sum_{s,t,l_1, \ldots, l_4}
                   F_{3,1}(\mathbf A)_{(l_1 l_2) (l_3 l_4)}
 \brac{
     \frac{1}{2}
( D^2_{s \bar l_1} p_2) ( D^2_{t \bar l_2} p_2) 
(\nabla^2_{st} (D^2_{\bar l_3 \bar l_4}p_2) )
}\\
                           &=
 4 \sum_{s,t,l_1, \ldots, l_4} F_{3,1}(\mathbf A)_{(l_1 l_2) (l_3 l_4)}
\brac{
    k_{s \bar l_1} k_{t \bar l_2} \nabla^2_{st}(k_{\bar l_3 \bar l_4}) 
}                          
,
\end{align*}
and
\begin{align*}
    v_2(P)_{\mathbf I,1,1}  
  & = - F_{1,1}(\mathbf A)_{\emptyset}
  \brac{ p_0
} .  
\end{align*}
Summation in $s,t, l_1, \ldots, l_4$  runs from $1$ to $m$. The barred
letters $\bar l_1, \ldots, \bar l_4 $  indicate symmetrization
$($without dividing by the factorial factor, cf. Eq.
\eqref{eq:symmetrization-l_1-l_2N}$)$ is applied.
%

\end{thm}

\begin{thm}
    \label{thm:v2P-II-components}
    Keep notations,
    we have the full expansion of
    $v_2(P)_{\mathbf{II}}:$
\begin{align}
    \label{eq:v2P-II-components-211}
    \begin{split}
v_2(P)_{\mathbf{II},2,1,1} &=
  \sum_{s,t,l_1, \ldots, l_4}
  F_{2,1,1}(\mathbf A)_{(l_1 l_2) (l_3 l_4)}
\left[ 
      \frac{1}{2}
    (D^2_{s \bar l_1} p_2) (\nabla_s(D^2_{t \bar l_2}  p_2)) 
    \otimes (\nabla_t  (D_{\bar l_3 \bar l_4} p_2) ) 
   \right.
   \\
   & + \left.       
\frac{1}{4}  
          (D^2_{st} p_2) (\nabla_s (D^2_{\bar l_1 \bar l_2} p_2) )
\otimes  (\nabla_t (D^2_{\bar l_3 \bar l_4} p_2) ) 
+
\frac{i}{2} 
( D^2_{s \bar l_1} p_2) ( \nabla_s (D^2_{\bar l_2 \bar l_3} p_2) ) 
\otimes  (D_{\bar l_4 } p_1)
 \right. \\
   & + \left.
\frac{i}{2} 
( D^2_{s \bar l_1} p_2)  (D_{\bar l_2 } p_1)
\otimes  ( \nabla_s (D^2_{\bar l_3 \bar l_4} p_2) )
\right] 
  \\
           &=
  \sum_{s,t,l_1, \ldots, l_4}
  F_{2,1,1}(\mathbf A)_{(l_1 l_2) (l_3 l_4)}
  \left[ 
    4  k_{s \bar l_1} (\nabla_s k_{t \bar l_2})   \otimes 
      (\nabla_t k_{\bar l_3 \bar l_4})   
+
    2  k_{st} (\nabla_s k_{\bar l_1 \bar l_2}) 
    \otimes (\nabla_t k_{\bar l_3 \bar l_4})
   \right. \\
           &+  \left.
   2 i k_{s \bar l_1} (\nabla_s k_{\bar l_2 \bar l_3} )\otimes r_{\bar l_4 }
   +2 i k_{s \bar l_1} r_{\bar l_2 } \otimes  (\nabla_s k_{\bar l_3 \bar l_4} )
  \right] 
    \end{split}
\end{align}
and
\begin{align}
    \label{eq:v2P-II-components-121}
    v_2(P)_{\mathbf{II},1,2,1} & = 
  \sum_{s, l_1, \ldots, l_4}
 i F_{1,2,1}(\mathbf A)_{ (l_1 l_2) (l_3 l_4) }
\brac{
  \frac{1}{2}  D_{\bar l_1} p_1 \otimes  ((D^2_{s \bar l_2}) p_2)
    (\nabla_s (D^2_{\bar l_3 \bar l_4} p_2))
} \\
                               &= 
  \sum_{s, l_1, \ldots, l_4}
  i F_{1,2,1}(\mathbf A)_{ (l_1 l_2) (l_3 l_4) }
  \brac{
 2 r_{l_1} \otimes  k_{s \bar l_2} (\nabla_s k_{\bar l_3 \bar l_4}) 
  } \nonumber
.
\end{align}
and
\begin{align}
    \label{eq:v2P-II-components-311}
  &\, \, v_2(P)_{\mathbf{II},3,1,1} \\
   =&\, \, 
  -2 \sum_{s,t, l_1, \ldots, l_6}
  F_{3,1,1}(\mathbf A)_{(l_1 l_2)(l_3 l_4) (l_5 l_6)}
  \brac{
      \frac{1}{4}
 (D^2_{s \bar l_1} p_2) (D^2_{t \bar l_2} p_2)
 (\nabla_s (D^2_{\bar l_3 \bar l_4} p_2) ) \otimes 
 (\nabla_t (D^2_{\bar l_5 \bar l_6} p_2))
  } 
  \nonumber \\
   = &\, \, 
- \sum_{s,t, l_1, \ldots, l_6}
  F_{3,1,1}(\mathbf A)_{(l_1 l_2)(l_3 l_4) (l_5 l_6)}
  \brac{
8  k_{s \bar l_1} k_{t \bar l_2} (\nabla_s k_{\bar l_3 \bar l_4})
 \otimes (\nabla_t k_{\bar l_5 \bar l_6}) 
  }, \nonumber
\end{align}
and 
\begin{align}
    \label{eq:v2P-II-components-221}
  &\, \,  v_2(P)_{\mathbf{II},2,2,1} \\
  =&\, \, 
  - \sum_{s,t, l_1, \ldots, l_6}
  F_{2,2,1}(\mathbf A)_{(l_1 l_2)(l_3 l_4) (l_5 l_6)}
  \brac{
      \frac{1}{4}
      (D^2_{s \bar l_1} p_2) (\nabla_s (D^2_{\bar l_2 \bar l_3} p_2))
      \otimes   (D^2_{t \bar l_4} p_2 ) 
      (\nabla_t (D^2_{\bar l_5 \bar l_6} p_2))
  },
  \nonumber \\
  =&\, \, 
  - \sum_{s,t, l_1, \ldots, l_6}
  F_{2,2,1}(\mathbf A)_{(l_1 l_2)(l_3 l_4) (l_5 l_6)}
  \brac{
4 k_{s \bar l_1} (\nabla_s k_{\bar l_2 \bar l_3} ) \otimes 
  k_{t \bar l_4}  (\nabla_t k_{\bar l_5 \bar l_6} )
  } \nonumber
\end{align}
and 
\begin{align}
    \label{eq:v2P-II-components-111}
    &\, \,  v_2(P)_{\mathbf{II},1,1,1} \\
     = &\, \, 
   -i \sum_{s, l_1,l_2}
F_{1,1,1}(\mathbf A)_{l_1 l_2}
\brac{
       (D_s p_1) \otimes  \nabla_s (D^2_{l_1 l_2} p_2)  
}
   + \sum_{ l_1,l_2}
F_{1,1,1}(\mathbf A)_{l_1 l_2}
\brac{
2 (D_{l_1} p_1) \otimes  (D_{l_2} p_1)
}
 \nonumber \\
       = &\, \, 
   -i \sum_{s, l_1,l_2}
F_{1,1,1}(\mathbf A)_{l_1 l_2}
\brac{
2 r_s \otimes (\nabla_s k_{l_1 l_2}) 
}
    + \sum_{ l_1,l_2}
F_{1,1,1}(\mathbf A)_{l_1 l_2}
\brac{
 2 r_{l_1} \otimes  r_{l_2}  
} . \nonumber
\end{align}
\end{thm}

\subsection{Diagonal Case}
\label{subsec:diagonal-case}

In this section, we  shall look at a special situation in which the
coefficient matrix
$\mathbf A = \op{diag}(k_1, \ldots, k_m)$ is diagonal, in other words,
the leading symbol of $P$ (defined  in Eq. \eqref{eq:P-gen-form-defn}) is of
the form $p_2(\xi) = \sum_1^m k_l \xi_l^2$. 
We have seen the simplification of the rearrangement operators 
$F_\alpha(\mathbf A)$ and their components in
\S\ref{subsec:the-diag-case-spec-funs} in which
the notations will be freely used.
The remaining work is to  simply the differential expressions on which the
rearrangement operators act.
 
Let us start with those terms in which no symmetrization occurs.
We only have to invoke Eq. \eqref{eq:FaplhaA-vs-mathsfFaA} to replace 
$F_\alpha(\mathbf A)$ by $\mathsf F_\alpha(\mathbf A)$:
\begin{align}
    \label{eq:diag-v_2P-11}
    v_2(P)_{\mathbf I,1,1} & = - F_{1,1}(\mathbf A)_{\emptyset} (p_0)
    =-(\det A)^{ -\frac{1}{2}}
    \mathsf F_{1,1}(\mathbf A)_{\emptyset} (p_0), 
\\
    \label{eq:diag-v_2P-21}
    v_2(P)_{\mathbf I,2,1} & = 
    \sum_{s,l}
    - 2 F_{2,1}(\mathbf A)_{l l}
    \brac{
k_{s}(\nabla^2_{ss} k_{ l})
}
    - 4 \sum_{s}
    F_{2,1}(\mathbf A)_{s s}
\brac{
  i k_{ s } (\nabla_s r_{ s }) 
} \\
                           & =
 (\det \mathbf A)^{-\frac{1}{2}}
 \brac{
-\frac{1}{2} \sum_{s,l}
\mathsf F_{2,1}(\mathbf z)_{ l}
 \left( \frac{k_s}{k_l} \nabla^2_{ss} k_{ l} \right) 
 -  \sum_s
\mathsf F_{2,1}(\mathbf z)_{s}
\brac{
i \nabla_s r_s
}
 }. \nonumber
\end{align}

The computation of the symmetrization
(appeared in in Eqs. \eqref{eq:v2P-II-components-211}
to \eqref{eq:v2P-II-components-111})
is straightforward, we only state the
results  in Lemma \ref{lem:sym-l1-l4} and \ref{lem:sym-l1-l6}.
\begin{lem} 
    \label{lem:sym-l1-l4}
The symmetrization over $\{l_1, \ldots, l_4\} $ yields $4! = 24$ terms which
are reduced to two cases divided  as $8+16$. Here are the summands:
\begin{align}
        \label{eq:diag-F-31}
    &\, \, 
                   \sum_{s,t,l_1, \ldots, l_4}
                   F_{3,1}(\mathbf A)_{(l_1 l_2) (l_3 l_4)}
 \brac{
      k_{s \bar l_1}  k_{t \bar l_2}   (\nabla^2_{st} k_{\bar l_3 \bar l_4})
} \\ = &\, \,  
    \sum_{s,t}              
  8  F_{3,1}(\mathbf A)_{(s s) (l l)}
   \brac{ k_{s s}^2 (\nabla^2_{ss} k_{l l}) }
+
 16 \sum_s  F_{3,1}(\mathbf A)_{(s s) (ss)}
   \brac{ k_{s s}^2 (\nabla^2_{ss} k_{ss}) },
   \nonumber
\end{align}
and
    \begin{align}
        \label{eq:diag-F-211-p-1}
        &\, \, 
        \sum_{s,t, l_1, \ldots, l_4}
        F_{2,1,1}(\mathbf A)_{(l_1 l_2 ) (l_3 l_4)} 
        \brac{
            k_{s \bar l_1}    (\nabla_s k_{ t \bar l_2 }) 
            \otimes  (\nabla_t k_{\bar l_3\bar l_4})
        }      
        \\
        = &\, \, 
        8 \sum_{s,l}
        F_{2,1,1}(\mathbf A)_{(s s) (l l)} 
        \brac{
         k_{ss} (\nabla_s k_{ss}) \otimes 
(\nabla_s k_{ll}) 
        }
+ 16 \sum_s F_{2,1,1}(\mathbf A)_{(s s) (s s)} 
\brac{
        k_{ss} (\nabla_s k_{ss}) \otimes 
(\nabla_s k_{s s})
},
   \nonumber
    \end{align}
    and
    \begin{align}
        \label{eq:diag-F-211-p-2}
        &\, \, 
        \sum_{s,t, l_1, \ldots, l_4}
        F_{2,1,1}(\mathbf A)_{(l_1 l_2 ) (l_3 l_4)} 
        \brac{
            k_{st} (\nabla_s k_{\bar l_1 \bar l_2})
            \otimes (\nabla_s k_{\bar l_3 \bar l_4})
        }
        \\
        = &\, \, 
        8 \sum_{s,l,t}
        F_{2,1,1}(\mathbf A)_{(t t) (l l)} 
        \brac{
         k_{ss} (\nabla_s k_{t t}) \otimes 
(\nabla_s k_{ll}) 
        }
+ 16 \sum_{s,l} F_{2,1,1}(\mathbf A)_{(l l) (l l)} 
\brac{
        k_{ss} (\nabla_s k_{l l}) \otimes 
(\nabla_s k_{l l})
},
   \nonumber
    \end{align}
    and
    \begin{align}
        \label{eq:diag-F-211-p-3}
        \begin{split}
            &\, \, 
        \sum_{s,t, l_1, \ldots, l_4}
        F_{2,1,1}(\mathbf A)_{(l_1 l_2 ) (l_3 l_4)} 
        \brac{
        k_{s l_1} (\nabla_s k_{l_2 l_3}) \otimes r_{l_4}
        } 
        \\
        = &\, \, 
       8 F_{2,1,1}(\mathbf A)_{(ss) (t t)}
        \brac{
         k_{s s} (\nabla_s k_{ t t}) \otimes  r_s 
        } +
       16 F_{2,1,1}(\mathbf A)_{(ss) (t t)}
       \brac{
         k_{s s} \nabla_s k_{s s} \otimes  r_s
       }\\
        &\, \, 
        \sum_{s,t, l_1, \ldots, l_4}
        F_{2,1,1}(\mathbf A)_{(l_1 l_2 ) (l_3 l_4)} 
        \brac{
        k_{s l_1} r_{l_2} \otimes  (\nabla_s k_{l_3 l_4}) 
        } 
        \\
        = &\, \, 
       8 F_{2,1,1}(\mathbf A)_{(ss) (t t)}
        \brac{
        k_{s s} r_s \otimes   (\nabla_s k_{ t t}) 
        } +
       16 F_{2,1,1}(\mathbf A)_{(ss) (t t)}
       \brac{
        k_{s s} r_s \otimes   \nabla_s k_{s s}
       },
        \end{split}
    \end{align}
and
\begin{align}
        \label{eq:diag-F-121}
    &\, \, 
    \sum_{s,l_1, \ldots, l_4}
    F_{1,2,1}(\mathbf A)_{(l_1 l_2) (l_3 l_4)}
    \brac{
r_{l_1} \otimes  k_{s \bar l_2} (\nabla_s k_{\bar l_3 \bar l_4})
    } \\
    = &\, \, 
8 \sum_{s,l} 
    F_{1,2,1}(\mathbf A)_{(s s) (l l)}
    \brac{
r_{s} \otimes  k_{s s} (\nabla_s k_{l l})
    } 
    + 
16 \sum_{s} 
    F_{1,2,1}(\mathbf A)_{(s s) (s s)}
    \brac{
r_{s} \otimes  k_{s s} (\nabla_s k_{s s})
    } .
   \nonumber
\end{align}
\end{lem}
\begin{lem}
    \label{lem:sym-l1-l6}
    Similar to the previous lemma, we collect terms involving 
 symmetrization over $\{l_1, \ldots, l_6\} $
 which leads to $6! = 720 = 384+96 \times  3+48$ terms as shown below:
    \begin{align}
        \label{eq:diag-F-311}
        &\, \, 
        \sum_{s,t, l_1, \ldots, l_6}
  F_{3,1,1}(\mathbf A)_{(l_1 l_2)(l_3 l_4) (l_5 l_6)}
  k_{s \bar l_1} k_{t \bar l_2} (\nabla_s k_{\bar l_3 \bar l_4})
 \otimes (\nabla_t k_{\bar l_5 \bar l_6}) \\
         = &\, \, 
 384\sum_{s} 
 F_{3,1,1}(\mathbf A)_{(ss)(ss) (ss)} 
 \brac{
 k_{ss}^2 (\nabla_s k_{ss}) \otimes  (\nabla_s k_{ss})
 } 
   \nonumber \\
         +&\, \, 
  96 \sum_{s,t} 
 F_{3,1,1}(\mathbf A)_{(ss)(t t) (t t)} 
 \brac{
 k_{ss}^2 (\nabla_s k_{t t}) \otimes  (\nabla_s k_{t t})
 }
   \nonumber
        \\ + &\, \, 
  96 \sum_{s,t} 
 F_{3,1,1}(\mathbf A)_{(ss)(ss) (t t)} 
 \brac{
 k_{ss}^2 (\nabla_s k_{ss}) \otimes  (\nabla_s k_{t t})
 + k_{ss}^2  (\nabla_s k_{t t}) \otimes  (\nabla_s k_{ss})
 } 
   \nonumber
        \\ + &\, \, 
  48 \sum_{s,t,l} 
 F_{3,1,1}(\mathbf A)_{(ss)(t t) (l l)} 
 \brac{
 k_{ss}^2 (\nabla_s k_{t t}) \otimes  (\nabla_s k_{l l})
 } ,
   \nonumber
    \end{align}
    and
\begin{align}
        \label{eq:diag-F-221}
        &\, \, 
        \sum_{s,t, l_1, \ldots, l_6}
  F_{2,2,1}(\mathbf A)_{(l_1 l_2)(l_3 l_4) (l_5 l_6)}
  k_{s \bar l_1}  (\nabla_s k_{\bar l_3 \bar l_4})
 \otimes  k_{t \bar l_2} (\nabla_t k_{\bar l_5 \bar l_6}) \\
         = &\, \, 
 384\sum_{s} 
 F_{2,2,1}(\mathbf A)_{(ss)(ss) (ss)} 
 \brac{
     k_{ss} (\nabla_s k_{ss}) \otimes k_{ss} (\nabla_s k_{ss})
 } 
 \nonumber \\
         + &\, \, 
         96 \sum_{s,t} 
 F_{2,2,1}(\mathbf A)_{(ss)(t t) (t t)} 
 \brac{
 k_{ss} (\nabla_s k_{t t}) \otimes k_{ss} (\nabla_s k_{t t})
 }
   \nonumber
        \\ + &\, \, 
  96 \sum_{s,t} 
 F_{2,2,1}(\mathbf A)_{(ss)(ss) (t t)} 
 \brac{
 k_{ss} (\nabla_s k_{ss}) \otimes k_{ss} (\nabla_s k_{t t})
 + k_{ss}  (\nabla_s k_{t t}) \otimes  k_{ss} (\nabla_s k_{ss})
 } 
   \nonumber
        \\ + &\, \, 
  48 \sum_{s,t ,l} 
 F_{2,2,1}(\mathbf A)_{(ss)(t t) (l l)} 
 \brac{
 k_{ss} (\nabla_s k_{t t}) \otimes  k_{ss} (\nabla_s k_{l l})
 } .
   \nonumber
    \end{align}  
\end{lem}

From Eq. \eqref{eq:diag-F-31}, we see that 
\begin{align}
    \label{eq:diag-v_2P-31}
    v_2(P)_{\mathbf I,3,1}  & =
\sum_{s,l}              
  32  F_{3,1}(\mathbf A)_{(s s) (l l)}
   \brac{ k_{s s}^2 (\nabla^2_{s s} k_{ l}) }
+
 64 \sum_s  F_{3,1}(\mathbf A)_{(s s) (ss)}
   \brac{ k_{s s}^2 (\nabla^2_{ss} k_{ss}) }
   \\
   &=
   (\det \mathbf A)^{-\frac{1}{2}}
   \brac{
   \sum_{s,l}
   \mathsf F_{3,1}(\mathbf z)_{s,l}
   \brac{
   \frac{k_s}{k_l} (\nabla^2_{s s} k_{ l})
   }
   + 2 \sum_{s}
   \mathsf F_{3,1}(\mathbf z)_{s}
   \brac{
   (\nabla^2_{s s} k_{ s})
   }
   }.
   \nonumber
\end{align}

By adding up Eqs. \eqref{eq:diag-v_2P-31}, \eqref{eq:diag-v_2P-21} and 
\eqref{eq:diag-v_2P-11}, we have finished the computation of 
$v_2(P)_{\mathbf{I}}$.
\begin{thm}
    \label{thm:diag-v2P-I-V1}
    We group the summands of $v_2(P)_{\mathbf{I}}$ according to the index
    \begin{align*}
        \alpha \in  \{ (1,1), (2,1) , (3,1)\} 
    \end{align*}
    in $F_{\alpha}(A)$ as below:
    \begin{align*}
        (\det \mathbf A)^{\frac{1}{2}}
        \brac{ v_2(P)_{\mathbf I,1,1} }   
      & =
    -\mathsf F_{1,1}(\mathbf A)_{\emptyset} (p_0), 
        \\
        (\det \mathbf A)^{\frac{1}{2}}
        \brac{ v_2(P)_{\mathbf I,2,1} }   
      & =
-\frac{1}{2} \sum_{s,l}
\mathsf F_{2,1}(\mathbf z)_{ l}
 \left( \frac{k_s}{k_l} \nabla^2_{ss} k_{ l} \right) 
 -  \sum_s
\mathsf F_{2,1}(\mathbf z)_{s}
\brac{
i \nabla_s r_s
},
        \\
        (\det \mathbf A)^{\frac{1}{2}}
        \brac{ v_2(P)_{\mathbf I,3,1} }   
      & =
   \sum_{s,l}
   \mathsf F_{3,1}(\mathbf z)_{s,l}
   \brac{
   \frac{k_s}{k_l} (\nabla^2_{s s} k_{ l})
   }
   + 2 \sum_{s}
   \mathsf F_{3,1}(\mathbf z)_{s}
   \brac{
   (\nabla^2_{s s} k_{ s})
   },
    \end{align*}
 where summation $s,l$ run  from $1$ to $m$.      
\end{thm}
We can reorganize the sum in terms of the differential expressions on which
$F_{\alpha}(\mathbf A)$ acts. 
\begin{thm}
    \label{thm:diag-v2P-I}
    In the diagonal case, $v_2(P)_{\mathbf{I}}$  is given by$:$
    \begin{align*}
   (\det \mathbf A)^{\frac{1}{2}}
    v_2(P)_{\mathbf I}  & =
     \sum_{s,l}
    \left(  
   \mathsf F_{3,1}(\mathbf z)_{s,l}
-\frac{1}{2} \mathsf F_{2,1}(\mathbf z)_{ l}
    \right) 
   \brac{
   \frac{k_s}{k_l} (\nabla^2_{s s} k_{ l})
   } \\
                        &+  \sum_{s}
   2 \mathsf F_{3,1}(\mathbf z)_{s}
   \brac{
   (\nabla^2_{s s} k_{ s})
   } -
   \mathsf F_{2,1}(\mathbf z)_{s}
   \brac{ i \nabla _s r_s}
\\
                & -
                \mathsf F_{1,1}(\mathbf z)_{\emptyset} (p_0) ,
    \end{align*}
 where summation $s,l$ run  from $1$ to $m$.      
\end{thm}
\begin{thm}
    \label{thm:diag-v2P-II-V1}
    We group the summands of $v_2(P)_{\mathbf{II}}$ according to the index
    \begin{align*}
        \alpha \in  \{ 
            (1,1,1), (2,1,1,),  (1,2,1), (3,1,1), (2,2,1)
\} 
    \end{align*}
    in $F_{\alpha}(A)$ as below:
\begin{align}
    \label{eq:diag-v_2P-111}
    &\, \, 
    (\det \mathbf A)^{\frac{1}{2}} 
    v_2(P)_{\mathbf{II},1,1,1} 
    \\
    =&\, \,  
    -\frac{i}{2} \sum_{s,l} \mathsf F_{1,1,1}(\mathbf z)_{s} 
    \brac{
      k_{ll}^{-1} r_s \otimes  (\nabla_s k_l) 
    }
    + \frac{1}{2} \sum_{s} \mathsf F_{1,1,1}(\mathbf z)_{s,s} 
    \brac{
     k_{ss}^{-1}    r_s \otimes  r_s
    }. \nonumber
\end{align}
and
\begin{align}
    \label{eq:diag-v_2P-211}
    &\, \, 
 (\det \mathbf A)^{\frac{1}{2}}
 \brac{
 v_2(P)_{\mathbf{II},2,1,1} 
 }\\
    =&\, \, 
  \sum_{s,l}  \mathsf F_{2,1,1}(\mathbf z)_{s , l}
    \brac{
        k_l^{-1} (\nabla_s k_s) \otimes  (\nabla_s k_l)
    } 
    +
   2 \sum_{s}   \mathsf F_{2,1,1}(\mathbf z)_{s,s}
    \brac{
        k_s^{-1} (\nabla_s k_s) \otimes  (\nabla_s k_s) }
    \nonumber
\\
    +&\, \, 
  \sum_{t,l}    \mathsf F_{2,1,1}(\mathbf z)_{t,l}
    \brac{
        \frac{k_s}{2 k_t k_l} (\nabla_s k_t) \otimes  (\nabla_s k_l)
    }
    +
 \sum_{l}       \mathsf F_{2,1,1}(\mathbf z)_{l,l}
    \brac{
        \frac{k_s}{ k^2_l} (\nabla_s k_l) \otimes  (\nabla_s k_l)
    }
    \nonumber
    \\
    +&\, \, 
   i \sum_{s,l}       \mathsf F_{2,1,1}(\mathbf z)_{s,l}
    \brac{
       \frac{1}{2} k^{-1}_l (\nabla_s k_l) \otimes  r_s
    }
+
    i \sum_{s}       \mathsf F_{2,1,1}(\mathbf z)_{s,s}
    \brac{
        k^{-1}_s (\nabla_s k_s) \otimes  r_s
    }
    \nonumber
    \\
+&\, \, 
   i \sum_{s,l}       \mathsf F_{2,1,1}(\mathbf z)_{s,l}
    \brac{
       \frac{1}{2} k^{-1}_l r_s \otimes  (\nabla_s k_l) 
    }
+
    i \sum_{s}       \mathsf F_{2,1,1}(\mathbf z)_{s,s}
    \brac{
        k^{-1}_s r_s \otimes  (\nabla_s k_s) 
    },
    \nonumber
    \nonumber
\end{align}
and
\begin{align}
    \label{eq:diag-v_2P-121}
    &\, \, 
    (\det \mathbf A)^{\frac{1}{2}}  \brac{
  v_2(P)_{\mathbf{II},1,2,1}   
    }\\  
    =&\, \, 
   i \mathsf F_{1,2,1}(\mathbf z)_{s,l}
    \brac{
    \frac{k^{(1)}_s }{ 2 k_s k_l}
       r_s \otimes  (\nabla_s k_l) 
    } 
       +
    i\mathsf F_{1,2,1}(\mathbf z)_{s,s}
    \brac{
    \frac{k^{(1)}_s }{ k_s^2 }
     r_s \otimes    (\nabla_s k_s) 
 } , \nonumber
\end{align}
and
\begin{align}
    \label{eq:diag-v_2P-311}
    &\, \, 
   - (\det \mathbf A)^{\frac{1}{2}}
 v_2(P)_{\mathbf{II},3,1,1}  \\
    =&\, \, 
    8 \sum_s
    \mathsf F_{3,1,1}(\mathbf z)_{s,s,s}
    \brac{
        k_s^{-1} (\nabla_s k_s) \otimes  (\nabla_s k_s) 
    } +
    2 \sum_{s,t}
    \mathsf F_{3,1,1}(\mathbf z)_{s,t,t}
    \brac{
       \frac{k_s}{k^2_t} (\nabla_s k_t) \otimes  (\nabla_s k_t) 
    }
    \nonumber
    \\
    +&\, \, 
2\sum_{s,t}
\mathsf F_{3,1,1}(\mathbf z)_{s,s,t}
 \sbrac{
 k_{t}^{-1} 
 \brac{
 (\nabla_s k_{s}) \otimes  (\nabla_s k_{t })
 + (\nabla_s k_{t }) \otimes  (\nabla_s k_{s})
 }
 }
    \nonumber
 \\
    +&\, \, 
    \sum_{s,t,l}
\mathsf F_{3,1,1}(\mathbf z)_{s,t,l}
\brac{
    \frac{k_s}{k_t k_l}
    (\nabla_s k_t) \otimes  (\nabla_s k_l)
}
    \nonumber
\end{align}
and
\begin{align}
    \label{eq:diag-v_2P-221}
    &\, \, 
   - (\det \mathbf A)^{\frac{1}{2}}
 v_2(P)_{\mathbf{II}, 2,2,1}  \\
    =&\, \, 
 4 \sum_{s} 
 \mathsf F_{2,2,1}(\mathbf z)_{s,s,s} 
 \brac{
     \frac{k^{(1)}_{s}}{ k_s^2} (\nabla_s k_{s}) \otimes k_{s} (\nabla_s k_{s})
 } 
  +  \sum_{s,t} 
 \mathsf F_{2,2,1}(\mathbf z)_{s,t,t} 
 \brac{
 \frac{k^{(1)}_{s}}{ k_t^2}
 (\nabla_s k_{ t}) \otimes k_{s} (\nabla_s k_{ t})
 }
    \nonumber
        \\ + &\, \, 
   \sum_{s,t} 
 \mathsf F_{2,2,1}(\mathbf z)_{s,s,t} 
 \brac{
 \frac{k^{(1)}_{s}}{ k_s k_t}
 (\nabla_s k_{s}) \otimes k_{s} (\nabla_s k_{ t})
 + k_{s}  (\nabla_s k_{ t}) \otimes  k_{s} (\nabla_s k_{s})
 } 
    \nonumber
        \\ + &\, \, 
 \frac{1}{2}  \sum_{s,t ,l} 
 \mathsf F_{2,2,1}(\mathbf z)_{s,t,l} 
 \brac{
 \frac{k^{(1)}_{s}}{ k_l k_t}
 (\nabla_s k_{ t}) \otimes  k_{s} (\nabla_s k_{ l})
 } .
    \nonumber
\end{align}
\end{thm}
\begin{proof}
    By collecting terms in Eqs. \eqref{eq:diag-F-211-p-1},  \eqref{eq:diag-F-211-p-2},
\eqref{eq:diag-F-211-p-3}, we obtain $v_2(P)_{\mathbf{II}, 2,1,1}$. For
$v_2(P)_{\mathbf{II}, 1,2,1}$, we need  \eqref{eq:diag-F-121}.
The remaining terms
$v_2(P)_{\mathbf{II}, 3,1,1}$ and $v_2(P)_{\mathbf{II}, 2,2,1}$ 
are computed in Lemma \ref{lem:sym-l1-l6}:
\end{proof}

Again, let us group the terms apropos of the differential expressions
that are indexed in the following way:
\begin{align}
 \begin{split}
  &  L^{(1)}_s = 
       (\nabla_s k_s) \otimes  (\nabla_s k_s), \, \, \,
       L^{(2)}_{s,t}   =
 (\nabla_s k_{ t}) \otimes (\nabla_s k_{ t}), \, \, \,
 L^{(3)}_{s,t} =
 (\nabla_s k_{s}) \otimes (\nabla_s k_{ t}) \\
  & L^{(4)}_{s,t} =
 (\nabla_s k_{t}) \otimes (\nabla_s k_{ s}), \, \, \,
 L^{(5)}_{s,t,l} =
 (\nabla_s k_{ t}) \otimes  (\nabla_s k_{ l})
     \end{split}
    \label{eq:L1-L5}
\end{align}
and
\begin{align}
    \label{eq:L6-L10}
    \begin{split}
        & L^{(6)}_s = i(\nabla_s k_s) \otimes  r_s, \, \, \,
     L^{(7)}_{s,l} = i (\nabla_s k_l) \otimes  r_s \\
    & 
     L^{(8)}_{s} =i r_s \otimes  (\nabla_s k_s) ,\, \, \,
     L^{(9)}_{s,l} =i r_s \otimes  (\nabla_s k_l) ,\, \, \,
     L^{(10)}_{s} = r_s \otimes  r_s  .\, \, \,
    \end{split}
\end{align}
We also make use of the substitution in Eq. \eqref{eq:A-Y-Z-defn-diag} 
$k_s^{(1)} = k_s \mathbf y_s^{(1)} = k_s (1- \mathbf z_s^{(1)})$ 
to move all $k$-factors to the very left. 
\begin{thm}
    \label{thm:diag-v2P-II}
    In the diagonal case, $v_2(P)_{\mathbf{II}}$  consists of two parts$:$
\begin{align}
    \label{eq:diag-v_2P-L1-L5}
    &\, \, 
    (\det \mathbf A)^{\frac{1}{2}}
   v_2(P)_{\mathbf{II}, L^{(1)}, \ldots, L^{(5)} }  
   \\
    =&\, \, 
 \sum_s  k_s^{-1} 
 \brac{
    - 8\mathsf F_{3,1,1}(\mathbf z)_{s,s,s}
    - 4 (1- \mathbf z_s^{(1)}) \mathsf F_{2,2,1}(\mathbf z)_{s,s,s}
    + 2\mathsf F_{2,1,1}(\mathbf z)_{s,s}
 }
 \brac{ L^{(1)}_s}
 \nonumber
   \\
    + &\, \, 
\sum_{s,t}  \frac{k_s}{k_t^2}   
 \brac{
     -2\mathsf F_{3,1,1}(\mathbf z)_{s,t,t}
    -  (1- \mathbf z_s^{(1)}) \mathsf F_{2,2,1}(\mathbf z)_{s,t,t}
     + \mathsf F_{2,1,1}(\mathbf z)_{t,t}
 }
 \brac{ L^{(2)}_{s,t}}
 \nonumber
 \\
    + &\, \, 
\sum_{s,t}  k_t^{-1}
 \brac{
     -2\mathsf F_{3,1,1}(\mathbf z)_{s,s,t}
    -  (1- \mathbf z_s^{(1)}) \mathsf F_{2,2,1}(\mathbf z)_{s,s,t}
     + \mathsf F_{2,1,1}(\mathbf z)_{s,t}
 }
 \brac{ L^{(3)}_{s,t}}
 \nonumber
\\
    + &\, \, 
\sum_{s,t}  k_t^{-1}
 \brac{
     -2\mathsf F_{3,1,1}(\mathbf z)_{s,s,t}
    -  (1- \mathbf z_s^{(1)}) \mathsf F_{2,2,1}(\mathbf z)_{s,s,t}
 }
 \brac{ L^{(4)}_{s,t}}
 \nonumber
 \\
    + &\, \, 
\sum_{s,t,l}  \frac{k_s}{k_t k_l}
 \brac{
    - \mathsf F_{3,1,1}(\mathbf z)_{s,l,t}
    - \frac{1}{2}
    (1- \mathbf z_s^{(1)}) \mathsf F_{2,2,1}(\mathbf z)_{s,l,t}
      + \frac{1}{2} \mathsf F_{2,1,1}(\mathbf z)_{l,t}
 }
 \brac{ L^{(5)}_{s,l,t}}
 \nonumber
\end{align}
and contribution from $p_1(\xi)$ $($the symbol of first order term of $P$ in
Eq. \eqref{eq:P-gen-form-defn}$)$ is
grouped as below$:$ 
\begin{align}
    \label{eq:diag-v_2P-L6-L10}
    &\, \,
    (\det \mathbf A)^{\frac{1}{2}}
   v_2(P)_{\mathbf{II}, L^{(6)}, \ldots, L^{(10)} }  
   \\
    =&\, \, 
 \sum_s  k_s^{-1} 
    \mathsf F_{2,1,1}(\mathbf z)_{s,s}
 \brac{ L^{(6)}_s}
 +
 \frac{1}{2}  \sum_{s,l}  k_l^{-1} 
    \mathsf F_{2,1,1}(\mathbf z)_{s,s}
 \brac{ L^{(7)}_{s,l}} +
\frac{1}{2}  \sum_{s}  k_s^{-1} 
    \mathsf F_{1,1,1}(\mathbf z)_{s,s}
 \brac{ L^{(10)}_{s}} 
 \nonumber
 \\
    +&\, \, 
    \sum_s  
   k_s^{-1}   \brac{
   \mathsf F_{2,1,1}(\mathbf z)_{s,s} +
   (1-\mathbf z_s^{(1)})\mathsf F_{1,2,1}(\mathbf z)_{s,s}
    }
 \brac{ L^{(8)}_s}
 \nonumber
 \\
    +&\, \, 
    \sum_{s,l}  
  \frac{1}{2} k_l^{-1}   \brac{
   \mathsf F_{2,1,1}(\mathbf z)_{s,l} +
   (1-\mathbf z_s^{(1)})\mathsf F_{1,2,1}(\mathbf z)_{s,l}
   - 
   \mathsf F_{1,1,1}(\mathbf z)_{l} 
    }
 \brac{ L^{(9)}_{s,l}}
 .
 \nonumber
\end{align}

\end{thm}

\section{Examples in Conformal Geometry}
\label{sec:eg-from-conformal-geo}
\subsection{Notations}
Let us test our results obtained in the previous sections, 
specially  \S\ref{subsec:diagonal-case}, on the conformal geometry of
noncommutative tori.
Based on the study  on $\mathbb{T}_\theta^2$ , we simply take 
\begin{align*}
    \Delta = \sum_{j=1}^m \delta_j^2 
    \to  \Delta_\varphi  \defeq k^{1 / 2} \Delta k^{1 / 2}
\end{align*}
as a simplified model on $\mathbb{T}_\theta^m$ for a conformal change of the
flat metric (represented by
the flat Laplacian $\Delta$), where $k = e^h$, $ h = h^* \in
C^\infty(\mathbb{T}^m_\theta)$  is the Weyl factor.
In \cite{MR3194491}, $\varphi$ denotes a   
rescaling of the canonical trace $\varphi_0$: 
$\varphi(a) = \varphi_0(a e^{-h})$, for all $a \in
C^\infty(\T^2_\theta)$ which plays the role of the  volume functional of the
new metric.
The perturbed Laplacian $\Delta_\varphi = k^{\frac{1}{2}} \Delta
k^{\frac{1}{2}}$  appeared as
the degree zero part of the new Dolbeault Laplacian
\footnote{$\bar \partial^*_{\varphi}$ denotes the adjoint of the $\bar
\partial$-operator taking with respect to the weight $\varphi$.} 
$ \bar \partial^*_{\varphi}\bar \partial$.



In this case, the metrics are parametrized by only one noncommutative
coordinate $k = e^h$. 
As before, we put
 $\mathbf x \defeq \mathbf x_h$
(resp. $\mathbf y \defeq \mathbf y_h$):
\begin{align}
    \mathbf x_h = [\bullet, h], \, \, \, 
    \mathbf y_h = e^{\mathbf x_h} = k^{-1}(\bullet ) k
    : C^\infty(\T^m_\theta) \to  C^\infty(\T^m_\theta)
    \label{eq:modop-der-conformal}
\end{align}
and the partial versions $\mathbf x^{(l)}, \mathbf y^{(l)} \in L(
C^\infty(\mathbb{T}^m_\theta)^{\otimes  n}, C^\infty(\mathbb{T}^m_\theta))$, 
$l=1, \ldots,n$.  
The reduction process for the rearrangement operators 
\begin{align*}
    F_{\alpha}(\mathbf A)_{(l_1l_2) \cdots (l_{2N-1} l_{2N})} 
    \xrightarrow{
    \mathbf A =\op{diag}(k_1, \ldots, k_m)
    }
    \mathsf F_\alpha(\mathbf A)_{s_1, \ldots, s_N} 
    \xrightarrow{k_1 = \cdots =k_m = k}
    \mathsf H_\alpha(\mathbf z^{(1)}, \ldots, \mathbf z^{(n)})
\end{align*}
has been explained in \S\ref{subsec:the-conformal-case-Halpha}, where 
$\mathbf z= (\mathbf z^{(1)}, \ldots, \mathbf z^{(n)})$ 
is a special case of  the change of variable 
in \cref{eq:A-Y-Z-defn-diag}:  
\begin{align}
    \label{eq:bfz-vs-bfy-conformal}
    \mathbf z^{(l)} = 1- \mathbf y^{(1)} \cdots \mathbf y^{(l)}. \, \, \,
\end{align}

The spectral functions $\mathsf H_\alpha(z_1, \ldots, z_n)$ differs from 
$H_\alpha(z_1, \ldots, z_n)$ derived in \cite{Liu:2018aa,Liu:2018ab} by
a Gamma-factor (see Eq. \eqref{eq:Hnow-vs-Hbefore}). 
For the differential expressions listed in \cref{eq:L1-L5,eq:L6-L10}, 
all of them are reduced to one of the two types below (after
summing over $s$ or $s,t$):
\begin{align}
    \label{eq:Tr-nabla^2k-and-nablak}
    \Tr \brac{ \nabla^2 k}  \defeq \sum_{s=1}^m \nabla^2_s k = - \Delta k,
    \, \, \,
    \Tr \brac{(\nabla k) (\nabla k)} \defeq \sum_{s=1}^m
(\nabla_s k) \otimes  (\nabla_s k).
\end{align}

In \cite{Liu:2015aa,LIU2017138} and \cite{Liu:2018aa,Liu:2018ab},
a conjugation trick was applied which yields substantial
simplification to the computation. Namely, one starts with another Laplacian
\begin{align*}
    \Delta_k \defeq k\Delta  = k^{\frac{1}{2}} \Delta_\varphi k^{-\frac{1}{2}}
    = \mathbf y^{-\frac{1}{2}} \brac{ \Delta_\varphi  }.
\end{align*}
Their heat coefficients are related in a similar way.
\begin{lem}
    \label{lem:-Delphi-vs-Delk}
    The two sets of heat coefficients agree upto a conjugation of the Weyl factor: 
\begin{align}
    \label{eq:v_2-Delphi-vs-Delk}
    v_j(\Delta_\varphi) = \mathbf y^{\frac{1}{2}} \brac{
        v_j(\Delta_k)
    }, \, \, \, j=0,1,2, \ldots
    .
\end{align}
\end{lem}
\begin{proof}
    It follows from the fact that the modular operator $\mathbf y$  commutes
    with the exponential. More precisely,
we have 
$\forall  f \in  C^\infty(\T^m_\theta)$,
\begin{align*}
    \Tr(f e^{-t\Delta_\varphi}) &= \Tr\brac{
        f \mathbf y^{\frac{1}{2}} \brac{
        e^{-t\Delta_k}
        }
    } = \Tr \brac{
    \mathbf y^{-\frac{1}{2}}(f)  e^{-t\Delta_k}
} 
.
\end{align*}
We see, by comparing the asymptotic expansion of two sides, that for
$j=0,1,2,\ldots$,  
\begin{align*}
\varphi_0 \brac{  f v_j(\Delta_\varphi)} =
\varphi_0 \brac{ \mathbf y^{-\frac{1}{2}}(f) v_j(\Delta_k)}
= \varphi_0 \brac{ f \mathbf y^{\frac{1}{2}} \brac{v_j(\Delta_k)} }.
\end{align*}
\end{proof}
The Laplacian $\Delta_k$ is much easier to handle since the symbol contains
only the leading term. As we can see later, \S\ref{subsec:ExII-Delvfi}
consists of the exact extra work if one attacks the heat asymptotic of
$\Delta_\varphi$ directly. 

\subsection{Example I: $\Delta_k = k \Delta$ }
\label{subsec:ExI-Delk}

For $\Delta_k = k \Delta$, the symbol $\sigma(\Delta_k) = p_2$ contains only
the leading term:
\begin{align*}
    p_2^{\Delta_k}(\xi) = k\abs\xi^{2},
    p_1^{\Delta_k}(\xi) = p^{\Delta_k}_0(\xi) =0.
\end{align*}

\begin{prop}
    \label{prop:v2-Delk}
    The second heat coefficient
    consists of two parts  $v_2(\Delta_k) = v_2(\Delta_k)_{\mathbf{I}}
    + v_2(\Delta_k)_{\mathbf{II}} $$:$ 
    \begin{align*}
        v_2(\Delta_k)_{\mathbf{I}} &= k^{-\frac{m}{2}}
        G_{\Delta_k, \mathbf{I}}(\mathbf z) (\Tr (\nabla^2 k)),
\\
        v_2(\Delta_k)_{\mathbf{II}} &=
        v_2(\Delta_k)_{\mathbf{II}, L^{(1)}, \ldots, L^{(5)} }   
        =  k^{-\frac{m}{2} -1}
        G_{\Delta_k, \mathbf{II}}(\mathbf z^{(1)} , \mathbf z^{(2)} )
        \brac{
    \Tr( (\nabla k) (\nabla k)) 
        }    ,
    \end{align*}
    where operators $\mathbf z$ and $\mathbf z^{(1)} , \mathbf z^{(2)} $ are
    defined in \cref{eq:bfz-vs-bfy-conformal} and the spectral functions are
    given in terms of hypergeometric functions as below:
    \begin{align}
        \label{eq:Gdelk-I}
 G_{\Delta_k, \mathbf{I}}(z) = 
 (m+2) \mathsf H_{3,1}(z;m) - \frac{m}{2} \mathsf H_{2,1}(z;m),
    \end{align}
   and with $\vec z =(z_1, z_2)$, 
    \begin{align}
        \label{eq:Gdelk-II}
        G_{\Delta_k, \mathbf{II}} (\vec z) & =- (m^2 + 6m +8)
         \brac{ 
            \mathsf H_{3,1,1}(\vec z;m) + \frac{1}{2} 
            (1-z_1) \mathsf H_{2,2,1}(\vec z;m) 
        }\\
                                           &+
  \frac{ (m^2+4m+4)}{2}
\mathsf H_{2,1,1}(\vec z;m) .  
\nonumber                               
    \end{align}
\end{prop}
\begin{proof}
We shall apply  reduction rules like:
\begin{align*}
    k_s \to  k,  \, \, \,
    \mathsf F_{\alpha}(\mathbf z)_{s,t,l} \to  \mathsf H_\alpha(z)
    ,
\end{align*}
to all terms appeared in \Cref{thm:diag-v2P-I,thm:diag-v2P-II}.
For instance, a summand in \cref{eq:diag-v_2P-311} works out as below:
\begin{align*}
    \sum_{s,t,l}
    \frac{k_s}{k_t k_l} \mathsf F_{3,1,1}(\mathbf z)_{s,l,t}
    \brac{
        (\nabla _s k_t ) \otimes  (\nabla_s k_l)
    } \to  m^2 k^{-1}
    \mathsf H_{3,1,1}(\vec{\mathbf z};m) \left( 
    \Tr( (\nabla k) (\nabla k))
    \right), 
\end{align*}
where $\vec{\mathbf z} = (\mathbf z^{(1)}, \mathbf z^{(2)})$.  
Notice that summing over $s$ yields $\Tr( (\nabla
k) (\nabla k)) $ while summing over $t,l$ produces $m^2$ copies of the same term.
By repeating the process, we obtain the reduction of 
\cref{eq:diag-v_2P-311,eq:diag-v_2P-221,eq:diag-v_2P-211}:
\begin{align*}
    & v_2(\Delta_k)_{\mathbf{II}, 3,1,1} \to  
    - \brac{
    m^2 + 6m +8
    } 
    \mathsf H_{3,1,1}(\vec{\mathbf z};m) \left( 
    \Tr( (\nabla k) (\nabla k))
    \right) \\ 
& v_2(\Delta_k)_{\mathbf{II}, 2,2,1} \to  
    - \frac{1}{2} \brac{
    m^2 + 6m +8
    } 
    (1-\mathbf z^{(1)}) \mathsf H_{2,2,1}(\vec{\mathbf z};m) \left( 
    \Tr( (\nabla k) (\nabla k))
    \right) \\
& v_2(\Delta_k)_{\mathbf{II}, 2,1,1} \to  
     \frac{1}{2} \brac{
    m^2 + 4m +4
    } 
    \mathsf H_{2,1,1}(\vec{\mathbf z};m) \left( 
    \Tr( (\nabla k) (\nabla k))
    \right) .
\end{align*}
They constitute all the non-zero contributes to $v_2(\Delta_k)_{\mathbf{II}}$,
hence we have proved \cref{eq:Gdelk-II}. For \cref{eq:Gdelk-I}, calculation is
similar, terms in \Cref{thm:diag-v2P-I-V1} are turned into:
\begin{align*}
& v_2(\Delta_k)_{\mathbf{I}, 2,1} \to  
\frac{1}{2} m \mathsf H_{2,1} (\mathbf z;m) 
\brac{\Tr(\nabla^2 k)},
\\
& v_2(\Delta_k)_{\mathbf{I}, 3,1} \to  
(m+2) \mathsf H_{3,1}(\mathbf z;m)
\brac{\Tr(\nabla^2 k)}.
\end{align*}
\end{proof}
In order to  recover the exact formulas in \cite{Liu:2018aa,Liu:2018ab},
we recall constants  from Eqs. \eqref{eq:Hnow-vs-Hbefore} and
\eqref{eq:overall-factor-before}:
\begin{align*}
    \frac{\mathsf H_{3,1,1} (\vec z;m) }{ H_{3,1,1} (\vec z;m) }
    =\frac{\mathsf H_{2,2,1} (\vec z;m) }{ H_{2,2,1} (\vec z;m) } =
    C_m \frac{1}{
        (\frac{m}{2} +2) (\frac{m}{2} +1) \frac{m}{2}
    } = C_m
    \frac{8}{ (m+4) (m+2) m}.
\end{align*}
and
\begin{align*}
    \frac{\mathsf H_{2,1,1} (\vec z;m) }{ H_{2,1,1} (\vec z;m) } =
    C_m \frac{1}{
         (\frac{m}{2} +1) \frac{m}{2}
    } 
    = C_m 
    \frac{4}{  (m+2) m}.
\end{align*}
Therefore:
\begin{align*}
    C_m^{-1} 
    G_{\Delta_k, \mathbf{II}}(z) & =  
       -\frac{8}{m}  \brac{ 
            H_{3,1,1}(\vec z;m) + \frac{1}{2} 
            (1-z_1) H_{2,2,1}(\vec z;m) 
        }\\
  &+ \brac{ \frac{4}{m} + 2}
H_{2,1,1}(\vec z;m)  .                                
\end{align*}
The right hand side is exactly $H_{\Delta_k}(\vec z;m)$ given in \cite[Prop.
4.1]{Liu:2018ab}.
Similarly, we have  
\begin{align*}
    \frac{\mathsf H_{3,1} (\vec z;m) }{ H_{3,1} (\vec z;m) }
    = C_m \frac{4}{(m+2) m} 
    ,
    \, \, \,
    \frac{\mathsf H_{2,1} (\vec z;m) }{ H_{2,1} (\vec z;m) }
    = C_m \frac{2}{ m} 
\end{align*}
thus 
\begin{align*}
    C_m^{-1} 
    G_{\Delta_k, \mathbf{I}}(z) & =  - H_{2,1} (z;m)
    +  \frac{4}{m} H_{3,1}(z;m),
\end{align*}
which equals $K_{\Delta_k}(z;m)$ in \cite[Prop. 4.1]{Liu:2018ab}.

\subsection{Example II: $\Delta_\varphi = k^{1 /2} \Delta k^{1 /2}$  }
\label{subsec:ExII-Delvfi}
We start with Lemma \ref{lem:delk1/2-to-delk},
the ``chain rule'' with with respect to the
change of variable $k^{\frac{1}{2}} \to  k$ and then compute   
the symbol of $\Delta_\varphi$ in Lemma \ref{lem:Del_vfi-sym}. 
\begin{lem}
    \label{lem:delk1/2-to-delk}
   Recall 
  \begin{align*}
      \delta_s(k^{\frac{1}{2}}) & 
      = k^{ - 1 / 2}
     \mathsf G_{\mathrm{pow}}^{(1)}(\mathbf y;1/2)     \brac{\delta_s(k)} , \\
      \delta_s^{2} (k^{\frac{1}{2}}) & = 
    k^{ - 1 / 2}
  \mathsf G_{\mathrm{pow}}^{(1)}(\mathbf y;1/2)     \brac{\delta_s^2(k)} 
      + 2 k^{ - 3 / 2}
      \mathsf G_{\mathrm{pow}}^{(1,1)}(\mathbf y^{(1)}, \mathbf y^{(2)};1/2)
      \brac{\delta_s (k) \otimes  \delta_s (k)},
  \end{align*}  
  where $\mathsf G_{\mathrm{pow}}$ is obtained by applying divided differences
  \footnote{
  We remind the reader the divided difference notion. For a single-variable
  function $f(x)$, the $n$-th divided difference $f[x_0, \ldots, x_n]$  is
  a function of $n+1$-variable:   
  \begin{align*}
      f[x_0, \ldots, x_n] = \sum_{j=0}^n
      \frac{f(x_j)}{ \prod_{l\neq j}(x_j - x_j)}
  \end{align*}
  }
 to the power functions $u^j$ with $j \in  \R$$:$ 
  \begin{align*}
      \mathsf G_{\mathrm{pow}}^{(1)} (y; j )
&= 
      u^{j}[1,y] = \frac{y^{j}-1}{y-1}
      ,\\
      \mathsf G_{\mathrm{pow}}^{(1,1)} (y_1, y_2; j ) 
&= u^{j }[1, y_1, y_1 y_2] =
\mathsf G_{\mathrm{pow}}^{(1)} (u; j )
[y_1, y_1 y_2]
=
\frac{
\mathsf G_{\mathrm{pow}}^{(1)} (y_1 y_2; j ) -
\mathsf G_{\mathrm{pow}}^{(1)} (y_1; j )
}{y_1 y_2 - y_1}
      .
  \end{align*}
  We will fix $j=1/2$ from now on and freely use the abbreviations
   $\mathsf
  G_{\mathrm{pow}}^{(1)} (y) \defeq \mathsf G_{\mathrm{pow}}^{(1)}(y; 1 /2) $  and 
  $\mathsf G_{\mathrm{pow}}^{(1,1)} (y_1 , y_2) = 
  \mathsf G_{\mathrm{pow}}^{(1,1)} (y_1 , y_2 ; 1 /2 ) $ in the rest of the
  computation.
\end{lem}
\begin{proof}
  See \cite[Lemma 2.13]{Liu:2018ab}. 
\end{proof}
\begin{lem}
    \label{lem:Del_vfi-sym}
    The symbol of $\Delta_\varphi = k^{\frac{1}{2}} \Delta k^{\frac{1}{2}}$ is
    given by
    \begin{align*}
        \pmb\sigma( \Delta_\varphi) = k \abs\xi^2 + 
        2 k^{\frac{1}{2}} \delta_s(k^{\frac{1}{2}}) \xi_s +
        k^{\frac{1}{2}} \delta_s^2( k^{\frac{1}{2}}).
    \end{align*}
    Following the notations in Eq. \eqref{eq:P-gen-form-defn}, we have:
    the leading term agrees with the previous one, 
    $p_2^{\Delta_\varphi} = p_2^{\Delta_k}$, and the linear term 
    \begin{align*}
       p_1^{\Delta_\varphi}(\xi) = \sum_{s=1}^m 
       2 k^{\frac{1}{2}} \delta_s(k^{\frac{1}{2}}) \xi_s,
       \, \, \, \text{with} \, \, \,
r_s^{\Delta_\varphi}(\xi) =
2 k^{\frac{1}{2}} \delta_s(k^{\frac{1}{2}}),
    \end{align*}
    and the constant $($in $\xi$$)$ term $
    p_0^{\Delta_\varphi} 
  = \sum_{s=1}^m k^{\frac{1}{2}} \delta_s^2( k^{\frac{1}{2}}).  $ 
\end{lem}
\begin{rem}
    By taking Lemma \ref{lem:delk1/2-to-delk}  into account and
    $\delta_s =-i \nabla_s $, also  Eq. \eqref{eq:Tr-nabla^2k-and-nablak},
    we can rewrite every thing in terms of $\nabla  k$ and $\nabla^2 k$:  
    \begin{align}
        \label{eq:r_s-delphi}
         r_s^{\Delta_\varphi}(\xi) 
= -2i 
 k^{\frac{1}{2}} \nabla_s(k^{\frac{1}{2}}) =
 -2i \mathsf G_{\mathrm{pow}}^{(1)} (\mathbf y ) (\nabla_s  k),
    \, \, \, s=1, \ldots, m.
    \end{align}
    and 
\begin{align}
        \label{eq:p_0-delphi}
  p_0^{\Delta_\varphi} 
   = -\sum_{s=1}^m k^{\frac{1}{2}} \nabla_s^2( k^{\frac{1}{2}})  
  =  
  \mathsf G_{\mathrm{pow}}^{(1)} (\mathbf y )
  \brac{
      \Tr( \nabla^2 k)
  } + k^{-1} 
  2 \mathsf G_{\mathrm{pow}}^{(1,1)} (\mathbf y_1 , \mathbf y_2 )
  \brac{ \Tr( (\nabla  k) (\nabla  k) )}.
\end{align}
\end{rem}
\begin{proof}
    Notice that $\forall a \in  C^\infty(\T^m_\theta) $, we have $[\delta_s,a]
    = \delta_s(a)$, thus:  
    \begin{align*}
        [\delta_s^2, k^{\frac{1}{2}}] 
        = \delta_s[\delta_s,k^{\frac{1}{2}}] + [\delta_s,k^{\frac{1}{2}}] \delta_s 
        =[\delta_s, [\delta_s, k^{\frac{1}{2}}]] + 
        2 [\delta_s,k^{\frac{1}{2}}] \delta_s       
         = \delta_s^2(k^{\frac{1}{2}}) +
         2 \delta_s( k^{\frac{1}{2}}) \delta_s.
    \end{align*}
We are ready to put $\Delta_\varphi$ into the form of Eq. 
\eqref{eq:P-gen-form-defn} with the help of commutators:
    \begin{align*}
        \Delta_\varphi &= k \Delta + k^{\frac{1}{2}} [\Delta, k^{\frac{1}{2}} ]
        = k   \Delta 
        + k^{\frac{1}{2}}
        \sum_{s=1}^m [\delta_s^2, k^{\frac{1}{2}}] \\
&= k \sum_{s=1}^m \delta_s^2 +
 \sum_{s=1}^m  \brac{ 2 \delta_s(k^{\frac{1}{2}}) \delta_s + 
                       \delta_s^2(k^{\frac{1}{2}})}.
    \end{align*}
  The symbol follows immediately by replacing the differential operators 
  $\delta_s \to  \xi_s$.   
    Another method is to make use of the $\star$-product. Notice that 
    $D^j \abs\xi^2 =0$ for all $j>2$. 
    There are only
    three non-zero terms for $\abs\xi^2 \star k^{\frac{1}{2}}$ which are given
    in Eq. \eqref{eq:a_0-to-a_2}: 
\begin{align*}
    \pmb\sigma(\Delta k^{\frac{1}{2}}) = \sum_{j=0}^2
    a_j(\abs\xi^2, k^{\frac{1}{2}}) = 
    k^{\frac{1}{2}} \abs\xi^2  -2 i \sum_1^m \nabla_s(k^{\frac{1}{2}}) \xi_s  -
    \sum_1^m \nabla^2_s(k^{\frac{1}{2}}) .
\end{align*}
\end{proof}
Here comes the key result.
\begin{prop}
    \label{prop:v2-Delvfi}
    The second heat coefficient of $\Delta_\varphi$  is of the form
    $v_2(\Delta_\varphi) = v_2(\Delta_\varphi)_{\mathbf{I}}
    + v_2(\Delta_\varphi)_{\mathbf{II}}$ with
    \begin{align*}
v_2(\Delta_\varphi)_{\mathbf{I}} =
        G_{\Delta_\varphi, \mathbf{I}}(\mathbf z)( \Tr(\nabla^2 k)),
        \, \, \,
v_2(\Delta_\varphi)_{\mathbf{II}} =
        G_{\Delta_\varphi, \mathbf{II}}(\mathbf z^{(1)}, \mathbf z^{(2)})
        \brac{
            \Tr( (\nabla k) (\nabla  k))
        },
    \end{align*}
where the two traces are defined in Eq.  \eqref{eq:Tr-nabla^2k-and-nablak}
and the spectral functions are give by:
\begin{align}
    \label{eq:Gdelvphi-I}
    G_{\Delta_\varphi, \mathbf{I}} (z)
    &= G_{\Delta_k, \mathbf{I}} (z) + J_{\mathbf{I},p_1,p_0}^{(2)} (z),\\
    G_{\Delta_\varphi, \mathbf{II}} (z_1 , z_2)
    &= G_{\Delta_k, \mathbf{II}} (z_1 , z_2)
    + J_{\mathbf{I},p_1,p_0}^{(1,1)}(z_1 , z_2)
    +J_{\mathbf{II},L^{(6)}, \ldots, L^{(10)}}(z_1 , z_2),
    \label{eq:Gdelvphi-II}
\end{align}
where $G_{\Delta_k, \mathbf{I}}$ and $G_{\Delta_k, \mathbf{II}}$  
is defined in \textup{Proposition} $\ref{prop:v2-Delk}$. The remaining terms
are also spanned by the hypergeometric family $\mathsf H_{\alpha}(\bar z;m)$ in 
\eqref{eq:sfHalpha} and $\mathsf G^{(1)}_{\mathrm{pow}}$,
$\mathsf G^{(1,1)}_{\mathrm{pow}}$  in 
\textup{Lemma \ref{lem:delk1/2-to-delk}:}
    \begin{align}
        \label{eq:JIp_1p_0-(1)}
        J^{(2)}_{\mathbf{I}, p_1,p_0} (z) =
       -  \mathsf G^{(1)}_{\mathrm{pow}}(y) \brac{
            2 \mathsf H_{2,1}(z;m) - \mathsf H_{1,1}(z;m)
        },
    \end{align}
    and
    \begin{align}
        \label{eq:JIIp_1p_0-(1-1)}
        J^{(1,1)}_{\mathbf{I}, p_1,p_0} (z_1 ,z_2)   
          &= 
        -2
\brac{ 
    y_1^{-\frac{1}{2}}
    \mathsf G^{(1)}_{\mathrm{pow}}(y_1)
    \mathsf G^{(1)}_{\mathrm{pow}}(y_2) 
    + 
    2 \mathsf G^{(1,1)}_{\mathrm{pow}}(y_1 , y_2) 
} 
\mathsf H_{2,1}(z_2;m) \\
          &+
    2 \mathsf G^{(1,1)}_{\mathrm{pow}}(y_1 , y_2) 
\mathsf H_{1,1}(z_2;m) ,
\nonumber
    \end{align}
    and
\begin{align}
    \label{eq:JIIL_6L_10}
       &\, \,  J_{\mathbf{II},L^{(6)}, \ldots, L^{(10)}} (z_1, z_2)      \\
        =&\, \, 
        \sbrac{
         (2+m)
         \brac{
         \mathsf G_{\mathrm{pow}}^{(1)}(y_1) + 
         \mathsf G_{\mathrm{pow}}^{(1)}(y_2) 
         }
        } 
        \mathsf H_{2,1,1}(z_1, z_2;m) 
        +
        (2+m)
        \mathsf G_{\mathrm{pow}}^{(1)}(y_1)
        (1-z_1)
        \mathsf H_{1,2,1} (z_1, z_2;m) 
        \nonumber \\
        -&\, \, 
        \brac{
        m \mathsf G_{\mathrm{pow}}^{(1)}(y_1)
         +2 \mathsf G_{\mathrm{pow}}^{(1)}(y_1)
         \mathsf G_{\mathrm{pow}}^{(1)}(y_2)
        }
        \mathsf H_{1,1,1} (z_1, z_2;m) .
        \nonumber 
    \end{align}
\end{prop}
\begin{proof}
Since $p_2^{\Delta_\varphi} = p_2^{\Delta_k}$, the
contribution  from the leading term is identical to those 
(i.e., $G_{\Delta_k, \mathbf{I}}$ and $G_{\Delta_k, \mathbf{II}}$)  obtained in
Proposition \ref{prop:v2-Delk}.
It remains to count the terms involving $p_1^{\Delta_\varphi}$ and
$p_0^{\Delta_\varphi}$ in
Theorem \ref{thm:diag-v2P-I} and  \ref{thm:diag-v2P-II},
which are 
$ -\sum_s \mathsf F_{2,1} (\mathbf z)_{(ss)} (i \nabla_s r_s) 
- \mathsf F_{1,1}(\mathbf z)_{\emptyset} (p_0) $ and
$ v_2(\Delta_\varphi)_{\mathbf{II},L^{(6)}, \ldots, L^{(10)}} $. 
We packed the tedious computation into the proof of Lemma \ref{lem:JIp1p0} and
\ref{lem:JIIL6-L10} at the end of this section.  
\end{proof}

We see that the contribution from $p_1^{\Delta_\varphi}$  and
$p_0^{\Delta_\varphi}$ is proportional to that of the leading term.
\begin{cor}
    \label{cor:Gdelk-vs-Gdelvfi-fun-relations}
 The following relations hold$:$   
    \begin{align}
        \label{eq:fn-relation-I}
        (y^{\frac{1}{2}}-1) G_{\Delta_k, \mathbf{I}}(z)
        &= J^{(2)}_{\mathbf{I}, p_1,p_0}(z),
        \\
        ((y_1 y_2)^{\frac{1}{2}}-1) G_{\Delta_k, \mathbf{II}}(z_1 , z_2)
        &= J^{(1,1)}_{\mathbf{II}, p_1,p_0}(z_1 , z_2)
        +J_{\mathbf{II},L^{(6)}, \ldots, L^{(10)}} (z_1, z_2),
        \label{eq:fn-relation-II}
    \end{align}
    where $z= 1-y$ and $z_1 =1-y_1$, $z_2 = 1-y_1 y_2$ as in 
    \cref{eq:A-Y-Z-defn-diag}.  
\end{cor}
\begin{proof}
    In the previous proposition, we obtain the first version of of
    $G_{\Delta_\varphi, \mathbf{I}}$ 
    and $G_{\Delta_\varphi, \mathbf{II}}$ (cf. Eqs \eqref{eq:Gdelvphi-I} and
    \eqref{eq:Gdelvphi-II}) making use of the general results 
(\Cref{thm:diag-v2P-I-V1,thm:diag-v2P-II-V1}).
On the other hand, 
    Lemma \ref{lem:-Delphi-vs-Delk} states that   
$v_2(\Delta_\varphi) = \mathbf y^{\frac{1}{2}}\brac{v_2(\Delta_k)}$,
which implies that their spectral functions agree upto a factor of 
$y^{\frac{1}{2}} = (1-z)^{\frac{1}{2}}$ (or $ (y_1 y_2)^{\frac{1}{2}}
= (1-z_2)^{\frac{1}{2}}$):
\begin{align*}
    G_{\Delta_\varphi, \mathbf{I}} (z) = (1-z)^{\frac{1}{2}}  
    G_{\Delta_k, \mathbf{I}} (z), \, \, \, 
    G_{\Delta_\varphi, \mathbf{II}} (z_1, z_2) = 
    (1-z_2)^{\frac{1}{2}} G_{\Delta_k, \mathbf{II}} (z_1 , z_2)
    .
\end{align*}
 The relations  follow immediately from comparison.
\end{proof}

\begin{lem}
    \label{lem:JIp1p0}
    With  $r_s^{\Delta_\varphi} = -2 i k^{\frac{1}{2}} 
    \nabla_s(k^{\frac{1}{2}})$ and $p_0^{\Delta_\varphi} = - k^{\frac{1}{2}}
      \sum_{l=1}^m \nabla^2_l (k^{\frac{1}{2}})$, 
the following sum in Theorem $\ref{thm:diag-v2P-I-V1}$ becomes$:$
    \begin{align*}
  &\, \, 
  -\sum_s \mathsf F_{2,1} (\mathbf z)_{(ss)} (i \nabla_s r_s) 
        - \mathsf F_{1,1}(\mathbf z)_{\emptyset} (p_0) \\ 
        =&\, \, \,
        J^{(2)}_{\mathbf{I}, p_1,p_0} (\mathbf z) (\Tr(\nabla^2 k))  +
   k^{-1} J^{(1,1)}_{\mathbf{I}, p_1,p_0}(\mathbf z^{(1)} , \mathbf z^{ (2) })
(\Tr((\nabla k) (\nabla  k))) 
    \end{align*}
    where
    \begin{align*}
        J^{(2)}_{\mathbf{I}, p_1,p_0} (z) =
       -  \mathsf G^{(1)}_{\mathrm{pow}}(y) \brac{
            2 \mathsf H_{2,1}(z;m) - \mathsf H_{1,1}(z;m)
        }
    \end{align*}
    and
    \begin{align*}
        J^{(1,1)}_{\mathbf{I}, p_1,p_0} (z_1 ,z_2)   
          &= 
        -2
\brac{ 
    y_1^{-\frac{1}{2}}
    \mathsf G^{(1)}_{\mathrm{pow}}(y_1)
    \mathsf G^{(1)}_{\mathrm{pow}}(y_2) 
    + 
    2 \mathsf G^{(1,1)}_{\mathrm{pow}}(y_1 , y_2) 
} 
\mathsf H_{2,1}(z_2;m) \\
          &+
    2 \mathsf G^{(1,1)}_{\mathrm{pow}}(y_1 , y_2) 
\mathsf H_{1,1}(z_2;m) .
    \end{align*}

\end{lem}
\begin{proof}
%

By replacing $\mathsf F_{a,b} (\mathbf z)_{(ss)} $ with $\mathsf H_{a,b}
(\mathbf z)$ and plugging in $r_1^{\Delta_\varphi}$ and
$ p_0^{\Delta_\varphi}$, we see that   
\begin{align*}
    &\, \, 
    \sum_s \mathsf F_{2,1} (\mathbf z)_{(ss)} (i \nabla_s r_s) 
        + \mathsf F_{1,1}(\mathbf z)_{\emptyset} (p_0) 
       \\ = &\, \,  
        \brac{
            2 \mathsf H_{2,1}(\mathbf z;m) -
\mathsf H_{1,1}(\mathbf z;m)
        } 
        \brac{
            \sum_s   k^{\frac{1}{2}} \nabla_s^2 k^{\frac{1}{2}}
        } +
2 \mathsf H_{2,1}(\mathbf z;m) 
\brac{
    (\nabla_s k^{\frac{1}{2}})^2
}.
\end{align*}
The results follow from the changing the derivatives (according to Lemma
\ref{lem:delk1/2-to-delk})  $\nabla_s k^{\frac{1}{2}}$ and
$\nabla_s^2 k^{\frac{1}{2}}$ to $\nabla_s k$ and $\nabla_s^2 k$:
\begin{align*}
        ( \nabla_s k^{\frac{1}{2}})^2 
  &= \brac{
      k^{-\frac{1}{2}}
  \mathsf G^{(1)}_{\mathrm{pow}}(\mathbf y)
  (\nabla_s k)
  }^2
  =k^{-1}
(\mathbf y^{(1)})^{-\frac{1}{2}}
    \mathsf G^{(1)}_{\mathrm{pow}}(\mathbf y^{(1)})
    \mathsf G^{(1)}_{\mathrm{pow}}(\mathbf y^{(2)}) 
         \brac{ \nabla_s k \otimes  \nabla_s k },
\\
           k^{\frac{1}{2}}   \nabla_s^2 k^{\frac{1}{2}} 
  &= k^{-1}
\sbrac{
    \mathsf G^{(1)}_{\mathrm{pow}}(\mathbf y)
  \brac{\nabla_s^2 k}
   + 2 \mathsf G^{(1,1)}_{\mathrm{pow}}(\mathbf y^{(1)}, \mathbf y^{(2)}) 
  }
  \brac{(\nabla_s k) \otimes  (\nabla_s k)}.
    \end{align*}
\end{proof}

\begin{lem}
    \label{lem:JIIL6-L10}
In Theorem $\ref{thm:diag-v2P-II}$, terms involving $p_1^{\Delta_\varphi}(\xi)$
give rise to the following contribution in the conformal case$:$
    \begin{align*}
        v_2(\Delta_\varphi)_{\mathbf{II},L^{(6)}, \ldots, L^{(10)}} =
        J_{\mathbf{II},L^{(6)}, \ldots, L^{(10)}} 
        (\mathbf z^{(1)}, \mathbf z^{(2)})
    \brac{
        \Tr((\nabla  k) ( \nabla  k ))
    }
    \end{align*}
    where the two variable function
    \begin{align*}
       &\, \,  J_{\mathbf{II},L^{(6)}, \ldots, L^{(10)}} (z_1, z_2)      \\
        =&\, \, 
        \sbrac{
         (2+m)
         \brac{
         \mathsf G_{\mathrm{pow}}^{(1)}(y_1) + 
         \mathsf G_{\mathrm{pow}}^{(1)}(y_2) 
         }
        } 
        \mathsf H_{2,1,1}(z_1, z_2;m) 
        +
        (2+m)
        \mathsf G_{\mathrm{pow}}^{(1)}(y_1)
        (1-z_1)
        \mathsf H_{1,2,1} (z_1, z_2;m) \\
        -&\, \, 
        \brac{
        m \mathsf G_{\mathrm{pow}}^{(1)}(y_1)
         +2 \mathsf G_{\mathrm{pow}}^{(1)}(y_1)
         \mathsf G_{\mathrm{pow}}^{(1)}(y_2)
        }
        \mathsf H_{1,1,1} (z_1, z_2;m).
    \end{align*}
\end{lem}
\begin{proof}
    With $ r_s^{\Delta_\varphi}  = -2 i 
 \mathsf G_{\mathrm{pow}}^{(1)}(\mathbf y) (\nabla_s k) $, 
 we first turn  $L_s^{(6)}$  to $L_s^{(10)}$ in \cref{eq:L6-L10}
 into $\Tr((\nabla k) (\nabla  k))$: 
    \begin{align*}
  &      \sum_{s=1}^m  L_s^{(10)} = - 4 
 \mathsf G_{\mathrm{pow}}^{(1)}(\mathbf y^{(1)}) 
 \mathsf G_{\mathrm{pow}}^{(1)}(\mathbf y^{(2)}) 
  \Tr((\nabla k) (\nabla  k)) \\
  &      \sum_{s=1}^m  L_s^{(6)} = 2
 \mathsf G_{\mathrm{pow}}^{(1)}(\mathbf y^{(1)}) 
\Tr((\nabla k) (\nabla  k)), \, \, \,
        \sum_{s,l=1}^m  L_s^{(7)} = 2 m
 \mathsf G_{\mathrm{pow}}^{(1)}(\mathbf y^{(1)}) 
\Tr((\nabla k) (\nabla  k)), \, \, \, \\
&
        \sum_{s=1}^m  L_s^{(8)} = 2 
 \mathsf G_{\mathrm{pow}}^{(1)}(\mathbf y^{(2)}) 
\Tr((\nabla k) (\nabla  k)), \, \, \,
        \sum_{s,l=1}^m  L_s^{(9)} = 2 m
 \mathsf G_{\mathrm{pow}}^{(1)}(\mathbf y^{(2)}) 
\Tr((\nabla k) (\nabla  k)). \, \, \,
    \end{align*}
Then we repeat reduction the process  used in the proof of \cref{prop:v2-Delk}
to complete the calculation.    
\end{proof}

\section{Verification of the Functional Relations}
\label{sec:Verification}
Last but not least, we provide some conceptual validation for the our lengthy
computation by  carefully examining 
the relations in
\cref{eq:fn-relation-I,eq:fn-relation-II} bases on their explicit expressions 
given in terms of   $\mathsf G_{\mathrm{pow}}^{(1)}$, 
$\mathsf G_{\mathrm{pow}}^{(1,1)}$ 
and the hypergeometric family $\mathsf H_\alpha(\bar z;m)$
(cf.  
\cref{eq:Gdelk-I,eq:Gdelk-II,eq:JIp_1p_0-(1),eq:JIIp_1p_0-(1-1),eq:JIIL_6L_10}).

\subsection{Preparations}
\label{subsec:preparation}

First of all, let us switch to the hypergeometric family $H_\alpha(\bar z;m)$
used in in \cite[\S5]{Liu:2018ab}:
\begin{align*}
    \mathsf H_\alpha(\bar z;m) = 
    \frac{1}{\Gamma \brac{ \sum_{0}^n\alpha_j + \frac{m}{2} -2}}
    H_\alpha(\bar z ;m), 
    \, \, \alpha = (\alpha_0, \ldots, \alpha_n),
\end{align*}
so that we can freely quote 
formulas listed there  without modifications.
The variables $\{z, z_1 ,z_2\}$ always denote the following change of
variable of  $\{y, y_1,y_2\}$:
\begin{align*}
    z = 1-y, \, \,  z_1 = 1- y_1, \, \,  z_2 = 1- y_1 y_2.
\end{align*}
For example,
\begin{align*}
    H_{a,b}(z;m) = H_{a,b}(1-y;m), \, \, 
    H_{a,b,c}(z_1, z_2;m) = H_{a,b,c}(1-y_1, 1- y_1 y_2;m).  
\end{align*}
We often drop the arguments and write $H_{a,b}$ or $H_{a,b,c}$ when no
confusion arises.

The two cyclic transformations are crucial:
\footnote{
In the setting of \cite[\S2]{Liu:2018ab}, 
    the cyclic transformations arise from integration by parts with respect to
    the modular operator. }
    \begin{align}
    \label{eq:tau1-and-2-defn}
        \begin{split}
            &   \tau_{(1)}: C(\R_+)\to  C(\R_+), 
        \, \, \tau_{(1)}(f) (y) = f(y^{-1})
\\
     &      \tau_{(2)}: C(\R_+^2 )\to  C(\R_+^{2}), \, \,       
     \tau_{(2)}(\tilde f) (y_1, y_2) = \tilde f( (y_1y_2)^{-1}, y_1)
        \end{split}
    \end{align}
defined on functions in $y$ and in $(y_1 ,y_2)$ respectively.  
We see  immediately that $\tau_{(1)}^2 =1$. For $\tau_{(2)}$, 
we have  $\tau_{(2)}$ and $\tau_{(2)}^2$  are implemented by the substitutions:
\begin{align}
    \label{eq:tau2-and-tau2-sqr-defn}
    y_1 \xrightarrow[]{\tau_{(2)}} (y_1 y_2)^{-1}
    \xrightarrow[]{\tau_{(2)}^2} y_2, 
    \, \, \, \, 
    y_2 \xrightarrow[]{\tau_{(2)}} y_1 
    \xrightarrow[]{\tau_{(2)}^2} (y_1 y_2)^{-1}, 
\, \, \, \, 
y_1 y_2 \xrightarrow[]{\tau_{(2)}} y_2^{-1}
\xrightarrow[]{\tau_{(2)}^2} y_1^{-1}
,
\end{align}
and $\tau_{(2)}^3 =1$. 
Notice that, $\tau_{(1)}$ (resp. $\tau_{(2)}$) becomes linear fractional
transformation with respect to $z$ (resp. $z_1, z_2$).
The key feature  of the cyclic permutations is the fact that 
(cf. \cite[Prop. 5.1]{Liu:2018ab}) 
 they permute the components of $\alpha$ in $H_{\alpha}(\bar z;m)$:
 \begin{align}
     \label{eq:tau1-on-Hab}
     \tau_{(1)}(H_{a,b}) &= (1-z)^{a+b+m /2 -2}H_{b,a}, 
     \\
     \label{eq:tau2-on-Habc}
     \tau_{(2)}(H_{a,b,c}) &=
     (1-z_2)^{a+b+c + m /2 -2}H_{b,c,a} .
 \end{align}
 We will also need  the formulas of $\tau_{(2)}^2$: 
 \begin{align}
     \label{eq:tau2-sqr-on-Habc}
     \tau_{(2)}^2 (H_{a,b,c})(z_1 , z_2;m) &=
     (1-z_1)^{a+b+c + m /2 -2}
     H_{c,a,b} (z_1 , z_2;m),
\\
     \label{eq:tau2-sqr-on-Hab}
\tau_{(2)}^2 \brac{
    H_{a,b}(z_2;m)
 } 
& =  (1-z_1)^{a+b + m /2 -2}
H_{b,a}(z_1;m).
 \end{align}
Note that \cref{eq:tau2-sqr-on-Hab} follows from \cref{eq:tau1-on-Hab} 
with the substitutions in \cref{eq:tau2-and-tau2-sqr-defn}:
\begin{align*}
 \tau_{(2)}^2 \brac{
    H_{a,b}( z_2 ;m)
} &=
 \tau_{(2)}^2 \brac{
    H_{a,b}(1-y_1y_2;m)
 }   
 = H_{a,b}(1-y_1^{-1}  ;m) \\
  & = \tau_{(1)}(H_{a,b})(z_1;m)
  =(1-z_1)^{a+b+\frac{m}{2 } -1} 
  H_{b,a}(z_1;m) .
\end{align*}


We now turn to recurrence relations. The goal is to
express $H_{a,b,c}$  and $H_{a,b}(z_1;m)$ in terms of 
$\{H_{1,2,1}, H_{1,1,2}, H_{1,1,1}\}$. 
For two variable functions, we need:
\begin{align}
        \label{eq:remove-H_131}
    H_{1,3,1}&= 
 \frac{ ((m+6) (1-y_1)-6) H_{1,2,1}  
 +(m+2) H_{1,1,1}+2 (1-y_1 y_2) \left(H_{1,1,2}-y_1 H_{1,2,2}\right)}
 {4 (1-y_1) y_1}   ,
            \\
    H_{1,2,2} &= \frac{H_{1,2,1}-H_{1,1,2}}{z_1-z_2}.
        \label{eq:remove-H122}
\end{align}
Therefore, $\forall  \eta_{1,1,3}, \eta_{1,2,2} \in  \mathbb{C}$:
\begin{align}
    \label{eq:H122-H113-coef-remove}
    \eta_{1,1,3} H_{1,1,3} + \eta_{1,2,2} H_{1,2,2} &= 
    \sbrac{
    \brac{
        \frac{m}{4 y_1} + \frac{1}{y_1 -1} - \frac{1}{2} \frac{1}{y_1 (y_2 -1)}
    } \eta_{1,1,3} 
    + \frac{\eta_{1,2,2}}{y_1(y_2-1)}
    }
    H_{1,2,1}\\
                                                    & + 
                                                    \sbrac{
        \frac{-1}{ y_1 (y_2 -1)} \eta_{1,2,2} 
        + \frac{y_2 (y_1 y_2 -1)}{2 y_1 (y_1-1) (y_2 -1)} \eta_{1,1,3}
    } H_{1,1,2}   
                                                    \nonumber
    \\
                                                    &+ \sbrac{
 \frac{-(2+m)}{4y_1(y_1-1)} \eta_{1,1,3}                                              
                                                    } H_{1,1,1}.
                                                    \nonumber
\end{align}
The connection between one and two variable families is given by the divided
difference operation:
\begin{align*}
    H_{a,1,1}(z_1, z_2;m) = (z H_{a+1,1}(z;m))[z_1 , z_2]_z 
    = \frac{z_1 H_{a+1,1}(z_1;m) - z_2 H_{a+1,1}(z_2;m) }{z_1 - z_2}.
\end{align*}
Set $a=1$ and apply $\partial_{z_1}$  or $\partial_{z_2}$ on both sides, we see that
\begin{align*}
    H_{1,2,1} = \partial_{z_1} H_{1,1,1} = 
  \frac{H_{1,2}\left(z_1;m\right)-H_{1,1,1}}{z_1-z_2}, 
  \, \, \, \, 
H_{1,1,2} = \partial_{z_1} H_{1,1,1} = 
  \frac{H_{1,2}\left(z_2;m\right)-H_{1,1,1}}{z_2-z_1}.  
\end{align*}
As a consequence, we obtain
\begin{align}
    \label{eq:H12-z_1-z_2-to-Habc}
    H_{1,2}(z_1;m) =  (z_1 - z_2) H_{1,2,1}+H_{1,1,1}
    ,\, \, \, \, 
    H_{1,2}(z_2;m) =  (z_2 - z_1) H_{1,1,2}+H_{1,1,1} .
\end{align}
To get $H_{2,1}(z_1;m)$, notice that  
\begin{align*}
    \tau_{(2)}^2 \brac{ H_{1,2}(z_2;m)} & = 
    \tau_{(2)}^2 \brac{ H_{1,2}(1 - y_1y_2;m)} = H_{1,2}(1- y_1^{-1} ;m)
     \\
                                        & = \tau_{(1)}(H_{1,2})(z_1) 
     = y_1^{1+\frac{m}{2}} H_{2,1}(z_1 ;m) 
    ,    
\end{align*}
thus
\begin{align*}
    H_{2,1}(z_1 ;m) &= 
    y_1^{-1 - \frac{m}{2}} 
    \tau_{(2)}^2 \brac{ H_{1,2}(z_2;m)} = 
    y_1^{-1 - \frac{m}{2}} 
    \tau_{(2)}^2 \brac{
    (z_2 - z_1) H_{1,1,2}+H_{1,1,1}
    } \\
  &= (y_1 y_2 -1) y_1 H_{2,1,1}  + H_{1,1,1},
\end{align*}
where we have used \cref{eq:tau2-on-Habc,eq:tau2-and-tau2-sqr-defn} to reach
the last equal sign. Next, we replace $H_{2,1,1}$  according to 
\begin{align*}
    H_{2,1,1} = \frac{m+2}{2} H_{1,1,1} - y_1 y_2 H_{1,1,2} - y_1 H_{1,2,1}.
\end{align*}
The final result reads:
\begin{align}
    \label{eq:H_21-as-Habc}
    H_{2,1}(z_1;m) =  \left[
        1- \frac{m+2}{2} z_2 
\right] H_{1,1,1}
    +z_2 (1-z_2) 
    H_{1,1,2}
    +z_2 (1-z_1) H_{1,2,1}
    .
\end{align}
Sum up,
\begin{align}
    \label{eq:H_21-H_12-z_1-as-Habc}
    \brac{
    \eta_{2,1} H_{2,1} + \eta_{1,2} H_{1,2}
}(z_1;m) & =
y_1 
\sbrac{
    (y_2 -1) \eta_{1,2} - (y_1 y_2 -1) \eta_{2,1}
}
H_{1,2,1} \\
         & +
\sbrac{
    y_1y_2 (1 -y_1y_2) \eta_{2,1}
}
H_{1,1,2} 
\nonumber
\\
         & +
\sbrac{
    \eta_{1,2} + \frac{1}{2} \brac{
        m(y_1y_2 -1) + 2 y_1y_2 
    } \eta_{2,1}
}
H_{1,1,1}   .
\nonumber
\end{align}

\subsection{Verification I}
Let us look at the following function
\begin{align*}
    V_{\mathrm{I}}(z) = -z \brac{ \frac{4}{m} H_{3,1}(z;m) - H_{2,1}}
    +\frac{2}{m} H_{2,1}(z;m) - H_{1,1}(z;m),
\end{align*}
which is obtained by taking 
the difference of the two sides of \cref{eq:fn-relation-I} 
with the substitution 
\begin{align*}
  \mathsf H_{a,b} \to  H_{a,b}/ \Gamma(a+b+ m / 2-2).  
\end{align*}
A common factor $\Gamma(m / 2)$ is dropped since we only care about
$V_{\mathrm{I}} =0 $ or not. 
 It is easier to show that $\tau_{(1)}(V_{\mathrm{I}}) =0$ where the cyclic
 transformation $\tau_{(1)}$ is given in \cref{eq:tau1-and-2-defn}. Indeed,
 according to \cref{eq:tau1-on-Hab,eq:tau1-and-2-defn}, we have
\begin{align}
    \label{eq:sigma-VI=0}
    (1-z)^{-m /2} \tau_{(1)}
    (V_{\mathrm{I}})(z) =
    \frac{4z}{m} (1-z) H_{1,3}(z;m) + 
    \sbrac{-z +\frac{4}{m} (1-z)} H_{1,2}(z;m) - H_{1,1}(z;m), 
\end{align}
where the right hand side vanishes because it is 
the hypergeometric ODE of $H_{1,1}(z;m)$ (cf. \cite{Liu:2018ab}).

\subsection{Verification II}
\label{subsec:verification-II}


The verification of \cref{eq:fn-relation-II} is much more involved compared to
the one-variable case in previous section.
As before, we begin with change of notations $\mathsf H_{\alpha} \to  H_\alpha$,
the functions 
in \cref{eq:Gdelk-II,eq:JIIp_1p_0-(1-1),eq:JIIL_6L_10} are turned into 
\begin{align}
    \label{eq:sf-JII-defn}
      \mathsf J_{\mathrm{II}}  &=
\sbrac{
     \brac{
         \mathsf G_{\mathrm{pow}}^{(1)} (y_1) 
         + \mathsf G_{\mathrm{pow}}^{(1)} (y_2)
     } H_{2,1,1}
     + 
        y_1 \mathsf G_{\mathrm{pow}}^{(1)} (y_1)
     H_{1,2,1}
     - \frac{1}{2}
     \mathsf G_{\mathrm{pow}}^{(1)} (y_1) \brac{
     m+ 2
         \mathsf G_{\mathrm{pow}}^{(1)} (y_2)
     }
     H_{1,1,1}
     } ,
     \\
\label{eq:sf-GII-defn}
      \mathsf G_{\mathrm{II}}  &=
        \sbrac{ (y_1 y_2)^{\frac{1}{2}} -1}   
        \brac{
        \frac{2+m}{2} H_{2,1,1} - 2 H_{3,1,1} - y_1 H_{2,2,1}
        } ,
        \\
     \label{eq:sf-JI-defn}
      \mathsf J_{\mathrm{I}}  &=
      -\sbrac{
          2 \mathsf G_{\mathrm{pow}}^{(1,1)}(y_1,y_2) +
          y_1^{-\frac{1}{2}} 
          \mathsf G_{\mathrm{pow}}^{(1)}(y_1)
          \mathsf G_{\mathrm{pow}}^{(1)}(y_2)
      } H_{2,1} (z_2;m) 
      +
           \frac{m}{2} 
          \mathsf G_{\mathrm{pow}}^{(1,1)}(y_1,y_2)  
      H_{1,1} (z_2;m),
\end{align}
where we have suppressed the arguments of the functions $\mathsf J_{\mathrm{II}}$,
$\mathsf J_{\mathrm{I}}$, $\mathsf G_{\mathrm{II}}$  and $H_{a,b,c}$, for example  
$\mathsf J_{\mathrm{II}} \defeq \mathsf J_{\mathrm{II}}(z_1,z_2;m)$.
A common factor has been factored out, that is
\begin{align*}
    \mathsf J_{\mathrm{II}} = C_m
     J_{\mathbf{II}, L^{(6)}, \ldots, L^{10}}, 
     \, \, \, \,   
 \frac{\mathsf G_{\mathrm{II}}}{(y_1 y_2)^{\frac{1}{2}} -1}
=
 C_m G_{\Delta_k, \mathbf{II}} , \, \, \, \, 
\mathsf J_{\mathrm{I}} =
C_m  J^{1,1}_{\mathbf{I}, p_1, p_0} ,
\end{align*}
where $ C_m = \Gamma(m / 2 +1)/ 2$. 
We apply  
the cyclic transformation $\tau_{(2)}^2$ (cf. \cref{eq:tau2-and-tau2-sqr-defn})
on both sides of \cref{eq:fn-relation-II} so that all the hypergeometric pieces 
$H_{a,b,c}$, $H_{a,b}$ appeared in \cref{eq:sf-JI-defn,eq:sf-GII-defn,eq:sf-JII-defn}
start with $1$ in the subscripts, that is, of the form
$H_{1,a,b}$ or $H_{1,a}$.   
\begin{prop}
 Keep the notations as above, we have   
    \begin{align}
        \tau_{(2)}^2 \brac{
            \mathsf G_{\mathrm{II}} -
            \mathsf J_{\mathrm{II}}
        } =
        \tau_{(2)}^2 \brac{ \mathsf J_{\mathrm{I}} }
        .
        \label{eq:GII-JII=JI}
\end{align}
In particular, \cref{eq:fn-relation-II} holds as well.
\end{prop}
\begin{proof}
 The verification is arranged as follows.
First, we express the two sides of \cref{eq:GII-JII=JI} as combinations of
$\{H_{1,2,1}, H_{1,1,2}, H_{1,1,1}\} $. The calculation and the results are is
postponed to \cref{lem:msf-GII-JII-to-Habc,lem:msf-JI-to-Habc} respectively.
The rest is to show that the corresponding coefficients are indeed equal.
We will take advantage of of the fact that $\mathsf G_{\mathrm{pow}}^{(1)}$  
and $\mathsf G_{\mathrm{pow}}^{(1,1)}$ are divided difference of the power
function $u^j$ (here $j= 1 / 2$, cf. \cref{lem:delk1/2-to-delk}).  
Recall that 
\begin{enumerate}
    \item A divided difference $f[x_0, \ldots, x_{n}]$ is symmetric in its
        $n+1$ arguments.
    \item If the function $f$ is homogeneous of degree $j \in  R$, namely $f(c
        y) = c^j f(y)$, $\forall c>0$, then the $n$-th divided difference is of
        homogeneity $j-n$, that is
        $
        f[c x_0, \ldots, c x_{n} ] = c^{j-n} f[x_0, \ldots, x_{n}]  
        $. 
\end{enumerate}
  Therefore, we have the following identities that are repeatedly applied in
  the proof:
  \begin{align}
      \label{eq:inproof-Gpow1-y-vs-yinvs}
      \mathsf G_{\mathrm{pow}}^{(1)} (y^{-1}) =
      u^{\frac{1}{2}}[1, y^{-1}]_u = 
      y^{\frac{1}{2}} 
      u^{\frac{1}{2}}[y, 1]_u = 
      y^{\frac{1}{2}} 
      \mathsf G_{\mathrm{pow}}^{(1)} ( y ) 
  \end{align}
  and
   \begin{align}
      \label{eq:inproof-Gpow11-two-denominators}
       &\, \, \, 
       \mathsf G_{\mathrm{pow}}^{(1,1)} (y_2, (y_1y_2)^{-1} ) = 
       u^{\frac{1}{2}}[1, y_2, y_2 (y_1y_2)^{-1}]_u =
       y_1^{-\frac{1}{2} +2}  u^{\frac{1}{2}}[y_1, y_2 y_1,1]_u
       \\
       =&\, \, 
y_1^{\frac{3}{2} }
       \frac{
       u^{\frac{1}{2}}[y_1y_2, y_1]_u -
       u^{\frac{1}{2}}[ 1, y_1]_u 
       }{y_1 y_2 -1}
       =
y_1^{\frac{3}{2} }
       \frac{
u^{\frac{1}{2}}[y_1y_2, 1]_u -
       u^{\frac{1}{2}}[ 1, y_1]_u 
       }{y_1 y_2 - y_1}
       ,
       \nonumber
   \end{align}
   in which
\begin{align*}
 &      u^{\frac{1}{2}}[ 1, y_1]_u =
    \mathsf G_{\mathrm{pow}}^{(1)} (y_1),  
       \, \, \, \, 
       u^{\frac{1}{2}}[ 1, y_1 y_2]_u = 
    \mathsf G_{\mathrm{pow}}^{(1)} (y_1 y_2) , 
       \\
 &       
       u^{\frac{1}{2}}[y_1y_2, y_1]_u =
y_1^{-\frac{1}{2}}
       u^{\frac{1}{2}}[ y_2, 1 ]_u =
y_1^{-\frac{1}{2}}
   \mathsf G_{\mathrm{pow}}^{(1)} (y_2) . 
\end{align*}

  Now we are ready to simplify the functions $ \tilde c_{1,1,1}$  $ \tilde
  c_{1,1,2}$ and $ \tilde c_{1,2,1}$ defined in
   \cref{eq:til-c111-inproof,eq:til-c112-inproof,eq:til-c121-inproof} to the
   corresponding coefficients given in \cref{lem:msf-GII-JII-to-Habc}, 
   starting with $ \tilde c_{1,1,2}$: 
\begin{align*}
    \tilde c_{1,1,2} =
      y_1 y_2 (1- y_1y_2) 
\mathsf G_{\mathrm{pow}}^{(1,1)} (y_2, (y_1y_2)^{-1}) =
y_1^2 y_2 \brac{
    y_1^{\frac{1}{2}}
   \mathsf G_{\mathrm{pow}}^{(1)} (y_1) - 
   \mathsf G_{\mathrm{pow}}^{(1)} (y_2)  
},
\end{align*}
which is equal to the coefficient of $H_{1,1,2}$  in \cref{lem:msf-GII-JII-to-Habc}.
Notice that \cref{eq:inproof-Gpow11-two-denominators} yields different
denominators for $ \mathsf G_{\mathrm{pow}}^{(1,1)} (y_2, (y_1y_2)^{-1} )$ to 
achieving cancellation.

Up to a minus sign, $ \tilde c_{1,2,1}$ in \cref{eq:til-c121-inproof} consists of
the following terms:
\begin{align}
 \label{eq:1-summands-of-til-c121-inproof}
 &   \sbrac{
  y_1 (y_1y_2 -1) + y_1^2 (y_2 -1) 
}
 \mathsf G_{\mathrm{pow}}^{(1,1)} (y_2, (y_1y_2)^{-1} )  =
 y_1^2 \sbrac{
  \mathsf G_{\mathrm{pow}}^{(1)} (y_2) - 
  2 y_1^{\frac{1}{2}}
  \mathsf G_{\mathrm{pow}}^{(1)} (y_1) + 
  y_1^{\frac{1}{2}}
   \mathsf G_{\mathrm{pow}}^{(1)} (y_1 y_2)  
 },
 \\
 &
y_1^2 (y_2 -1) y_2^{-\frac{1}{2}}
\mathsf G_{\mathrm{pow}}^{(1)} (y_2)
\mathsf G_{\mathrm{pow}}^{(1)} ((y_1y_2)^{-1}) =
y_1^2 \sbrac{
  (y_1 y_2)^{\frac{1}{2}}
   \mathsf G_{\mathrm{pow}}^{(1)} (y_1 y_2)  
   - 
  y_1^{\frac{1}{2}}
   \mathsf G_{\mathrm{pow}}^{(1)} (y_1 y_2)  
 },
 \label{eq:2-summands-of-til-c121-inproof}
\end{align}
so that  the sum reads:
\begin{align*}
   - \tilde c_{1,2,1}  = y_1^2
    \sbrac{
   -2   \mathsf G_{\mathrm{pow}}^{(1)} (y_1^{-1} )  
  + \mathsf G_{\mathrm{pow}}^{(1)} ( y_2 )   
  + (y_1y_2)^{\frac{1}{2}}
  \mathsf G_{\mathrm{pow}}^{(1)} \brac{ y_1 y_2 } 
    }.
\end{align*}
In the calculation of \cref{eq:2-summands-of-til-c121-inproof}, we have  used 
$ (y_2-1) \mathsf G_{\mathrm{pow}}^{(1)} ( y_2 ) = y_2^{\frac{1}{2}} -1$ 
and the relation
$
(y_1y_2)^{\frac{1}{2}}
  \mathsf G_{\mathrm{pow}}^{(1)} \brac{ y_1 y_2 } 
=\mathsf G_{\mathrm{pow}}^{(1)} \brac{ (y_1 y_2)^{-1} }
$. The relation also implies that 
$ \tilde c_{1,2,1}$  is indeed equal to the corresponding
coefficients in \cref{lem:msf-GII-JII-to-Habc}. 

By adding up the three terms in \cref{eq:til-c111-inproof}:
\begin{align*}
    & y_1 (y_2-1) 
\mathsf G_{\mathrm{pow}}^{(1,1)} (y_2, (y_1y_2)^{-1}) =
y_1 \brac{
y_1^{\frac{1}{2}}
   \mathsf G_{\mathrm{pow}}^{(1)} (y_1 y_2)  -
y_1^{\frac{1}{2}}
   \mathsf G_{\mathrm{pow}}^{(1)} (y_1 )  
},
\\
    & \frac{m}{2} (y_2y_1 -1)
\mathsf G_{\mathrm{pow}}^{(1,1)} (y_2,(y_1y_2)^{-1}) =
\frac{m}{2} y_1 \brac{
\mathsf G_{\mathrm{pow}}^{(1)} (y_2 )   -
y_1^{\frac{1}{2}} 
   \mathsf G_{\mathrm{pow}}^{(1)} (y_1 )  
} ,
\\
    & -
y_1 y_2^{-\frac{1}{2}}
\mathsf G_{\mathrm{pow}}^{(1)} (y_2) 
\mathsf G_{\mathrm{pow}}^{(1)} ((y_1y_2)^{-1}) =
 - y_1^{\frac{3}{2}}
\mathsf G_{\mathrm{pow}}^{(1)} (y_2) 
\mathsf G_{\mathrm{pow}}^{(1)} (y_1 y_2) 
\end{align*}
we have for the last one:
\begin{align*}
    \tilde c_{1,1,1} = y_1 \sbrac{
    \frac{m}{2}
\mathsf G_{\mathrm{pow}}^{(1)} (y_2) 
    -(1+\frac{m}{2}) y_1^{\frac{1}{2}}
\mathsf G_{\mathrm{pow}}^{(1)} (y_1) +
y_1^{\frac{1}{2}}
\mathsf G_{\mathrm{pow}}^{(1)} (y_1 y_2) 
\brac{
1- \mathsf G_{\mathrm{pow}}^{(1)} (y_2) 
}
    },
\end{align*}
in which the third term can also be written as
\begin{align*}
    y_1^{\frac{1}{2}}
\mathsf G_{\mathrm{pow}}^{(1)} (y_1 y_2) 
\brac{
    1- \mathsf G_{\mathrm{pow}}^{(1)} (y_2) 
} = (y_1y_2)^{\frac{1}{2}}
\mathsf G_{\mathrm{pow}}^{(1)} (y_1 y_2) 
\mathsf G_{\mathrm{pow}}^{(1)} (y_2)
=
\mathsf G_{\mathrm{pow}}^{(1)} \brac{ (y_1 y_2)^{-1} } 
\mathsf G_{\mathrm{pow}}^{(1)} (y_2). 
\end{align*}
Therefore $ \tilde c_{1,1,1}$  agrees with the coefficient of $H_{1,1,1}$  in
\cref{lem:msf-GII-JII-to-Habc}.
\end{proof}

  \begin{lem}
        \label{lem:msf-GII-JII-to-Habc}
      The function
 $\tau_{(2)}^2 \brac{ \mathsf G_{\mathrm{ II }} - \mathsf J_{\mathrm{ II }} }$
 belongs to the span of $\{H_{1,2,1}, H_{1,1,2}, H_{1,1,1}\} $ 
 with coefficients given by:
      \begin{align}
          \label{eq:-sig-2-sf-GII-minus-JII}
          &\, \, 
       y_1^{-\frac{m}{2}}   \tau_{(2)}^2 \brac{
         \mathsf G_{\mathrm{ II }} - \mathsf J_{\mathrm{ II }}
     } (z_1 ,z_2;m)
     \\
          = &\, \, 
          y_1^2 \sbrac{
          2 y_1^{\frac{1}{2}} 
          \mathsf G_{\mathrm{pow}}^{(1)} (y_1) -
          \mathsf G_{\mathrm{pow}}^{(1)} (y_2) -
          \mathsf G_{\mathrm{pow}}^{(1)} \brac{ (y_1 y_2)^{-1} }
          }
     H_{1,2,1} 
     +
     y_1^2 y_2 \sbrac{
           y_1^{\frac{1}{2}} 
          \mathsf G_{\mathrm{pow}}^{(1)} (y_1) -
     \mathsf G_{\mathrm{pow}}^{(1)} (y_2) 
     }
     H_{1,1,2} 
     \nonumber
     \\
          +&\, \, 
   y_1  \sbrac{
        - \brac{ 1+\frac{m}{2} } y_1^{\frac{1}{2}}
          \mathsf G_{\mathrm{pow}}^{(1)} (y_1) 
         + \frac{1}{2}
         \mathsf G_{\mathrm{pow}}^{(1)} (y_2) \brac{
         m+2
\mathsf G_{\mathrm{pow}}^{(1)} \brac{ (y_1 y_2)^{-1} }
         } 
     }
     H_{1,1,1}.
     \nonumber
      \end{align}
  \end{lem} 
  \begin{proof}
      We first follow    \cref{eq:tau2-and-tau2-sqr-defn,eq:tau2-on-Habc} to
      carry out the transformation $\tau_{(2)}^2$:
    \begin{align*}
     y_1^{-\frac{m}{2}}  
     \tau^2_{(2)}(\mathsf G_{\mathrm{ II }}) (z_1,z_2;m)=
     \tau_{(2)}^2 \brac{
         (y_1y_2)^{\frac{1}{2}} -1
     }
     \brac{
      \frac{m+2}{2} y_1^2 H_{1,2,1}
       - y_2 y_1^3 H_{1,2,2}- 2 y_1^3 H_{1,3,1}
     },
   \end{align*}  
   where 
   \begin{align*}
      \tau_{(2)}^2 \brac{
         (y_1y_2)^{\frac{1}{2}} -1
     }  = \tau_{(2)}^2 \brac{
     \mathsf G^{(1)}_{\mathrm{pow}}( y_1 y_2) (y_1 y_2 -1)
} = 
\mathsf G^{(1)}_{\mathrm{pow}} (y_1^{-1}) y_1^{-1}(1-y_1)
   \end{align*}
   and
\begin{align}
    \label{eq:tau-sqr-JII-as-Habc-inproof}
     y_1^{-\frac{m}{2}}  
     \tau^2_{(2)}(\mathsf J_{\mathrm{ II }}) (z_1,z_2;m)
     & =
     \sbrac{
       \mathsf G^{(1)}_{\mathrm{pow}}(y_2) +
       \mathsf G^{(1)}_{\mathrm{pow}} \brac{ (y_1 y_2)^{-1} } 
     }
y_1^2 
H_{1,2,1}  
+ y_2 y_1^2 
       \mathsf G^{(1)}_{\mathrm{pow}}(y_2) 
H_{1,1,2}
\\
\nonumber
     &+
    \frac{1}{2} y_1 \sbrac{
       \mathsf G^{(1)}_{\mathrm{pow}}(y_2) 
       \brac{
       m+2 
       \mathsf G^{(1)}_{\mathrm{pow}} \brac{(y_1 y_2)^{-1}} 
       }
     }
     H_{1,1,1} .
   \end{align}
 Notice that $\tau^2_{(2)}(\mathsf J_{\mathrm{ II }})$ is already 
 written as a span of
 $\{H_{1,1,1}, H_{1,2,1} , H_{1,1,2}\} $. For 
 $\tau^2_{(2)}(\mathsf G_{\mathrm{ II }}) $, we need to     
replace $H_{1,2,2}$  and $H_{1,3,1}$ whose general form has been computed in 
 \cref{eq:H122-H113-coef-remove}
with coefficients  
\begin{align*}
    \eta_{1,3,1} = 
    2 y_1^2 (y_1 -1) \mathsf G^{(1)}_{\mathrm{pow}}(y_1^{-1}),
    \, \, \, \, 
    \eta_{1,2,2} = \frac{y_2}{2} \eta_{1,3,1}.
\end{align*}
Therefore 
we have $
y_1^{-\frac{m}{2}}  
     \tau^2_{(2)}(\mathsf G_{\mathrm{ II }}) =
     c_{1,2,1}  H_{1,2,1} + 
     c_{1,1,2}  H_{1,1,2} +
     c_{1,1,1}  H_{1,1,1} 
$ where the coefficients are computed  according to
\cref{eq:H122-H113-coef-remove} as below. First
\begin{align*}
    c_{1,2,1}  & =
    -(\frac{m}{2} +1)
y_1 (y_1 -1) \mathsf G^{(1)}_{\mathrm{pow}}(y_1^{-1})   
+
    \sbrac{
    \brac{
        \frac{m}{4 y_1} + \frac{1}{y_1 -1} - \frac{1}{2} \frac{1}{y_1 (y_2 -1)}
    } \eta_{1,1,3} 
    + \frac{\eta_{1,2,2}}{y_1(y_2-1)}
    } \\
               &=
               -(\frac{m}{2} +1)
y_1 (y_1 -1) \mathsf G^{(1)}_{\mathrm{pow}}(y_1^{-1})   
+ \brac{ \frac{m}{4} + \frac{1}{2}} y_1^{-1} \eta_{1,3,1}
+ \frac{ \eta_{1,3,1}}{y_1-1} 
\\
    & =
     2 y_1^2 \mathsf G^{(1)}_{\mathrm{pow}}(y_1^{-1}),
\end{align*}
where the first two terms in the middle line cancel out. 
The next one:
\begin{align*}
    c_{1,1,2}   &= 
        \frac{-1}{ y_1 (y_2 -1)} \eta_{1,2,2} 
        + \frac{y_2 (y_1 y_2 -1)}{2 y_1 (y_1-1) (y_2 -1)} \eta_{1,1,3}
    \\
                &    = 
         \brac{
      \frac{ y_1y_2 -1}{y_1-1} -1   
         }
         \frac{ y_2 \eta_{1,2,1} }{2 y_1 (y_2-1) }
    =  y_2 y_1^2 \mathsf G^{(1)}_{\mathrm{pow}}(y_1^{-1})
\end{align*}
The last one
\begin{align*}
    c_{1,1,1}   &= 
 \frac{-(2+m)}{4y_1(y_1-1)} \eta_{1,1,3} =
 -\frac{2+m}{2} y_1
 \mathsf G^{(1)}_{\mathrm{pow}}(y_1^{-1}).
\end{align*}
With $
 \mathsf G^{(1)}_{\mathrm{pow}}(y_1^{-1}) = 
 y_1^{\frac{1}{2}}
 \mathsf G^{(1)}_{\mathrm{pow}}(y_1)  
$, 
the total reads
\begin{align}
    \label{eq:tau-sqr-GII-as-Habc-inproof}
     y_1^{-\frac{m}{2}}  
     \tau^2_{(2)}(\mathsf G_{\mathrm{ II }}) (z_1,z_2;m)=
 y_1^{\frac{1}{2}}
 \mathsf G^{(1)}_{\mathrm{pow}}(y_1)  
     \brac{
    2 y_1^2  H_{1,2,1} + y_2 y_1^2 H_{1,1,2}
     - \frac{2+m}{2} y_1 H_{1,1,1}
     } .
\end{align}
We finish the computation by taking the difference of
\cref{eq:tau-sqr-JII-as-Habc-inproof,eq:tau-sqr-GII-as-Habc-inproof}.
  \end{proof}

    \begin{lem}
        \label{lem:msf-JI-to-Habc}
 The function $\tau_{(2)} \brac{ \mathsf J_{\mathrm{I}}}$ can also be written
 as a combination of $\{H_{1,2,1}, H_{1,1,2}, H_{1,1,1}\} $:
\begin{align*}
y_1^{-\frac{m}{2}} \tau_{(2)}^2 \brac{\mathsf J_{\mathrm{I}}} =
  \tilde c_{1,2,1} H_{1,2,1} +
  \tilde c_{1,1,2} H_{1,1,2} +
  \tilde c_{1,1,1} H_{1,1,1}, 
\end{align*}
where the coefficients $c_{1,2,1}, c_{1,1,2}$ and $c_{1,1,1}$  are given in
\cref{eq:til-c111-inproof,eq:til-c112-inproof,eq:til-c121-inproof} respectively.
    \end{lem}

\begin{proof}
 Notice that  $\mathsf J_{\mathrm{I}}$ in \cref{eq:sf-JI-defn} is of the form  
 $\mathsf J_{\mathrm{I}} = a_1 H_{2,1} + \frac{m}{4} a_2  \brac{
  H_{1,1} - \frac{4}{m} H_{2,1}
  }$ 
  where $H_{a,b} = H_{a,b}(z_2;m)$ and
  the coefficients are given by:  
  \begin{align*}
      a_1 = - y_1^{- \frac{1}{2}}  
      \mathsf G_{\mathrm{pow}}^{(1)} ( y_1 )
      \mathsf G_{\mathrm{pow}}^{(1)} ( y_2 )
      ,\, \, \, \, 
     a_2 = 2 \mathsf G_{\mathrm{pow}}^{(1,1)} (y_1, y_2 ).
  \end{align*}
    For the part containing $a_2$, we have:
    \begin{align*}
      y_1^{-\frac{m}{2}}
      \tau_{(2)}^2 \brac{
            - \frac{4}{m} H_{2,1} (z_2;m) + H_{1,1} (z_2;m)
        }    & = 
        -\frac{4}{m}   y_1 H_{1,2} (z_1;m)+
H_{1,1} (z_1;m)
\\
             &=
 -\frac{2}{m} y_1 H_{1,2} (z_1 ;m) + \frac{2}{m } H_{2,1} (z_1;m)
,
    \end{align*}
 where  the relation  $H_{1,1} = \frac{2}{m}\left((1-z)
 H_{1,2}+H_{2,1}\right)$ from \cite{Liu:2018ab} is applied.
    Now we have shown that 
  $
   y_1^{-\frac{m}{2}} 
      \tau_{(2)}^2 \brac{\mathsf J_{\mathrm{I}}} =
\eta_{1,2} H_{1,2} + \eta_{2,1} H_{2,1}
  $ with 
  $H_{a,b} \defeq H_{a,b}(z_1;m)$ and 
  \begin{align*}
      \eta_{1,2} & = y_1 \sbrac{
          \tau_{(2)}^2 (a_1) - \frac{1}{2} \tau_{(2)}^2 ( a_2 )
      } 
      = y_1 \sbrac{
        -  y_2^{-\frac{1}{2}} 
          \mathsf G_{\mathrm{pow}}^{(1)} (y_2) 
          \mathsf G_{\mathrm{pow}}^{(1)} ( (y_1y_2)^{-1} )
      - \mathsf G_{\mathrm{pow}}^{(1,1)} (y_2, (y_1y_2)^{-1})
      }
      ,
      \\
          \eta_{2,1} &= \frac{1}{2} \tau_{(2)}^2 (a_2 )
      = \mathsf G_{\mathrm{pow}}^{(1,1)} (y_2, (y_1y_2)^{-1})
      .
  \end{align*}
    Accordingly, $
  y_1^{-\frac{m}{2}} \tau_{(2)}^2 \brac{\mathsf J_{\mathrm{I}}} =
  \tilde c_{1,2,1} H_{1,2,1} +
  \tilde c_{1,1,2} H_{1,1,2} +
  \tilde c_{1,1,1} H_{1,1,1} 
  $ 
  where the coefficients are determined  by \cref{eq:H_21-H_12-z_1-as-Habc}:
  \begin{align}
      \label{eq:til-c121-inproof}
      \tilde c_{1,2,1} &= y_1 
      \sbrac{(y_2 -1) \eta_{1,2} -(y_1 y_2 -1)
      \eta_{2,1}
      } 
\\
\nonumber
&=
y_1^2 (y_2 -1)
\tau_{(2)}^2(a_1) -
 \frac{1}{2}
 \sbrac{
   y_1 (y_1y_2 -1)
   + y_1^2 (y_2 -1) 
}
\tau_{(2)}^2(a_2)
\\
&=
- y_1^2 (y_2 -1) y_2^{-\frac{1}{2}}
\mathsf G_{\mathrm{pow}}^{(1)} (y_2)
\mathsf G_{\mathrm{pow}}^{(1)} ((y_1y_2)^{-1}) -
\sbrac{
   y_1 (y_1y_2 -1) + y_1^2 (y_2 -1) 
}
 \mathsf G_{\mathrm{pow}}^{(1,1)} (y_2, (y_1y_2)^{-1} ) 
 \nonumber
  \end{align} 
  and
  \begin{align}
      \label{eq:til-c112-inproof}
      \tilde c_{1,1,2} &= 
      y_1 y_2 (1- y_1y_2) \eta_{2,1} =
      \frac{1}{2} y_1 y_2 (1- y_1y_2) \tau_{(2)}^2 (a_2) 
      \\
                       & =
      y_1 y_2 (1- y_1y_2) 
\mathsf G_{\mathrm{pow}}^{(1,1)} (y_2, (y_1y_2)^{-1}) 
      ,
\nonumber
  \end{align}
  and
  \begin{align}
      \label{eq:til-c111-inproof}
      \tilde c_{1,1,1} &= 
    \eta_{1,2} + \frac{1}{2} \brac{
        m(y_1y_2 -1) + 2 y_1y_2 
    } \eta_{2,1} 
    \\
                       & = 
     y_1  \tau_{(2)}^2( a_1 ) +
   \frac{1 }{2} y_1 (y_2-1)  \tau_{(2)}^2(a_2) +
   \frac{m}{4} (y_1y_2 -1) \tau_{(2)}^2 ( a_2 ) 
\nonumber
    \\
                       &=
 y_1 y_2^{-\frac{1}{2}} 
 \mathsf G_{\mathrm{pow}}^{(1)} (y_2)
 \mathsf G_{\mathrm{pow}}^{(1)} ( (y_1 y_2)^{-1} ) +
 \sbrac{
     y_1 (y_2 -1) +
  \frac{m}{2} (y_1y_2 -1)
 }
 \mathsf G_{\mathrm{pow}}^{(1,1)} (y_2, (y_1y_2)^{-1} ) 
 .
\nonumber
  \end{align}
\end{proof}


\bibliographystyle{halpha}

\bibliography{mylib}

\end{document}